\documentclass{amsart}
\usepackage[lite]{amsrefs}
\usepackage{amssymb}
\usepackage[all,cmtip]{xy}
\usepackage{amsthm}
\usepackage{amsmath}
\usepackage{amsfonts}
\usepackage[hidelinks]{hyperref}
\usepackage{enumerate}
\usepackage{physics}
\usepackage[titletoc]{appendix}
\usepackage{fancyhdr}
\usepackage{xcolor}
\usepackage{comment}
\usepackage{tikz-cd}
\usepackage{tikzpagenodes}
\usepackage{caption}
\usepackage{float}
\usetikzlibrary{patterns}

\usepackage[margin = 1in]{geometry}
\usepackage{setspace}
\allowdisplaybreaks
\numberwithin{equation}{section}
\newtheorem{definition}{Definition}[section]

\newtheorem{theorem}[definition]{Theorem}
\newtheorem{lemma}[definition]{Lemma}
\newtheorem{corollary}[definition]{Corollary}
\newtheorem{proposition}[definition]{Proposition}

\newtheorem{assumption}[definition]{Assumption}

\begin{document}
\title[Supersymmetric Semi--infinite Cohomology and $N = 2$ Superconformal Structures]{Supersymmetric Semi--infinite Cohomology and $N = 2$ Superconformal Structures}

\author[T. Creutzig]{Thomas Creutzig}
\address{Department Mathematik, Friedrich--Alexander Universit{\"a}t Erlangen--N{\"u}rnberg, Erlangen 91058, Germany}
\email{thomas.creutzig@fau.de}

\author[U. R. Suh]{Uhi Rinn Suh}
\address{Department of Mathematical Sciences and Research institute of Mathematics, Seoul National University, Seoul 08826, Korea}
\email{uhrisu1@snu.ac.kr}

\author[S. Yoon]{Sangwon Yoon}
\address{Department of Mathematical Sciences, Seoul National University, Seoul 08826, Korea}
\email{ysw317@snu.ac.kr}

\thanks{U. R. Suh is supported by National Research Foundation of Korea (NRF) Grant No. 2022R1C1C1008698}

\begin{abstract}
We study supersymmetric semi-infinite cohomology and its $N = 2$ superconformal structure via the BRST reduction of an $N = {\frac{3}{2}}$ superconformal structure on the super-affine vertex algebra associated to a super Manin pair arising from the Kazama--Suzuki $N = 2$ condition. Moreover, we realize super coset models in terms of relative supersymmetric semi-infinite cohomology.
\end{abstract}

\maketitle

\vskip 4mm
\section{Introduction}\label{section: Intro}
\setcounter{equation}{0}
\vskip 2mm

Vertex algebras are connected to physics in that they are a rigorous formulation of the symmetry algebra of a two-dimensional conformal field theory. The latter plays a prominent role as the sigma models of string theory. Physics usually studies supersymmetric string theories and in that instance the underlying conformal field theory and hence its vertex algebra have an $N=2$ superconformal structure. This structure is then tied to various interesting problems. The target of a super string theory is a Calabi--Yau threefold times Minkowski spacetime. The first Chern class of a Calabi--Yau $3$-fold vanishes and this ensures that the $3$-fold is Ricci flat \cite{Yau78}. It is also precisely the property that guarantees that the chiral de Rham complex of the Calabi--Yau $3$-fold has $N=2$ superconformal symmetry \cite{BHS08}. The $N=2$ superconformal structure then has itself important properties. For example its chiral rings, that is certain natural cohomologies, describe de Rham and Dolbeault cohomology of the Calabi--Yau and give itself rise to topological field theories \cites{Witten88, LVW89}. The notion of sigma model and its cohomology has recently be translated to vertex algebras by Moriwaki \cite{Moriwaki25}. It is surely a rewarding task to translate some more of these ideas to vertex algebras with $N=2$ superconformal symmetry. Hence it is important to construct such vertex algebras in the first place. In physics one of the major constructions of $N=2$ superconformal field theories are cosets like the Kazama--Suzuki cosets \cite{KS89N} and the Gepner models \cite{Gepner87}. The coset construction is often realized in terms of gauged WZW-models, where the gauging is performed via some BRST-cohomology (a relative semi-infinite Lie algebra cohomology). While it is usually not difficult to find vertex algebras with $N=1$ superconformal symmetry it is often much more subtle to see if this extends to an $N=2$ superconformal structure. The purpose of this work is to study relative semi-infinite Lie algebra cohomology and cosets on $N=1$ supersymmetric vertex algebras with the particular aim of understanding when the cohomology (coset) has an $N=2$ superconformal structure.

\vskip 1mm

We will now start by explaining relevant previous works and our precise main questions before stating the main results.

\vskip 2mm
\subsection{Historical background and related work}
\vskip 2mm

Note that the Sugawara construction provides an affine vertex algebra with a conformal vector, and the corresponding conformal vector can be found, for example, in \cite{KZ84}. Moreover, when ${\mathfrak{h}}$ is a Lie subalgebra of ${\mathfrak{g}}$, the affine coset model for ${\mathfrak{h}} \subset {\mathfrak{g}}$ and its conformal vector are given in \cite{GKO85}, where the conformal vector is given by the difference of the conformal vectors associated with ${\mathfrak{g}}$ and ${\mathfrak{h}}$. On the other hand, semi-infinite cohomology and relative semi-infinite cohomology were introduced in \cite{Feigin84} and studied in \cites{Feigin84, FGZ86}, and these works also contain explicit computations of relative semi-infinite cohomology. The coset model was also studied from the viewpoint of semi-infinite cohomology in \cite{HR93}, where, in particular, it was shown that the conformal vector of the coset model and that of the semi-infinite complex coincide in cohomology. Furthermore, a cohomological realization of the coset vertex algebra follows from the computations of relative semi-infinite cohomology in Theorem 1.12 of \cite{FGZ86}, Section 2.5 of \cite{CFL22}, Proposition 4 of \cite{CGNS22}, and Theorem 3.3 of \cite{CLNS25}.

\vskip 1mm

In the supersymmetric setting, the Kac--Todorov superconformal vector on a super-affine vertex algebra was introduced in \cite{KT85}. Super coset models associated with super-affine vertex algebras and their $N = 1$ superconformal vectors have also been studied (see \cites{KS89N, KS89C} and \cite{FS97B}). On the other hand, supersymmetric BRST cohomology can be found in \cite{MR94}, \cite{Kuramoto94}, and \cite{MRS21}, while \cite{Rhedin96} and \cite{FS97B} study supersymmetric semi-infinite cohomology for super coset models and show that the $N = 1$ superconformal vectors of the super coset model and of the supersymmetric semi-infinite complex coincide in cohomology.

\vskip 1mm

Moreover, the $N = 2$ superconformal structure on a super-affine vertex algebra was introduced in \cite{STVPS88}, can also be found in \cites{Mohammedi94B, Figueroa-O'Farrill96}, and was formulated in the language of Manin triples in \cite{Parkhomenko92}. For a super-affine vertex algebra associated with a Lie superalgebra admitting a super Manin triple, the $N = 2$ superconformal construction can be found in \cites{Creutzig09, CRS09}. For $N = 2$ superconformal structures on super coset models, the conditions under which the $N = 1$ symmetry of a super coset model can be extended to an $N = 2$ symmetry were studied in \cites{KS89N, KS89C}, and related constructions can also be found in \cite{HT91} and \cite{FS97A}, as well as in \cite{Parkhomenko14} in the language of Manin triples. In particular, \cite{FS97A} realizes the $N = 2$ superconformal structure of the super coset model within supersymmetric semi-infinite cohomology. Finally, the $N = {\frac{3}{2}}$ superconformal algebra was introduced in \cite{Kazama91}, while $N = {\frac{3}{2}}$ superconformal structures on super-affine vertex algebras can be found in \cites{IR93, IR94} and, in the language of Manin pairs, \cite{Getzler95}.

\vskip 2mm
\subsection{Main questions}
\vskip 2mm

This paper has three main objectives. First, we study a cohomological realization of a coset vertex algebra of super-affine vertex algebras. Second, we investigate the origin of the $N = 2$ superconformal structure on the super coset model. More precisely, we determine the symmetry satisfied by a lift of this $N = 2$ superconformal structure on the super coset model to the supersymmetric semi-infinite complex and, in particular, the corresponding symmetry on the tensor product of super-affine vertex algebras appearing as the coefficient module. Third, we extend all these constructions to the case in which the input data are Lie superalgebras.

\vskip 1mm

To place these problems in context, consider the following coset vertex algebra, called the {\textit{super coset model}}:
\begin{align*}
{\text{Com}}({V}^{{\kappa}_{\mathfrak{h}}}_{{\,\,} {\text{s}}}({\mathfrak{h}}), {V}^{{\kappa}_{\mathfrak{g}}}_{{\,\,} {\text{s}}}({\mathfrak{g}})).
\end{align*}
In the previous literature, the following points are particularly notable. First, Kazama and Suzuki studied the $N = 2$ superconformal structure on this super coset model in \cites{KS89N, KS89C}. In particular, the $N = 1$ superconformal vector of the super coset model is given in (2.37) of \cite{KS89N} (also in (15) of \cite{KS89C}), and together with the other odd vector of conformal weight ${\frac{3}{2}}$ given in (3.10) of \cite{KS89N} and (16) of \cite{KS89C}, these two vectors form an $N = 2$ superconformal structure on the super coset model, provided that the $N = 2$ conditions given in (3.14), (3.16), (3.17), (3.19), and (3.20) of \cite{KS89N} (see also (17)--(20) of \cite{KS89C}) are satisfied.

\vskip 1mm

On the other hand, the gauged supersymmetric WZNW model was shown to provide a BRST realization of the coset SCFT. At the level of chiral algebras, the coefficient module of the corresponding BRST complex is the tensor product of two super-affine vertex algebras, and the $N = 1$ superconformal vector of the super coset model coincides in cohomology with the $N = 1$ superconformal vector of the supersymmetric semi-infinite complex (see \cite{Rhedin96}, Section 3 of \cite{FS97B}, and Section 4 of \cite{FS97A}). Moreover, in Section 7 of \cite{FS97A}, the $N = 2$ superconformal structure of the super coset model was shown to define an $N = 2$ superconformal structure on the supersymmetric semi-infinite cohomology. More precisely, together with the above $N = 1$ superconformal structure, the other two vectors given in (51) and (52) of \cite{FS97A}, which complete the $N = 2$ superconformal structure of the super coset model, form an $N = 2$ superconformal structure of the supersymmetric semi-infinite cohomology. In particular, when ${\mathfrak{h}}$ is abelian, a lift to the complex of the $N = 2$ superconformal generators of the super coset model in cohomology forms an $N = 2$ superconformal structure itself.

\vskip 1mm

For super-affine vertex algebras themselves, $N = 2$ superconformal structures were studied in the language of Manin triples for Lie algebras in \cite{Parkhomenko92}, while the Lie superalgebra version appears in Proposition 6.4.1 of \cite{Creutzig09} and (25) of \cite{CRS09}. In particular, \cite{CRS09} states that the total $N = 2$ superconformal generators obtained by adding ghost sector contributions to the $N = 2$ superconformal generators of the supersymmetric WZNW model on the product supergroup represent the same cohomology classes as the corresponding $N = 2$ generators of the super coset model. Finally, note that \cite{Getzler95} constructed an $N = {\frac{3}{2}}$ superconformal structure on the super-affine vertex algebra associated with a Lie algebra admitting a Manin pair structure.

\vskip 1mm

Taking these developments into account, we resolve the above problems using the language of super Manin pairs. We first introduce relative supersymmetric semi-infinite cohomology and prove that the super coset model is realized precisely as this cohomology. We then explain the lift problem of the $N = 2$ superconformal structure of the super coset model in terms of an $N = {\frac{3}{2}}$ superconformal symmetry. More precisely, starting from the Kazama--Suzuki $N = 2$ condition, we construct a particular super Manin pair. The associated super-affine vertex algebra, before BRST reduction, then carries an $N = {\frac{3}{2}}$ superconformal structure. We combine these $N = {\frac{3}{2}}$ superconformal generators with specific $N = 2$ superconformal generators on the ghost sector so that the resulting total generators on the supersymmetric semi-infinite complex are BRST-invariant. At this point, the additional fields that distinguish the $N = {\frac{3}{2}}$ superconformal symmetry from the $N = 2$ superconformal symmetry are BRST-trivial in cohomology. Consequently, the total generators on the complex induce an $N = 2$ superconformal structure on the relative supersymmetric semi-infinite cohomology. As a corollary, we obtain an $N = 2$ superconformal structure on the super coset model associated with Lie superalgebras.

\vskip 2mm
\subsection{Cohomological realization of super coset model}
\vskip 2mm

We first introduce supersymmetric semi-infinite cohomology and relative supersymmetric semi-infinite cohomology (Definition \ref{def: SUSY semi-inf coh} and Definition \ref{def: rel SUSY semi-inf coh}), which are the main objects of this paper.

\vskip 1mm

Let ${\mathfrak{h}}$ be a finite-dimensional Lie superalgebra, and let $\{ {u}_{i} \}_{i \in H}$ be a parity-homogeneous basis of ${\mathfrak{h}}$. We denote the structure constants of ${\mathfrak{h}}$ with respect to the basis $\{ {u}_{i} \}_{i \in H}$ by
\begin{align*}
[{u}_{i}, {u}_{j}] = {\sum}_{k \in H} {h}^{k}_{ij}{u}_{k},
\end{align*}
for $i, j \in H$. Let ${\textstyle{\bigwedge}}^{{\frac{\infty}{2}} + {\bullet}}_{{\,\,} {\text{s}}}({\mathfrak{h}})$ be the vertex algebra consisting of ${\text{dim}}{\,}{\mathfrak{h}}$ copies of the ${b}{c}$--${\beta}{\gamma}$ system, whose only nontrivial singular OPEs among its generators are
\begin{align*}
{D}{\phi}_{i}(z){\phi}^{j}(w) \sim {\frac{{\delta}_{ij}}{(z - w)}},
{\quad}
{D}{\phi}^{i}(z){\phi}_{j}(w) \sim {\frac{{\delta}_{ij}}{(z - w)}},
\end{align*}
for $i, j \in H$, where the parity of each generator is given by $p({\phi}_{i}) = p({D}{\phi}^{i}) = p({u}_{i})$ and $p({\phi}^{i}) = p({D}{\phi}_{i}) = p({u}_{i}) + 1$ for $i \in H$. Then, for a ${V}^{0}_{{\,\,} {\text{s}}}({\mathfrak{h}})$-module ${M}$, the complex for the {\textit{supersymmetric semi-infinite cohomology}} of ${\widehat{\mathfrak{h}}}$ with coefficients in ${M}$ is given by
\begin{align*}
{C}^{{\frac{\infty}{2}} + {\bullet}}_{{\,\,} {\text{s}}}({\mathfrak{h}}, {M}) := {M} \otimes {\textstyle{\bigwedge}}^{{\frac{\infty}{2}} + {\bullet}}_{{\,\,} {\text{s}}}({\mathfrak{h}}),
\end{align*}
where the differential is defined as follows:
\begin{align*}
{d}_{\text{s}} 
=& {\,} {\sum}_{i \in H}{\sum}_{m \in {\mathbb{Z}}} (-1)^{i} {u}^{M}_{i {\,} (m)} \otimes {\phi}^{i}_{(-m-1)} 
- {\sum}_{i \in H}{\sum}_{m \in {\mathbb{Z}}} {\overline{u}}^{M}_{i {\,} (m)} \otimes {D}{\phi}^{i}_{(-m-1)}
\\
&- {\frac{1}{2}}{\,}{\sum}_{i, j, k \in H} (-1)^{jk + j}{h}^{i}_{jk} {\,}{\text{id}}_{M} \otimes ({D}({\phi}_{i}{\phi}^{j}{\phi}^{k}))_{(0)}.
\end{align*}
Moreover, the subcomplex given by
\begin{align*}
{C}^{{\frac{\infty}{2}} + {\bullet}}_{{\,\,} {\text{s, rel}}}({\mathfrak{h}}, {M}) := {\Big(}{\bigcap}_{i \in H}{\,}{\text{ker}}{\,}({J}_{{u}_{i}})_{(0)}{\Big)} \cap {\Big(}{\bigcap}_{i \in H}{\,}{\text{ker}}{\,}({D}{\phi}_{i})_{(0)}{\Big)},
\end{align*}
defines the {\textit{relative supersymmetric semi-infinite cohomology}} of ${\widehat{\mathfrak{h}}}$ with coefficients in ${M}$:
\begin{align*}
{H}^{{\frac{\infty}{2}} + {\bullet}}_{{\,\,} {\text{s, rel}}}({\mathfrak{h}}, {M}).
\end{align*}
Here, we use the following notation:
\begin{align*}
{J}_{{u}_{i}} := {u}_{i} \otimes {\ket{0}} + {\ket{0}} \otimes {J}_{\bigwedge}({u}_{i}) \in {V}^{0}_{{\,\,} {\text{s}}}({\mathfrak{h}}) \otimes {\textstyle{\bigwedge}}^{{\frac{\infty}{2}} + {\bullet}}_{{\,\,} {\text{s}}}({\mathfrak{h}}),
\end{align*}
for $i \in H$, where ${J}_{\bigwedge}({u}_{i})$ is defined in Lemma \ref{lem: sec SUSY semi inf cohom J wedge homo} by
\begin{align*}
{J}_{\bigwedge}({u}_{i}) = {\sum}_{j, k \in H} (-1)^{jk}{\big(}{h}^{k}_{ij}({D}{\phi}^{j}){\phi}_{k} + (-1)^{j+1} {h}^{k}_{ij}{\phi}^{j}{D}{\phi}_{k}{\big)}.
\end{align*}

\vskip 1mm

In particular, let ${\mathfrak{g}}$ be a finite-dimensional Lie superalgebra, and let ${\mathfrak{h}}$ be a Lie subalgebra of ${\mathfrak{g}}$. Suppose that ${\mathfrak{g}}$ is equipped with a nondegenerate consistent supersymmetric invariant bilinear form ${\kappa}_{\mathfrak{g}}$, and that its restriction to the Lie superalgebra ${\mathfrak{h}}$, denoted by ${\kappa}_{\mathfrak{h}}$, is also nondegenerate (Assumption \ref{assum: first assumption}). We then consider the Lie superalgebra ${\mathfrak{g}} \oplus {\mathfrak{h}}$ with the bilinear form ${\kappa}_{{\mathfrak{g}} \oplus {\mathfrak{h}}} := {\kappa}_{\mathfrak{g}} \oplus -{\kappa}_{\mathfrak{h}}$ in \eqref{eq: sec sMP bilinear form on g+h}, and regard the corresponding super-affine vertex algebra
\begin{align*}
{V}^{{\kappa}_{{\mathfrak{g}} \oplus {\mathfrak{h}}}}_{{\,\,} {\text{s}}}({\mathfrak{g}} \oplus {\mathfrak{h}}) \cong {V}^{{\kappa}_{\mathfrak{g}}}_{{\,\,} {\text{s}}}({\mathfrak{g}}) \otimes {V}^{-{\kappa}_{\mathfrak{h}}}_{{\,\,} {\text{s}}}({\mathfrak{h}})
\end{align*}
as an ${\widehat{\mathfrak{h}}}$-module via the diagonal action given in \eqref{eq: sec N=2 of SUSY semi inf diag homo}. Then we prove a cohomological realization of the super coset model:

\vskip 1mm

\begin{theorem}[Theorem \ref{theorem: sec cohom desc cohom and supercoset}]\label{theorem: Intro cohom realization of supercoset}\rm
Let ${\mathfrak{h}}$ be a direct sum of a reductive Lie algebra and Lie superalgebras of type ${\mathfrak{osp}}_{1|2m}$. Suppose that the restriction of ${\kappa}_{\mathfrak{h}} - {\frac{1}{2}}{\overline{\kappa}}_{\mathfrak{h}}$ to each nonabelian simple summand of ${\mathfrak{h}}$ is an irrational level. Then we have
\begin{align*}
{H}^{{\frac{\infty}{2}} + {n}}_{{\,\,} {\text{s, rel}}}({\mathfrak{h}}, {V}^{{\kappa}_{\mathfrak{g}}}_{{\,\,} {\text{s}}}({\mathfrak{g}}) \otimes {V}^{-{\kappa}_{\mathfrak{h}}}_{{\,\,} {\text{s}}}({\mathfrak{h}})) \cong {\delta}_{n,0}{\,}{\text{Com}}({V}^{{\kappa}_{\mathfrak{h}}}_{{\,\,} {\text{s}}}({\mathfrak{h}}), {V}^{{\kappa}_{\mathfrak{g}}}_{{\,\,} {\text{s}}}({\mathfrak{g}})),
\end{align*}
for all $n \in {\mathbb{Z}}$.
\end{theorem}

\vskip 1mm

The proof of Theorem \ref{theorem: Intro cohom realization of supercoset} proceeds as follows. First, from the supersymmetric semi-infinite complex with coefficients in
\begin{align*}
{V}^{{\kappa}_{\mathfrak{g}}}_{{\,\,} {\text{s}}}({\mathfrak{g}}) \otimes {V}^{-{\kappa}_{\mathfrak{h}}}_{{\,\,} {\text{s}}}({\mathfrak{h}}),
\end{align*}
we extract a subcomplex that is isomorphic to the semi-infinite complex with coefficients in
\begin{align*}
{V}^{{\kappa}_{\mathfrak{g}} - {\frac{1}{2}}{\overline{\kappa}}_{\mathfrak{g}}}({\mathfrak{g}}) 
\otimes 
{\mathcal{F}}({\overline{\mathfrak{h}}}^{\perp}) 
\otimes 
{V}^{- {\kappa}_{\mathfrak{h}} - {\frac{1}{2}}{\overline{\kappa}}_{\mathfrak{h}}}({\mathfrak{h}}),
\end{align*}
where ${\overline{\kappa}}_{\mathfrak{g}}$ and ${\overline{\kappa}}_{\mathfrak{h}}$ denote the Killing forms of ${\mathfrak{g}}$ and ${\mathfrak{h}}$, respectively (see Lemma \ref{lem: sec cohom desc VA iso of C3} and Lemma \ref{lem: sec cohom desc SUSY decomposition of SUSY complex into subcomplexes}). We then show that these two complexes are homotopy equivalent (Proposition \ref{prop: sec cohom desc homotopy equiv for non-rel}). Consequently, the supersymmetric semi-infinite cohomology can be computed in terms of the semi-infinite cohomology (Proposition \ref{prop: sec cohom desc SUSY H(M) to non-SUSY H(M0)}). We further prove that the corresponding relative complexes are homotopy equivalent, which reduces the relative supersymmetric semi-infinite cohomology to the corresponding relative semi-infinite cohomology (Proposition \ref{prop: sec cohom desc rel SUSY H(M) to rel non-SUSY H(M0)}). Finally, by computing this relative semi-infinite cohomology, we obtain a cohomological description of the super coset model (Theorem \ref{theorem: sec cohom desc cohom and supercoset}).

\vskip 2mm
\subsection{Induced $N = 2$ superconformal structure on supersymmetric semi-infinite cohomology}
\vskip 2mm

The next result concerns the superconformal structure on supersymmetric semi-infinite cohomology. First, we can consider the following equivalent definitions of $N = 2$ and $N = {\frac{3}{2}}$ superconformal structures on a vertex algebra (Lemma \ref{lem: equiv def of N=2 sconf} and Proposition \ref{prop: equiv def of N=3/2 sconf}).

\vskip 1mm

Let $V$ be a vertex algebra with an $N = 1$ superconformal vector $G$ of central charge $c \in {\mathbb{C}}$. Suppose that $U \in V$ is an even primary vector of conformal weight $1$ satisfying the following OPEs:
\begin{align*}
{G}(z)U(w) \sim {\frac{({G}_{(0)}{U})(w)}{(z - w)}},
{\quad}
{U}(z){U}(w) \sim {\frac{c}{3(z - w)^{2}}},
{\quad}
{U}(z)({G}_{(0)}{U})(w) \sim {\frac{({U}_{(0)}({G}_{(0)}{U}))(w)}{(z - w)}}.
\end{align*}
If we set ${\Phi} := - {\frac{1}{8}}({G} + {U}_{(0)}({G}_{(0)}{U}))$ and the following additional OPEs hold:
\begin{gather*}
{G}(z){\Phi}(w) \sim {\frac{({G}_{(0)}{\Phi})(w)}{(z - w)}},
{\quad}
({G}_{(0)}{U})(z){\Phi}(w) \sim {\frac{-({G}_{(0)}{\Phi})(w)}{(z - w)}},
\\
{U}(z){\Phi}(w) \sim {\frac{-3{\Phi}(w)}{(z - w)}},
{\quad}
{\Phi}(z){\Phi}(w) \sim 0,
\end{gather*}
then (${G}$, ${U}$) is called an {\textit{$N = {\frac{3}{2}}$ superconformal structure}} on $V$. Moreover, if the following relation
\begin{align*}
{U}_{(0)}({G}_{(0)}{U}) = - {G}
\end{align*}
holds, then (${G}$, ${U}$) is called an {\textit{$N = 2$ superconformal structure}} on $V$.

\vskip 1mm

Now, suppose that the following assumption, which we refer to as the {\textit{Kazama--Suzuki $N = 2$ condition}}, holds.

\vskip 1mm

\begin{assumption}[Assumption \ref{assum: second assumption}]\label{assump: Intro KS condition}\rm
Let ${\mathfrak{g}}$ be a finite-dimensional Lie superalgebra with a nondegenerate consistent supersymmetric invariant bilinear form ${\kappa}_{\mathfrak{g}}$. Suppose that there exists an injective Lie superalgebra homomorphism ${\iota} : {\mathfrak{h}} \rightarrow {\mathfrak{g}}$, and that the bilinear form ${\kappa}_{\mathfrak{h}} := {\kappa}_{\mathfrak{g}} \circ ({\iota} \times {\iota})$ is nondegenerate on ${\mathfrak{h}}$. Assume further that the orthogonal complement of ${\iota}({\mathfrak{h}})$ in ${\mathfrak{g}}$ with respect to ${\kappa}_{\mathfrak{g}}$ admits a vector superspace decomposition ${\iota}({\mathfrak{h}})^{\perp} = {\mathfrak{t}}^{+} \oplus {\mathfrak{t}}^{-}$ such that:
\begin{align*}
{\kappa}_{\mathfrak{g}}({\mathfrak{t}}^{\pm}, {\,} {\mathfrak{t}}^{\pm}) = 0,
{\quad}
[{\mathfrak{t}}^{\pm}, {\iota}({\mathfrak{h}}) \oplus {\mathfrak{t}}^{\pm}] \subseteq {\mathfrak{t}}^{\pm}.
\end{align*}
\end{assumption}

\vskip 1mm

Then we construct the following super Manin pair structure, arising from the Kazama--Suzuki $N = 2$ condition, on the Lie superalgebra ${\mathfrak{g}} \oplus {\mathfrak{h}}$ equipped with the bilinear form ${\kappa}_{{\mathfrak{g}} \oplus {\mathfrak{h}}}$:
\begin{align*}
{\mathfrak{g}} \oplus {\mathfrak{h}} = {\mathfrak{a}}^{+}_{{\mathfrak{g}} \oplus {\mathfrak{h}}} \oplus {\mathfrak{a}}^{-}_{{\mathfrak{g}} \oplus {\mathfrak{h}}},
\end{align*}
where the isotropic subspaces ${\mathfrak{a}}^{+}_{{\mathfrak{g}} \oplus {\mathfrak{h}}}$ and ${\mathfrak{a}}^{-}_{{\mathfrak{g}} \oplus {\mathfrak{h}}}$ are given by:
\begin{align*}
{\mathfrak{a}}^{+}_{{\mathfrak{g}} \oplus {\mathfrak{h}}} = {\mathfrak{t}}^{+} \oplus {\{} {\iota}(h) + h {\,} | {\,} {h} \in {\mathfrak{h}}{\}},
{\quad}
{\mathfrak{a}}^{-}_{{\mathfrak{g}} \oplus {\mathfrak{h}}} = {\mathfrak{t}}^{-} \oplus {\{} {\iota}(h) - h {\,} | {\,} {h} \in {\mathfrak{h}}{\}}.
\end{align*}

\vskip 1mm

Thus, we obtain the following theorem. The super-affine vertex algebra associated with the above super Manin pair arising from the Kazama--Suzuki $N = 2$ condition admits an $N = {\frac{3}{2}}$ superconformal structure, which induces an $N = 2$ superconformal structure on the supersymmetric semi-infinite cohomology via supersymmetric quantum Hamiltonian reduction.

\vskip 1mm

\begin{theorem}[Theorem \ref{theorem: N=3/2 of saff} and Theorem \ref{theorem: sec N=2 of SUSY semi inf N=2 on cohom}]\label{theorem: Intro N=2 on SUSY cohom}\rm
Under Assumption \ref{assump: Intro KS condition}, suppose that ${\mathfrak{h}}$ is unimodular. Then ${G}_{{\mathfrak{g}} \oplus {\mathfrak{h}}}$ and ${U}_{{\mathfrak{g}} \oplus {\mathfrak{h}}}$ given in \eqref{eq: Intro N=3/2 on saff} form an $N = {\frac{3}{2}}$ superconformal structure on the super-affine vertex algebra ${V}^{{\kappa}_{{\mathfrak{g}} \oplus {\mathfrak{h}}}}_{{\,\,} {\text{s}}}({\mathfrak{g}} \oplus {\mathfrak{h}})$. Moreover, combined with the $N = 2$ superconformal structure ${G}_{\wedge}$ and ${U}_{\wedge}$ on ${\textstyle{\bigwedge}}^{{\frac{\infty}{2}} + {\bullet}}_{{\,\,} {\text{s}}}({\mathfrak{h}})$ given in \eqref{eq: Intro N=2 on ghost}, the following vectors on the supersymmetric semi-infinite complex
\begin{align*}
{G}_{\mathcal{C}} :=& {\,} {G}_{{\mathfrak{g}} \oplus {\mathfrak{h}}} \otimes {\ket{0}} + {\ket{0}} \otimes {G}_{\wedge},
\\
{U}_{\mathcal{C}} :=& {\,} {U}_{{\mathfrak{g}} \oplus {\mathfrak{h}}} \otimes {\ket{0}} + {\ket{0}} \otimes {U}_{\wedge},
\end{align*}
induce an $N = 2$ superconformal structure on the supersymmetric semi-infinite cohomology
\begin{align*}
{H}^{{\frac{\infty}{2}} + {\bullet}}_{{\,\,} {\text{s}}}({\mathfrak{h}}, {V}^{{\kappa}_{{\mathfrak{g}} \oplus {\mathfrak{h}}}}_{{\,\,} {\text{s}}}({\mathfrak{g}} \oplus {\mathfrak{h}})),
\end{align*}
with central charge ${\frac{3}{2}}{\,}{\text{sdim}}{\,}{\mathfrak{g}} - {\frac{3}{2}}{\,}{\text{sdim}}{\,}{\mathfrak{h}} + {3}{\,}{\text{str}}({\mathcal{D}}_{{\mathfrak{g}} \oplus {\mathfrak{h}}})$.
\end{theorem}

\vskip 1mm

The strategy for the proof of Theorem \ref{theorem: Intro N=2 on SUSY cohom} is as follows. First, we construct the particular super Manin pair structure arising from the Kazama--Suzuki $N = 2$ condition introduced above, and in the corresponding super-affine vertex algebra, we introduce a $U(1)$-current ${U}_{{\mathfrak{g}} \oplus {\mathfrak{h}}}$ compatible with the Kac--Todorov superconformal vector ${G}_{{\mathfrak{g}} \oplus {\mathfrak{h}}}$ as follows:
\begin{align}
{G}_{{\mathfrak{g}} \oplus {\mathfrak{h}}} 
=& {\,} {\sum}_{i \in I} {\big(}(-1)^{i}{y}^{i}{\overline{y}}_{i} + {y}_{i}{\overline{y}}^{i}{\big)} + {\sum}_{i,j,k \in I} (-1)^{ki} {C}^{i}_{jk}{\overline{y}}_{i}({\overline{y}}^{j}{\overline{y}}^{k})
\notag
\\
&+ {\sum}_{i,j,k \in I} (-1)^{ij} {F}^{jk}_{i}{\overline{y}}^{i}({\overline{y}}_{j}{\overline{y}}_{k}) + {\frac{1}{3}}{\,}{\sum}_{i,j,k \in I} (-1)^{ij}{\mathring{F}}^{jk}_{i}{\overline{y}}_{i}({\overline{y}}_{j}{\overline{y}}_{k}),
\label{eq: Intro N=3/2 on saff}
\\
{U}_{{\mathfrak{g}} \oplus {\mathfrak{h}}} 
=& {\,} {\sum}_{i \in I} (-1)^{i}{\overline{y}}^{i}{\overline{y}}_{i} + {\sum}_{i,j \in I} (-1)^{i}{F}^{ij}_{i}{y}_{j} + {\sum}_{i,j \in I} (-1)^{i}{C}^{i}_{ji}{y}^{j}.
\notag
\end{align}
Here, $\{ {y}_{i} \}_{i \in I}$ and $\{ {y}^{i} \}_{i \in I}$ are parity-homogeneous dual bases of ${\mathfrak{a}}^{+}_{{\mathfrak{g}} \oplus {\mathfrak{h}}}$ and ${\mathfrak{a}}^{-}_{{\mathfrak{g}} \oplus {\mathfrak{h}}}$, respectively (see \eqref{eq: sec sMP structure constant notation for sMP} for notation of structure constants associated with super Manin pair structures). Then, by Theorem \ref{theorem: N=3/2 of saff}, the fields in \eqref{eq: Intro N=3/2 on saff} form an $N = {\frac{3}{2}}$ superconformal structure on the super-affine vertex algebra. Combining this $N = {\frac{3}{2}}$ superconformal structure with the shifted $N = 2$ superconformal structure on the ${b}{c}$--${\beta}{\gamma}$ system 
(Proposition \ref{prop: sec N=2 of SUSY semi inf N=2 shift of ghost} and \eqref{eq: sec N=2 of SUSY semi inf def of G and U in ghost}) given by
\begin{align}
{G}_{\wedge} 
= {\,} {\sum}_{i \in {H}} {\big(}(-1)^{i + 1}{\phi}_{i}{\partial}{\phi}^{i} + {D}{\phi}_{i}{D}{\phi}^{i}{\big)},
{\quad}
{U}_{\wedge} 
= {\,} - {\sum}_{i \in {H}} {\phi}_{i}{D}{\phi}^{i},
\label{eq: Intro N=2 on ghost}
\end{align}
we can verify that the resulting fields are BRST-invariant in the supersymmetric semi-infinite complex (Proposition \ref{prop: sec N=2 of SUSY semi inf G and U in kernel}), and that they form an $N = {\frac{3}{2}}$ superconformal structure on the complex by Theorem \ref{theorem: N=3/2 of saff} and Proposition \ref{prop: sec N=2 of SUSY semi inf N=2 shift of ghost}. Moreover, Lemma \ref{lem: sec sMP F and ring F vanish}, Lemma \ref{lem: sec sMP coeff E vanish}, Proposition \ref{prop: sec N=3/2 computation of UGU}, and Proposition \ref{prop: sec N=2 of SUSY semi inf N=2 shift of ghost} yield the identity
\begin{align*}
{U}_{{\mathcal{C}} (0)}({G}_{{\mathcal{C}} (0)}{U}_{\mathcal{C}}) = -{G}_{\mathcal{C}} - {8}{\Phi}_{\mathcal{C}},
\end{align*}
where ${\Phi}_{\mathcal{C}}$ is given by
\begin{align*}
{\Phi}_{\mathcal{C}} = - {\frac{1}{6}}{\,}{\sum}_{i, j, k \in H} (-1)^{ij}{\mathring{F}}^{jk}_{i}{\overline{y}}_{i}({\overline{y}}_{j}{\overline{y}}_{k}) \otimes {\ket{0}}.
\end{align*}
Finally, to show that the above $N = {\frac{3}{2}}$ superconformal structure on the complex induces an $N = 2$ superconformal structure on the cohomology, it remains to deal with the term ${\Phi}_{\mathcal{C}}$. That is, the proof is completed by finding a preimage of ${\Phi}_{\mathcal{C}}$ under the BRST differential (Lemma \ref{lem: sec N=2 of SUSY semi inf BRST triv of Phi}).

\vskip 1mm

In particular, Theorem \ref{theorem: Intro N=2 on SUSY cohom} also holds for the relative supersymmetric semi-infinite cohomology (Corollary \ref{corollary: sec N=2 of SUSY semi inf N=2 on rel SUSY cohom}), and together with Theorem \ref{theorem: Intro cohom realization of supercoset}, this shows that we obtain a Lie superalgebra version of the Kazama--Suzuki $N = 2$ construction and, at the same time, that the $N = 2$ superconformal structure on the super coset model arises from the $N = {\frac{3}{2}}$ superconformal structure on the super-affine vertex algebra given in \eqref{eq: Intro N=3/2 on saff}.

\vskip 1mm

\begin{corollary}[Corollary \ref{corollary: sec cohom desc lift of N=2 in scoset to swzw}]\rm
Under Assumption \ref{assump: Intro KS condition}, suppose that ${\mathfrak{h}}$ is a direct sum of a reductive Lie algebra and Lie superalgebras of type ${\mathfrak{osp}}_{1|2m}$, and that the restriction of ${\kappa}_{\mathfrak{h}} - {\frac{1}{2}}{\overline{\kappa}}_{\mathfrak{h}}$ to each nonabelian simple summand of ${\mathfrak{h}}$ is an irrational level. Then the $N = {\frac{3}{2}}$ superconformal structure given in \eqref{eq: Intro N=3/2 on saff} on the product of super-affine vertex algebras
\begin{align*}
{V}^{{\kappa}_{\mathfrak{g}}}_{{\,\,} {\text{s}}}({\mathfrak{g}}) \otimes {V}^{-{\kappa}_{\mathfrak{h}}}_{{\,\,} {\text{s}}}({\mathfrak{h}})
\end{align*}
induces an $N = 2$ superconformal structure of ${\text{Com}}({V}^{{\kappa}_{\mathfrak{h}}}_{{\,\,} {\text{s}}}({\mathfrak{h}}), {V}^{{\kappa}_{\mathfrak{g}}}_{{\,\,} {\text{s}}}({\mathfrak{g}}))$ via the relative supersymmetric semi-infinite cohomology.
\end{corollary}

\vskip 2mm
\subsection{Organization of the paper}
\vskip 2mm

In Section \ref{section: sMP}, we study the basic properties of super Manin pairs and a particular class of super Manin pairs arising from the Kazama--Suzuki $N = 2$ condition. In Section \ref{section: N=2 and N=3/2 sconf VA}, we recall the definitions of $N = 2$ and $N = {\frac{3}{2}}$ superconformal structures on vertex algebras and present their equivalent formulations. In Section \ref{section: N=3/2 sconf of super affine VA}, we construct an $N = {\frac{3}{2}}$ superconformal structure on the super-affine vertex algebra associated with a super Manin pair structure arising from the Kazama--Suzuki $N = 2$ condition. In Section \ref{section: SUSY semi inf cohom}, we consider supersymmetric semi-infinite cohomology and introduce its relative version. In Section \ref{section: N=2 sconf of SUSY semi inf cohom}, we show that the $N = {\frac{3}{2}}$ superconformal structure constructed in Section \ref{section: N=3/2 sconf of super affine VA} induces an $N = 2$ superconformal structure on the supersymmetric semi-infinite cohomology. In Section \ref{section: cohom description of super coset model}, we prove that the super coset model can be realized as the relative supersymmetric semi-infinite cohomology.

\vskip 4mm
\section{Super Manin pairs}\label{section: sMP}
\vskip 2mm

In the first part of this section, we will study the basic properties of super Manin pair structures on Lie superalgebras, following the Lie algebra case in Section 2 of \cite{Getzler95}, and we refer to Section 6.3 of \cite{Creutzig09} for super Manin triples. In the second part, we will study a particular class of super Manin pairs arising from the Kazama--Suzuki $N = 2$ condition. For the basic notions of Lie superalgebras, we refer to \cite{Kac77}. Throughout this paper, Lie superalgebras and vertex algebras are defined over the base field ${\mathbb{C}}$. The parity of an element ${a}$ is denoted by $p(a)$, and we write $(-1)^{a}$ for $(-1)^{p(a)}$.

\vskip 1mm

\begin{definition}\label{def: sMP}\rm
Let ${\mathfrak{g}}$ be a finite-dimensional Lie superalgebra over ${\mathbb{C}}$ with a nondegenerate consistent supersymmetric invariant bilinear form ${\kappa}_{\mathfrak{g}}$. Suppose that there is a decomposition ${\mathfrak{g}} = {\mathfrak{a}}^{+} {\oplus} {\,} {\mathfrak{a}}^{-}$ as a vector superspace, where ${\mathfrak{a}}^{+}$ and ${\mathfrak{a}}^{-}$ are isotropic subspaces of ${\mathfrak{g}}$ with respect to ${\kappa}_{\mathfrak{g}}$, i.e., ${\kappa}_{\mathfrak{g}}({\mathfrak{a}}^{\pm}, {\,} {\mathfrak{a}}^{\pm}) = 0$, and ${\mathfrak{a}}^{+}$ is a Lie subalgebra of ${\mathfrak{g}}$. Then $({\mathfrak{g}}, {\mathfrak{a}}^{+}, {\mathfrak{a}}^{-})$ is called a {\textit{super Manin pair}} and we also simply say that ${\mathfrak{g}}$ admits a super Manin pair structure. Moreover, if ${\mathfrak{a}}^{-}$ is also a Lie subalgebra of ${\mathfrak{g}}$, then $({\mathfrak{g}}, {\mathfrak{a}}^{+}, {\mathfrak{a}}^{-})$ is called a {\textit{super Manin triple}}. We can also consider the notion of a {\textit{super Manin subpair}}.
\end{definition}

\vskip 1mm

For a super Manin pair $({\mathfrak{g}}, {\mathfrak{a}}^{+}, {\mathfrak{a}}^{-})$, let $\{ {y}_{i} \}_{i \in I}$ be a parity-homogeneous basis of ${\mathfrak{a}}^{+}$ and let $\{ {y}^{i} \}_{i \in I}$ be the basis of ${\mathfrak{a}}^{-}$ satisfying ${\kappa}_{\mathfrak{g}}({y}_{i}, {y}^{j}) = {\delta}_{ij}$ for $i,j \in I$. For the structure constants of ${\mathfrak{g}}$, we use the following notations:
\begin{align}
[{y}_{i}, {y}_{j}] = {\sum}_{k \in I} {C}^{k}_{ij}{y}_{k},
{\quad}
[{y}^{i}, {y}^{j}] = {\sum}_{k \in I} {F}^{ij}_{k}{y}^{k} + {\sum}_{k \in I} {\mathring{F}}^{ij}_{k}{y}_{k},
\label{eq: sec sMP structure constant notation for sMP}
\end{align}
for $i,j \in I$. Since ${\kappa}_{\mathfrak{g}}$ is a supersymmetric invariant bilinear form, the Lie bracket between ${\mathfrak{a}}^{+}$ and ${\mathfrak{a}}^{-}$ can be written as
\begin{align*}
[{y}_{i}, {y}^{j}] = {\sum}_{k \in I} {F}^{jk}_{i}{y}_{k} + {\sum}_{k \in I}{C}^{j}_{ki}{y}^{k},
\end{align*}
for $i,j \in I$. Moreover, the structure constants ${\mathring{F}}$ satisfy
\begin{align}
{\mathring{F}}^{ij}_{k} = (-1)^{i}{\mathring{F}}^{jk}_{i} = (-1)^{i+j}{\mathring{F}}^{ki}_{j},
\label{eq: cyclic str const MP}
\end{align}
which follows from the equality:
\begin{align*}
{\mathring{F}}^{ij}_{k} 
&= {\kappa}_{\mathfrak{g}}{\big(}{\sum}_{r \in I} {F}^{ij}_{r}{y}^{r} + {\sum}_{r \in I} {\mathring{F}}^{ij}_{r}{y}_{r}, {\,} {y}^{k}{\big)} 
= {\kappa}_{\mathfrak{g}}([{y}^{i}, {y}^{j}], {y}^{k}) 
\\
&= {\kappa}_{\mathfrak{g}}({y}^{i}, [{y}^{j}, {y}^{k}]) 
= {\kappa}_{\mathfrak{g}}{\big(}{y}^{i}, {\,} {\sum}_{r \in I} {F}^{jk}_{r}{y}^{r} + {\sum}_{r \in I} {\mathring{F}}^{jk}_{r}{y}_{r}{\big)} 
= (-1)^{i}{\mathring{F}}^{jk}_{i},
\end{align*}
for $i,j,k \in I$. The Jacobi identity for the Lie subalgebra ${\mathfrak{a}}^{+}$ in terms of structure constants is given by
\begin{align}
{\sum}_{l \in I} {\big(}(-1)^{ki}{C}^{m}_{il}{C}^{l}_{jk} + (-1)^{ij}{C}^{m}_{jl}{C}^{l}_{ki} + (-1)^{jk}{C}^{m}_{kl}{C}^{l}_{ij}{\big)} = 0,
\label{eq: JI 1 for MP}
\end{align}
for $i,j,k,m \in I$. Similarly, the Jacobi identity for the elements ${y}^{i}, {y}^{j}, {y}^{k}$ in the subspace ${\mathfrak{a}}^{-}$ corresponds to the following two equalities:
\begin{gather}
\begin{split}
{\sum}_{l \in I} \left(
\begin{aligned}
&(-1)^{ki}{F}^{jk}_{l}{F}^{il}_{m} + (-1)^{ij}{F}^{ki}_{l}{F}^{jl}_{m} + (-1)^{jk}{F}^{ij}_{l}{F}^{kl}_{m} \\
&+ (-1)^{ki + j + k}{C}^{i}_{lm}{\mathring{F}}^{jk}_{l} + (-1)^{ij + k + i}{C}^{j}_{lm}{\mathring{F}}^{ki}_{l} + (-1)^{jk + i + j}{C}^{k}_{lm}{\mathring{F}}^{ij}_{l}
\end{aligned}
\right) = 0,
\label{eq: JI B for MP}
\end{split}
\end{gather}

\begin{gather*}
\begin{split}
{\sum}_{l \in I} \left(
\begin{aligned}
&(-1)^{ki}{F}^{jk}_{l}{\mathring{F}}^{il}_{m} + (-1)^{ij}{F}^{ki}_{l}{\mathring{F}}^{jl}_{m} + (-1)^{jk}{F}^{ij}_{l}{\mathring{F}}^{kl}_{m} \\ 
&+ (-1)^{ki + i}{F}^{mi}_{l}{\mathring{F}}^{jk}_{l} + (-1)^{ij + j}{F}^{mj}_{l}{\mathring{F}}^{ki}_{l} + (-1)^{jk + k}{F}^{mk}_{l}{\mathring{F}}^{ij}_{l}
\end{aligned}
\right) = 0,
\end{split}
\end{gather*}
for $i,j,k,m \in I$. Now consider the Jacobi identity in the case where two elements lie in ${\mathfrak{a}}^{+}$ and the remaining one lies in ${\mathfrak{a}}^{-}$. In this case, the Jacobi identity for the elements ${y}_{i}, {y}_{j}, {y}^{k}$ yields equation \eqref{eq: JI 1 for MP} together with the following additional relation:
\begin{align}
{\sum}_{l \in I} {\big(}{C}^{m}_{il}{F}^{kl}_{j} + (-1)^{ij + 1}{C}^{k}_{li}{F}^{lm}_{j} + (-1)^{ij + 1}{C}^{m}_{jl}{F}^{kl}_{i} + {C}^{k}_{lj}{F}^{lm}_{i} - {C}^{l}_{ij}{F}^{km}_{l}{\big)} = 0,
\label{eq: JI 3 for MP}
\end{align}
for $i,j,k,m \in I$. Moreover, applying the Jacobi identity to the elements ${y}_{i}, {y}^{j}, {y}^{k}$, we obtain two relations, one of which coincides with the equation \eqref{eq: JI 3 for MP}, while the other can be shown to coincide with the equation \eqref{eq: JI B for MP} after using \eqref{eq: cyclic str const MP}. In particular, we have
\begin{align*}
{\sum}_{i,l \in I} (-1)^{ki}{C}^{i}_{il}{C}^{l}_{jk} = 0,
\end{align*}
for any $j,k \in I$, by taking $m = i$ in \eqref{eq: JI 1 for MP} and considering the sum over all $i$ in $I$. If we multiply \eqref{eq: JI 1 for MP} by $(-1)^{ik + i}$ and proceed in the same way, we obtain another identity, namely
\begin{align}
{\sum}_{i,l \in I} (-1)^{i}{C}^{i}_{il}{C}^{l}_{jk} = 0,
\label{eq: JI 1 modified for MP}
\end{align}
for any $j,k \in I$. Similarly, we have
\begin{gather}
{\sum}_{i,l \in I} {\big(}(-1)^{i}{C}^{i}_{li}{F}^{lk}_{j} + (-1)^{i}{C}^{k}_{lj}{F}^{li}_{i} + (-1)^{il + 1}{C}^{k}_{li}{F}^{li}_{j}{\big)} = 0,
\label{eq: JI 3 modified for MP}
\\
{\sum}_{i,l \in I} {\big(}(-1)^{i}{F}^{jk}_{l}{F}^{il}_{i} + (-1)^{i}{C}^{i}_{li}{\mathring{F}}^{jk}_{l} + (-1)^{j}{C}^{j}_{li}{\mathring{F}}^{ki}_{l} + (-1)^{ij}{C}^{k}_{li}{\mathring{F}}^{ij}_{l}{\big)} = 0,
\label{eq: JI B modified for MP}
\end{gather}
for $j,k \in I$, from \eqref{eq: JI 3 for MP} multiplied by $(-1)^{i}$ and from \eqref{eq: JI B for MP} multiplied by $(-1)^{ik + i}$, respectively.

\vskip 1mm

We now consider an even element
\begin{align}
{\rho} := {\sum}_{i \in I} (-1)^{i}[{y}_{i}, {y}^{i}] = {\sum}_{i,j \in I} (-1)^{i}{F}^{ij}_{i}{y}_{j} + {\sum}_{i,j \in I} (-1)^{i}{C}^{i}_{ji}{y}^{j},
\label{eq: def of the element rho}
\end{align}
and the operator ${\mathcal{D}} := -{\pi}_{{\mathfrak{a}}^{+}} \circ {\text{ad}}^{+}_{\rho}$ on ${\mathfrak{a}}^{+}$, where ${\pi}_{{\mathfrak{a}}^{+}} : {\mathfrak{g}} \rightarrow {\mathfrak{a}}^{+}$ is the projection map and ${\text{ad}}^{+}_{\rho} : {\mathfrak{a}}^{+} \rightarrow {\mathfrak{g}}$ is the map defined by ${\text{ad}}^{+}_{\rho}(x) := [{\rho}, x]$ for $x \in {\mathfrak{a}}^{+}$. Then, from the equality \eqref{eq: JI 3 modified for MP}, we have:
\begin{gather}
{\mathcal{D}}({y}_{i}) = {\sum}_{j \in I} {\mathcal{D}}^{j}_{i}{y}_{j},
{\quad}
{\kappa}_{\mathfrak{g}}({\rho}, {\rho}) = -{\sum}_{i \in I} (-1)^{i}{\mathcal{D}}^{i}_{i} = - {\text{str}}({\mathcal{D}}),
\label{eq: def of map D and D^ij}
\end{gather}
where ${\mathcal{D}}^{j}_{i} = {\sum}_{m,n \in I} (-1)^{mn}{C}^{j}_{mn}{F}^{mn}_{i}$ for $i,j \in I$ (see equations (15) and (16) of \cite{CRS09} and equation (6.3.12) of \cite{Creutzig09}. For the Lie algebra version, see Lemma 2.2 of \cite{Getzler95}).

\vskip 1mm

For simplicity of notation in Section \ref{section: N=3/2 sconf of super affine VA}, we define the map ${\mathcal{E}} := {\pi}_{{\mathfrak{a}}^{+}} \circ {\text{ad}}^{-}_{\rho}$ from ${\mathfrak{a}}^{-}$ to ${\mathfrak{a}}^{+}$, where the map ${\text{ad}}^{-}_{\rho} : {\mathfrak{a}}^{-} \rightarrow {\mathfrak{g}}$ is defined by ${\text{ad}}^{-}_{\rho}(x) := [{\rho}, x]$ for $x \in {\mathfrak{a}}^{-}$. Then, from direct computations, we have:
\begin{align*}
{\text{ad}}^{-}_{\rho}({y}^{i}) 
=& {\sum}_{n \in I} (-1)^{n}[[{y}_{n}, {y}^{n}], {y}^{i}] 
= {\sum}_{m,n \in I} (-1)^{n}[{F}^{nm}_{n}{y}_{m} + {C}^{n}_{mn}{y}^{m}, {y}^{i}] 
\\
=& {\sum}_{m,n \in I} (-1)^{n}{F}^{nm}_{n} [{y}_{m}, {y}^{i}] 
+ {\sum}_{m,n \in I} (-1)^{n}{C}^{n}_{mn} [{y}^{m}, {y}^{i}] 
\\
=& {\sum}_{m,n \in I} (-1)^{n}{F}^{nm}_{n} {\big(}{\sum}_{j \in I}{F}^{ij}_{m}{y}_{j} + {\sum}_{j \in I}{C}^{i}_{jm}{y}^{j}{\big)} 
\\
&+ {\sum}_{m,n \in I} (-1)^{n}{C}^{n}_{mn} {\big(}{\sum}_{j \in I}{F}^{mi}_{j}{y}^{j} + {\sum}_{j \in I}{\mathring{F}}^{mi}_{j}{y}_{j}{\big)},
\end{align*}
hence we obtain:
\begin{align}
{\mathcal{E}}({y}^{i}) = {\sum}_{j \in I} {\big(}{\sum}_{m,n \in I} (-1)^{n}{\big(}{F}^{nm}_{n}{F}^{ij}_{m} + {C}^{n}_{mn}{\mathring{F}}^{mi}_{j}{\big)}{\big)}{y}_{j}.
\label{eq: map E for MP}
\end{align}
If we denote the coefficients in the above equality \eqref{eq: map E for MP} by 
\begin{align}
{\mathcal{E}}^{ij} := {\sum}_{m,n \in I} (-1)^{n}{\big(}{F}^{nm}_{n}{F}^{ij}_{m} + {C}^{n}_{mn}{\mathring{F}}^{ij}_{m}{\big)},
\label{eq: def of E^ij}
\end{align}
then we know that ${\mathcal{E}}^{ij} = (-1)^{ij + 1}{\mathcal{E}}^{ji}$ for $i,j \in I$. Moreover, from \eqref{eq: JI B modified for MP}, we obtain the following equality:
\begin{align}
{\sum}_{m,n \in I} (-1)^{n}{\big(}{F}^{nm}_{n}{F}^{ij}_{m} + {C}^{n}_{mn}{\mathring{F}}^{ij}_{m}{\big)} = {\sum}_{m,n \in I} (-1)^{mn}{\big(}{C}^{i}_{mn}{\mathring{F}}^{mn}_{j} - (-1)^{m + n}{C}^{j}_{mn}{\mathring{F}}^{mn}_{i}{\big)},
\label{eq: sec sMP equality for map E}
\end{align}
which will also be used in Section \ref{section: N=3/2 sconf of super affine VA}.

\vskip 1mm

In the remainder of Section \ref{section: sMP}, we study a particular class of super Manin pairs arising from the condition in the theorem in Section 2 of \cite{Schweigert92} as well as the condition stated between equations (50) and (51) in Section 7.2 of \cite{FS97A}. This condition originates from the Kazama--Suzuki $N = 2$ condition (see also equations (3.14)--(3.20) of \cite{KS89N} and equations (17)--(20) of \cite{KS89C}). In this setting, we develop the required ingredients using the language of super Manin pairs. To this end, we begin with the following setting.

\vskip 1mm

\begin{assumption}\label{assum: first assumption}\rm
Let ${\mathfrak{g}}$ be a finite-dimensional Lie superalgebra with a nondegenerate consistent supersymmetric invariant bilinear form ${\kappa}_{\mathfrak{g}}$. Assume that there exists an injective Lie superalgebra homomorphism ${\iota} : {\mathfrak{h}} \rightarrow {\mathfrak{g}}$, and that the bilinear form ${\kappa}_{\mathfrak{h}} := {\kappa}_{\mathfrak{g}} \circ ({\iota} \times {\iota})$ is nondegenerate on ${\mathfrak{h}}$.
\end{assumption}

\vskip 1mm

Under Assumption \ref{assum: first assumption}, consider the Lie superalgebra ${\mathfrak{g}} \oplus {\mathfrak{h}}$ (i.e., a direct sum as Lie superalgebras) with a nondegenerate consistent supersymmetric invariant bilinear form ${\kappa}_{{\mathfrak{g}} \oplus {\mathfrak{h}}} := {\kappa}_{\mathfrak{g}} \oplus -{\kappa}_{\mathfrak{h}}$ given by
\begin{align}
{\kappa}_{{\mathfrak{g}} \oplus {\mathfrak{h}}}({g}_{1} + {h}_{1}, {g}_{2} + {h}_{2}) = {\kappa}_{\mathfrak{g}}({g}_{1}, {g}_{2}) - {\kappa}_{\mathfrak{g}}({\iota}({h}_{1}), {\iota}({h}_{2})),
\label{eq: sec sMP bilinear form on g+h}
\end{align}
for ${g}_{1}, {g}_{2} \in {\mathfrak{g}}$, ${h}_{1}, {h}_{2} \in {\mathfrak{h}}$. This Lie superalgebra will be our main object of interest. First, let $\{ {u}_{w} \}_{w \in H}$ be a parity-homogeneous basis of ${\mathfrak{h}}$, and let $\{ {u}^{w} \}_{w \in H}$ be the dual basis with respect to ${\kappa}_{\mathfrak{h}}$, satisfying ${\kappa}_{\mathfrak{h}}({u}_{w}, {u}^{z}) = {\delta}_{wz}$ for $w, z \in H$. Consider the elements
\begin{align}
{y}_{w} := {\iota}({u}_{w}) + {u}_{w},
{\quad}
{y}^{w} := {\frac{1}{2}}{\big(}{\iota}({u}^{w}) - {u}^{w}{\big)},
\label{eq: sec sMP dual bases y for h}
\end{align}
of the Lie superalgebra direct sum ${\iota}({\mathfrak{h}}) \oplus {\mathfrak{h}}$, and set:
\begin{align*}
{\mathfrak{a}}^{+}_{{\iota}({\mathfrak{h}}) \oplus {\mathfrak{h}}} 
:= {\text{span}}_{\mathbb{C}}\{ {y}_{w} \}_{w \in H},
{\quad}
{\mathfrak{a}}^{-}_{{\iota}({\mathfrak{h}}) \oplus {\mathfrak{h}}} 
:= {\text{span}}_{\mathbb{C}}\{ {y}^{w} \}_{w \in H}.
\end{align*}
Then we have:
\begin{align}
{\iota}({\mathfrak{h}}) \oplus {\mathfrak{h}} = {\mathfrak{a}}^{+}_{{\iota}({\mathfrak{h}}) \oplus {\mathfrak{h}}} \oplus {\mathfrak{a}}^{-}_{{\iota}({\mathfrak{h}}) \oplus {\mathfrak{h}}},
\label{eq: sec sMP decomposition of example sMP1}
\end{align}
as vector superspaces, and ${\mathfrak{a}}^{+}_{{\iota}({\mathfrak{h}}) \oplus {\mathfrak{h}}}$ is a Lie subalgebra of ${\iota}({\mathfrak{h}}) \oplus {\mathfrak{h}}$. Moreover, with respect to the nondegenerate consistent supersymmetric invariant bilinear form ${\kappa}_{{\iota}({\mathfrak{h}}) \oplus {\mathfrak{h}}} := {\kappa}_{\mathfrak{g}}|_{{\iota}({\mathfrak{h}}) \times {\iota}({\mathfrak{h}})} \oplus -{\kappa}_{\mathfrak{h}}$ on ${\iota}({\mathfrak{h}}) \oplus {\mathfrak{h}}$, the subspaces ${\mathfrak{a}}^{\pm}_{{\iota}({\mathfrak{h}}) \oplus {\mathfrak{h}}}$ are isotropic, and hence $({\iota}({\mathfrak{h}}) \oplus {\mathfrak{h}}, {\mathfrak{a}}^{+}_{{\iota}({\mathfrak{h}}) \oplus {\mathfrak{h}}}, {\mathfrak{a}}^{-}_{{\iota}({\mathfrak{h}}) \oplus {\mathfrak{h}}})$ is a super Manin pair (see also Section 2 of \cite{Getzler95}). Note that the bases $\{ {y}_{w} \}_{w \in H}$ of ${\mathfrak{a}}^{+}_{{\iota}({\mathfrak{h}}) \oplus {\mathfrak{h}}}$ and $\{ {y}^{w} \}_{w \in H}$ of ${\mathfrak{a}}^{-}_{{\iota}({\mathfrak{h}}) \oplus {\mathfrak{h}}}$ satisfy ${\kappa}_{{\iota}({\mathfrak{h}}) \oplus {\mathfrak{h}}}({y}_{w}, {y}^{z}) = {\delta}_{wz}$ for $w, z \in H$.

\vskip 1mm

Furthermore, since the bilinear form ${\kappa}_{\mathfrak{h}}$ is nondegenerate on ${\mathfrak{h}}$, we have:
\begin{align}
{\mathfrak{g}} = {\iota}({\mathfrak{h}}) \oplus {\iota}({\mathfrak{h}})^{\perp},
\label{eq: sec sMP orthogonal decomposition of g}
\end{align}
as vector superspaces, where ${\iota}({\mathfrak{h}})^{\perp}$ denotes the orthogonal complement of ${\iota}({\mathfrak{h}})$ in ${\mathfrak{g}}$ with respect to ${\kappa}_{\mathfrak{g}}$. Therefore, the restriction of ${\kappa}_{\mathfrak{g}}$ to ${\iota}({\mathfrak{h}})^{\perp}$ is also nondegenerate. Now, in addition to Assumption \ref{assum: first assumption}, suppose that the subspace ${\iota}({\mathfrak{h}})^{\perp}$ admits a decomposition
\begin{align}
{\iota}({\mathfrak{h}})^{\perp} = {\mathfrak{t}}^{+} \oplus {\mathfrak{t}}^{-},
\label{eq: sec sMP t decomposition assump of orthogonal complem}
\end{align}
as a vector superspace, satisfying the following conditions:
\begin{align}
{\kappa}_{\mathfrak{g}}({\mathfrak{t}}^{\pm}, {\,} {\mathfrak{t}}^{\pm}) = 0,
{\quad}
[{\mathfrak{t}}^{+}, {\iota}({\mathfrak{h}}) \oplus {\mathfrak{t}}^{+}] \subseteq {\mathfrak{t}}^{+}.
\label{eq: sec sMP additional assumption eq 1}
\end{align}
Then, from the first condition of \eqref{eq: sec sMP additional assumption eq 1} and the fact that the restriction of ${\kappa}_{\mathfrak{g}}$ to ${\iota}({\mathfrak{h}})^{\perp}$ is nondegenerate and consistent, we know that ${\mathfrak{t}}^{+}$ and ${\mathfrak{t}}^{-}$ are isomorphic as vector superspaces. In particular, for a parity-homogeneous basis $\{ {y}_{a} \}_{a \in A}$ of ${\mathfrak{t}}^{+}$, we can construct the basis $\{ {y}^{a} \}_{a \in A}$ of ${\mathfrak{t}}^{-}$ such that ${\kappa}_{\mathfrak{g}}({y}_{a}, {y}^{b}) = {\delta}_{ab}$ for $a, b \in A$. If we set ${\mathfrak{a}}^{\pm}_{{\mathfrak{g}} \oplus {\mathfrak{h}}} := {\mathfrak{t}}^{\pm} \oplus {\mathfrak{a}}^{\pm}_{{\iota}({\mathfrak{h}}) \oplus {\mathfrak{h}}}$, then, from \eqref{eq: sec sMP decomposition of example sMP1}, \eqref{eq: sec sMP orthogonal decomposition of g}, and \eqref{eq: sec sMP t decomposition assump of orthogonal complem}, the Lie superalgebra ${\mathfrak{g}} \oplus {\mathfrak{h}}$ admits the following decomposition as a vector superspace:
\begin{align}
{\mathfrak{g}} \oplus {\mathfrak{h}} = {\mathfrak{a}}^{+}_{{\mathfrak{g}} \oplus {\mathfrak{h}}} \oplus {\mathfrak{a}}^{-}_{{\mathfrak{g}} \oplus {\mathfrak{h}}}.
\label{eq: sec sMP decomposition of example sMP2}
\end{align}
Moreover, from \eqref{eq: sec sMP additional assumption eq 1}, the subspaces ${\mathfrak{a}}^{\pm}_{{\mathfrak{g}} \oplus {\mathfrak{h}}}$ are isotropic with respect to the bilinear form ${\kappa}_{{\mathfrak{g}} \oplus {\mathfrak{h}}}$ defined in \eqref{eq: sec sMP bilinear form on g+h}, and ${\mathfrak{a}}^{+}_{{\mathfrak{g}} \oplus {\mathfrak{h}}}$ is a Lie subalgebra of ${\mathfrak{g}} \oplus {\mathfrak{h}}$. Hence, we obtain the following super Manin pair:
\begin{align*}
({\mathfrak{g}} \oplus {\mathfrak{h}}, {\mathfrak{a}}^{+}_{{\mathfrak{g}} \oplus {\mathfrak{h}}}, {\mathfrak{a}}^{-}_{{\mathfrak{g}} \oplus {\mathfrak{h}}}),
\end{align*}
where $\{ {y}_{i} \}_{i \in I}$ and $\{ {y}^{i} \}_{i \in I}$ are parity-homogeneous bases of ${\mathfrak{a}}^{+}_{{\mathfrak{g}} \oplus {\mathfrak{h}}}$ and ${\mathfrak{a}}^{-}_{{\mathfrak{g}} \oplus {\mathfrak{h}}}$, respectively, satisfying
\begin{align*}
{\kappa}_{{\mathfrak{g}} \oplus {\mathfrak{h}}}({y}_{i}, {y}^{j}) = {\delta}_{ij},
\end{align*}
for $i, j \in I$. Here, the index set $I$ is defined by $I = A \cup H$.

\vskip 1mm

In this setting, we denote the structure constants of ${\mathfrak{h}}$ as
\begin{align*}
[{u}_{w}, {u}_{x}] = {\sum}_{z \in H} {h}^{z}_{wx}{u}_{z},
\end{align*}
for $w, x \in H$, and we denote the structure constants of ${\mathfrak{g}} \oplus {\mathfrak{h}}$ with respect to the super Manin pair structure $({\mathfrak{g}} \oplus {\mathfrak{h}}, {\mathfrak{a}}^{+}_{{\mathfrak{g}} \oplus {\mathfrak{h}}}, {\mathfrak{a}}^{-}_{{\mathfrak{g}} \oplus {\mathfrak{h}}})$ as in \eqref{eq: sec sMP structure constant notation for sMP}. Then, for $w, x \in H$, we have
\begin{align*}
[{y}_{w}, {y}_{x}] = [{\iota}({u}_{w}) + {u}_{w}, {\iota}({u}_{x}) + {u}_{x}] = {\iota}([{u}_{w}, {u}_{x}]) + [{u}_{w}, {u}_{x}] = {\sum}_{z \in H} {h}^{z}_{wx}{y}_{z},
\end{align*}
and hence
\begin{align}
{C}^{z}_{wx} = 
\begin{cases}
{h}^{z}_{wx} & \text{ if } {z} \in {H},
\\
0 & \text{ if } {z} \in {A}.
\end{cases}
\label{eq: sec sMP str const C}
\end{align}
Note that if ${\mathfrak{h}}$ is not abelian, then the super Manin pair structrue \eqref{eq: sec sMP decomposition of example sMP2} on ${\mathfrak{g}} \oplus {\mathfrak{h}}$, and the super Manin subpair structure \eqref{eq: sec sMP decomposition of example sMP1} on ${\iota}({\mathfrak{h}}) \oplus {\mathfrak{h}}$, are not super Manin triples, since
\begin{align}
[{y}^{w}, {y}^{x}] = {\frac{1}{4}}{\big(}{\iota}([{u}^{w}, {u}^{x}]) + [{u}^{w}, {u}^{x}]{\big)} \in {\mathfrak{a}}^{+}_{{\iota}({\mathfrak{h}}) \oplus {\mathfrak{h}}},
\label{eq: sec sMP bracket in a+}
\end{align}
for $w, x \in H$.

\vskip 1mm

We now collect all the assumptions, including Assumption \ref{assum: first assumption}, conditions in \eqref{eq: sec sMP t decomposition assump of orthogonal complem} and \eqref{eq: sec sMP additional assumption eq 1}, together with an additional condition on ${\mathfrak{t}}^{-}$, into the following assumption, which we refer to as the {\textit{Kazama--Suzuki $N = 2$ condition}}.

\vskip 1mm

\begin{assumption}[\cites{KS89N, KS89C, Schweigert92, FS97A}]\label{assum: second assumption}\rm
Let ${\mathfrak{g}}$ be a finite-dimensional Lie superalgebra with a nondegenerate consistent supersymmetric invariant bilinear form ${\kappa}_{\mathfrak{g}}$. Suppose that there exists an injective Lie superalgebra homomorphism ${\iota} : {\mathfrak{h}} \rightarrow {\mathfrak{g}}$, and that the bilinear form ${\kappa}_{\mathfrak{h}} := {\kappa}_{\mathfrak{g}} \circ ({\iota} \times {\iota})$ is nondegenerate on ${\mathfrak{h}}$. Assume further that the orthogonal complement of ${\iota}({\mathfrak{h}})$ in ${\mathfrak{g}}$ with respect to ${\kappa}_{\mathfrak{g}}$ admits a vector superspace decomposition ${\iota}({\mathfrak{h}})^{\perp} = {\mathfrak{t}}^{+} \oplus {\mathfrak{t}}^{-}$ such that:
\begin{align*}
{\kappa}_{\mathfrak{g}}({\mathfrak{t}}^{\pm}, {\,} {\mathfrak{t}}^{\pm}) = 0,
{\quad}
[{\mathfrak{t}}^{\pm}, {\iota}({\mathfrak{h}}) \oplus {\mathfrak{t}}^{\pm}] \subseteq {\mathfrak{t}}^{\pm}.
\end{align*}
\end{assumption}

\vskip 1mm

Then we have the following simple observation on the structure constants with respect to the super Manin pair structure on ${\mathfrak{g}} \oplus {\mathfrak{h}}$.

\vskip 1mm

\begin{lemma}\label{lem: sec sMP F and ring F vanish}\rm
Under Assumption \ref{assum: second assumption}, we obtain the following relation for the structure constants of ${\mathfrak{g}} \oplus {\mathfrak{h}}$ with respect to the super Manin pair structure in \eqref{eq: sec sMP decomposition of example sMP2}:
\begin{gather}
{F}^{ij}_{w} = 0
{\quad}
{\text{for $i, j \in I$, $w \in H$}},
\label{eq: sec sMP str const F zero}
\\
{\mathring{F}}^{ij}_{k} = 0
{\quad}
{\text{unless $i, j, k \in H$}}.
\label{eq: sec sMP str const ring F zero}
\end{gather}
\end{lemma}

\begin{proof}
First, for $i, j \in H$, it follows from \eqref{eq: sec sMP bracket in a+} that $[{y}^{i}, {y}^{j}] \in {\mathfrak{a}}^{+}_{{\iota}({\mathfrak{h}}) \oplus {\mathfrak{h}}}$. If $i \in H$ and $j \in A$, then we have
\begin{align*}
[{y}^{i}, {y}^{j}] = {\frac{1}{2}}[{\iota}({u}^{i}) - {u}^{i}, {y}^{j}] = {\frac{1}{2}}[{\iota}({u}^{i}), {y}^{j}] \in {\mathfrak{t}}^{-},
\end{align*}
and if $i, j \in A$, then $[{y}^{i}, {y}^{j}] \in {\mathfrak{t}}^{-}$. Therefore, the desired result follows.
\end{proof}

\vskip 1mm

Moreover, if ${\mathfrak{h}}$ is a {\textit{unimodular}} Lie superalgebra, i.e., ${\sum}_{z \in H} (-1)^{z}{h}^{z}_{wz} = 0$ for all $w \in H$, then we have the following lemma for the coefficients ${\mathcal{E}}^{ij}$ defined in \eqref{eq: def of E^ij}.

\vskip 1mm

\begin{lemma}\label{lem: sec sMP coeff E vanish}\rm
Under Assumption \ref{assum: second assumption}, suppose that ${\mathfrak{h}}$ is unimodular. Then, for the coefficients ${\mathcal{E}}^{ij}$ of the super Manin pair structure \eqref{eq: sec sMP decomposition of example sMP2} on ${\mathfrak{g}} \oplus {\mathfrak{h}}$, we have ${\mathcal{E}}^{ij} = 0$ for $i, j \in I$.
\end{lemma}

\begin{proof}
First, for $i, j \in I$, from \eqref{eq: sec sMP equality for map E} and \eqref{eq: sec sMP str const ring F zero}, we have
\begin{align}
\begin{split}
{\mathcal{E}}^{ij} 
&= {\sum}_{m,n \in I} (-1)^{mn}{\big(}{C}^{i}_{mn}{\mathring{F}}^{mn}_{j} - (-1)^{m + n}{C}^{j}_{mn}{\mathring{F}}^{mn}_{i}{\big)}
\\
&= {\sum}_{m,n \in H} (-1)^{mn}{\big(}{C}^{i}_{mn}{\mathring{F}}^{mn}_{j} - (-1)^{m + n}{C}^{j}_{mn}{\mathring{F}}^{mn}_{i}{\big)}.
\end{split}
\label{eq: sec sMP proof of lem for coeff E}
\end{align}

\vskip 1mm

Now, in the case where $i, j \in H$, since ${\iota}({\mathfrak{h}}) \oplus {\mathfrak{h}}$ itself forms a super Manin subpair of ${\mathfrak{g}} \oplus {\mathfrak{h}}$, it follows again from \eqref{eq: sec sMP equality for map E} that the last term in \eqref{eq: sec sMP proof of lem for coeff E} becomes:
\begin{align}
{\sum}_{m,n \in H} (-1)^{n}{\big(}{F}^{nm}_{n}{F}^{ij}_{m} + {C}^{n}_{mn}{\mathring{F}}^{ij}_{m}{\big)} = {\sum}_{m,n \in H} (-1)^{n}{h}^{n}_{mn}{\mathring{F}}^{ij}_{m} = 0,
\label{eq: sec sMP proof of lem for coeff E 2}
\end{align}
where the first equality in \eqref{eq: sec sMP proof of lem for coeff E 2} follows from \eqref{eq: sec sMP str const F zero} and \eqref{eq: sec sMP str const C}, while the last equality follows from the unimodularity of ${\mathfrak{h}}$. If $i \in A$ and $j \in I$, then we also have ${\mathcal{E}}^{ij} = 0$ from \eqref{eq: sec sMP proof of lem for coeff E}, \eqref{eq: sec sMP str const C}, and \eqref{eq: sec sMP str const ring F zero}.
\end{proof}

\vskip 4mm
\section{$N = 2$ and $N = {\frac{3}{2}}$ superconformal vertex algebras}\label{section: N=2 and N=3/2 sconf VA}
\vskip 2mm

In this section, we review the definition of an $N = 2$ superconformal vertex algebra, one of its equivalent formulations, and its shift. We also briefly recall the definition of an $N = {\frac{3}{2}}$ superconformal vertex algebra, and introduce an equivalent formulation. Note that the $N = 2$ superconformal algebra was introduced in \cites{Ademollo76, DVPYZ86}. Various explicit constructions of the $N = 2$ superconformal algebra and its representation theory have also been studied (see, for example, \cites{Adamovic99, FGLS92, OS90, Poletaeva01, STF01, Wakimoto98}).

\vskip 1mm

Following \cites{BK03, DSK06, Kac98}, we first recall the definition of a vertex algebra. Let ${V}$ be a vector superspace, and let ${z}$ be an even indeterminate. A {\textit{field}} is an ${\text{End}}(V)$-valued formal distribution of the form
\begin{align*}
a(z) = {\sum}_{{n} \in {\mathbb{Z}}}{a}_{(n)}{z}^{-n-1} \in {\text{End}}(V)[[{z}, {z}^{-1}]],
\end{align*}
such that ${a}_{(n)}{v} = 0$ for ${n}$ sufficiently large, for each ${v} \in {V}$.

\vskip 1mm

A {\textit{vertex algebra}} is a vector superspace ${V}$ equipped with an even vector ${\ket{0}}$ called the {\textit{vacuum vector}}, an even endomorphism ${\partial} \in {\text{End}}(V)$ called the {\textit{infinitesimal translation operator}}, and a parity-preserving linear map $Y({\,}{\cdot}{\,}, {z})$ from ${V}$ to the space of ${\text{End}}(V)$-valued fields, i.e., ${a} \mapsto Y(a, z) = a(z) = {\sum}_{{n} \in {\mathbb{Z}}}{a}_{(n)}{z}^{-n-1}$ such that
\begin{align*}
{\text{(i)}}&
{\quad}
{\partial}{\ket{0}} = 0,
{\quad}
a(z){\ket{0}}{|}_{z = 0} = a,
\\
{\text{(ii)}}&
{\quad}
[{\partial}, a(z)] = {\partial}_{z}a(z),
\\
{\text{(iii)}}&
{\quad}
{(z - w)}^{n}[a(z), b(w)] = 0,
{\quad}
{\text{for some ${n} \in {\mathbb{Z}}_{\geq 0}$}},
\end{align*}
for $a, b \in {V}$.

\vskip 1mm

Note that the bilinear product ${\cdot} {}_{(n)} {\cdot}$ is called the {\textit{${n}$-th product}} for ${n} \in {\mathbb{Z}}$. In particular, the $(-1)$-th product ${a}_{(-1)}{b}$ is denoted by $:ab:$ and is called the {\textit{normally ordered product}}. In addition, the equality ${a}_{(-n-1)}{b} = {\frac{1}{n!}}:({\partial}^{n}{a}){b}:$ holds for $n \in {\mathbb{Z}}_{\geq 0}$. We denote the normally ordered product $:ab:$ simply by $ab$, and the expression ${a}_{1}{a}_{2} \cdots {a}_{l-1}{a}_{l}$ denotes the element obtained by performing the normally ordered product from right to left. For the notions of a vertex algebra homomorphism (i.e., a parity-preserving linear map that sends the vacuum vector to the vacuum vector and preserves all ${n}$-th products), vertex subalgebra, ideal of a vertex algebra, and tensor product of vertex algebras, we refer to Section 4.3 of \cite{Kac98}. These notions can also be found in other literature, such as Section 3.9 and Section 3.12 of \cite{LL04}, and Section 2.4, Section 2.5, and Section 4.3 of \cite{FHL93}.

\vskip 1mm

Let $V$ be a vertex algebra. In order to encode the singular part of the operator product expansion (OPE) (see equations (2.3.7b), (2.3.9a), and (2.3.11) of \cite{Kac98}), we consider a parity-preserving ${\mathbb{C}}$-linear map, called the {\textit{${\lambda}$-bracket}}, ${V} \otimes {V} \to {\mathbb{C}}[{\lambda}] \otimes {V}$ defined by
\begin{align*}
{a} \otimes {b} {\,} \mapsto {\,} [{a} {\,} {}_{\lambda} {\,} {b}] = {\sum}_{n \in {\mathbb{Z}}_{\geq 0}} {\frac{{\lambda}^{n}}{n!}}{a}_{(n)}{b},
\end{align*}
where ${\lambda}$ is an even indeterminate. Then the ${\lambda}$-bracket satisfies the following properties:
\begin{enumerate}[]
\item (Sesquilinearity)
{\quad}
$
[{\partial}{a} {\,} {}_{\lambda} {\,} {b}] = - {\lambda}[{a} {\,} {}_{\lambda} {\,} {b}],
{\quad}
[{a} {\,} {}_{\lambda} {\,} {\partial}{b}] = ({\partial} + {\lambda})[{a} {\,} {}_{\lambda} {\,} {b}],
$
\item (Skew-symmetry)
{\quad}
$
[{b} {\,} {}_{\lambda} {\,} {a}] = (-1)^{ab + 1}[{a} {\,} {}_{-{\partial}-{\lambda}} {\,} {b}],
$
\item (Jacobi identity) 
{\quad}
$
[{a} {\,} {}_{\lambda} {\,} [{b} {\,} {}_{\gamma} {\,} {c}]] = [[{a} {\,} {}_{\lambda} {\,} {b}] {\,} {}_{{\lambda} + {\gamma}} {\,} {c}] + (-1)^{ab}[{b} {\,} {}_{\gamma} {\,} [{a} {\,} {}_{\lambda} {\,} {c}]],
$
\end{enumerate}
where $a, b, c$ are parity-homogeneous elements in $V$. More precisely, the right-hand side of the second equality in sesquilinearity can be computed by the relation ${\partial}{\lambda} = {\lambda}{\partial}$. The right-hand side of skew-symmetry can be interpreted as
\begin{align*}
[{a} {\,} {}_{-{\partial}-{\lambda}} {\,} {b}] =  {\sum}_{n \in {\mathbb{Z}}_{\geq 0}} {\frac{{(- {\partial} - {\lambda})}^{n}}{n!}} {a}_{(n)}{b},
\end{align*}
and the Jacobi identity is defined in ${\mathbb{C}}[{\lambda}, {\gamma}] \otimes {V}$. In addition, we assume that ${\lambda}{\gamma} = {\gamma}{\lambda}$, where ${\gamma}$ is also an even indeterminate. Moreover, the ${\lambda}$-bracket and the normally ordered product satisfy the following additional properties:
\begin{enumerate}[]
\item (Quasi-commutativity)
{\quad}
$
\displaystyle
{a}{b} - (-1)^{ab}{b}{a} = \int^{0}_{-{\partial}} [{a} {\,} {}_{\lambda} {\,} {b}] {\,} {d}{\lambda},
$
\item (Quasi-associativity)
\begin{equation*}
({a}{b}){c} - {a}({b}{c}) = {\Big(} \int^{\partial}_{0} {\,} {d}{\lambda} {\,} {a} {\Big)}[{b} {\,} {}_{\lambda} {\,} {c}] + (-1)^{ab}{\Big(} \int^{\partial}_{0} {\,} {d}{\lambda} {\,} {b} {\Big)}[{a} {\,} {}_{\lambda} {\,} {c}],
\end{equation*}
\item (Non-commutative Wick formula)
\begin{equation*}
[{a} {\,} {}_{\lambda} {\,} {b}{c}] = [{a} {\,} {}_{\lambda} {\,} {b}]{c} + (-1)^{ab}{b}[{a} {\,} {}_{\lambda} {\,} {c}] + \int^{\lambda}_{0} [[{a} {\,} {}_{\lambda} {\,} {b}] {\,} {}_{\gamma} {\,} {c}] {\,} {d}{\gamma},
\end{equation*}
\end{enumerate}
where $(\int^{\partial}_{0} {\,} {d}{\lambda} {\,} {a})[{b} {\,} {}_{\lambda} {\,} {c}] := {\sum}_{n \in {\mathbb{Z}}_{\geq 0}}{a}_{(-n-2)}({b}_{(n)}{c})$ in quasi-associativity. For more details, see \cite{BK03} and Section 1.5 of \cite{DSK06}.

\vskip 1mm

We now begin with the definition of a conformal vector (see Section 5 of \cite{Borcherds86} and Definition 4.10 of \cite{Kac98}). For a vertex algebra $V$, an even vector ${T} \in {V}$ is called a {\textit{conformal vector}} if its zero-mode coincides with the translation operator of the vertex algebra, i.e. ${T}_{(0)} = {\partial}$, if ${T}_{(1)}$ acts semisimply on $V$ and if it satisfies the following OPE relation:
\begin{align*}
[{T} {\,} {}_{\lambda} {\,} {T}] = ({\partial} + 2{\lambda}){T} + {\frac{c}{12}}{\lambda}^{3},
\end{align*}
for some ${c} \in {\mathbb{C}}$. In this case, the complex number ${c}$ is called the {\textit{central charge}} of ${T}$. For a conformal vector ${T}$ of a vertex algebra ${V}$, if a nonzero vector $v \in V$ satisfies
\begin{align}
[{T} {\,} {}_{\lambda} {\,} {v}] = ({\partial} + {\Delta}_{v}{\lambda}){v} + O({\lambda}^{2}),
\label{eq: conformal weight}
\end{align}
for some ${\Delta}_{v} \in {\mathbb{C}}$, then ${\Delta}_{v}$ is called the {\textit{conformal weight}} of ${v}$, where $O({\lambda}^{2})$ is a polynomial in ${\lambda}$ without constant or linear terms. Moreover, if $O({\lambda}^{2}) = 0$ in \eqref{eq: conformal weight}, then the vector ${v}$ is called {\textit{primary}} of conformal weight ${\Delta}_{v}$.

\vskip 1mm

Following Section 5.9 of \cite{Kac98}, we recall the definitions of superconformal structures of vertex algebras. An odd vector ${G}$ of a vertex algebra ${V}$ is called an {\textit{$N = 1$ superconformal vector}} if $T := {\frac{1}{2}}{G}_{(0)}{G}$ is a conformal vector with central charge ${c} \in {\mathbb{C}}$ and the following OPEs hold:
\begin{align*}
[{T} {\,} {}_{\lambda} {\,} {G}] = ({\partial} + {\frac{3}{2}}{\lambda}){G},
{\quad}
[{G} {\,} {}_{\lambda} {\,} {G}] = 2{T} + {\frac{c}{3}}{\lambda}^{2}.
\end{align*}
As in Section 1.2 of \cite{HK07}, an $N = 1$ superconformal vector $G$ of a vertex algebra $V$ gives rise to an $N = 1$ supersymmetric vertex algebra structure on $V$ via the odd derivation ${D} := {G}_{(0)}$.

\vskip 1mm

\begin{definition}[\cites{Kac98, KW04}]\rm
A vertex algebra $V$ is called an {\textit{$N = 2$ superconformal vertex algebra}} if it has two odd vectors ${G}^{+}$, ${G}^{-}$ such that ${T} := {\frac{1}{2}}({G}^{+}_{(0)}{G}^{-} + {G}^{-}_{(0)}{G}^{+})$ is a conformal vector with central charge ${c} \in {\mathbb{C}}$, and the following $N = 2$ super-Virasoro relations hold:
\begin{gather*}
[{T} {\,} {}_{\lambda} {\,} {G}^{\pm}] = ({\partial} + {\frac{3}{2}}{\lambda}){G}^{\pm},
{\quad}
[{G}^{\pm} {\,} {}_{\lambda} {\,} {G}^{\pm}] = 0,
{\quad}
[{G}^{+} {\,} {}_{\lambda} {\,} {G}^{-}] = {T} + {\frac{1}{2}}{\partial}{U} + {\lambda}{U} + {\frac{c}{6}}{\lambda}^{2},
\\
[{T} {\,} {}_{\lambda} {\,} {U}] = ({\partial} + {\lambda}){U},
{\quad}
[{G}^{\pm} {\,} {}_{\lambda} {\,} {U}] = {\mp}{G}^{\pm},
{\quad}
[{U} {\,} {}_{\lambda} {\,} {U}] = {\frac{c}{3}}{\lambda},
\end{gather*}
where ${U} := {G}^{+}_{(1)}{G}^{-}$. In this case, we call the collection ${T}, {G}^{\pm}, {U}$ (or simply ${G}^{\pm}$) an {\textit{$N = 2$ superconformal structure}} of $V$.
\label{def: N=2 sconf VA}
\end{definition}

\vskip 1mm

Note that, from Remark 2.8 of \cite{HK07}, an $N = 2$ superconformal vertex algebra admits an $N = 2$ supersymmetric vertex algebra structure via the two odd derivations $({\mu}{G}^{+} + {\mu}^{-1}{G}^{-})_{(0)}$ and ${\pm}{\sqrt{-1}}({\mu}{G}^{+} - {\mu}^{-1}{G}^{-})_{(0)}$ for ${\mu} \in {\mathbb{C}}^{*}$. This ${\mathbb{Z}} \slash {2}{\mathbb{Z}} \times {\textrm{SO}}_{2}(\mathbb{C})$ family of $N = 2$ supersymmetric vertex algebra structures arises from the automorphism group (see also \cite{SS87}) of the $N = 2$ super-Virasoro vertex algebra, which preserves the conformal vector.

\vskip 1mm

We now briefly recall the definition of the $N = {\frac{3}{2}}$ superconformal structure of vertex algebras, which is slightly different from the $N = 2$ superconformal structure. This algebra was introduced in \cite{Kazama91} and can also be found in \cites{IR93, IR94, Getzler95}. Here we present the OPEs of the $N = {\frac{3}{2}}$ superconformal algebra following the notation of \cite{Getzler95}.

\vskip 1mm

\begin{definition}[\cites{Getzler95}]\rm
Let $V$ be a vertex algebra with a conformal vector $T$ of central charge ${c} \in {\mathbb{C}}$. Suppose that ${G}^{+}$, ${G}^{-}$, and $\Phi$ are odd primary vectors of conformal weight $\frac{3}{2}$, and that $U$ and $F$ are even primary vectors of conformal weights $1$ and $2$, respectively. Then the remaining nonzero OPEs of the $N = {\frac{3}{2}}$ superconformal algebra are given as follows:
\begin{gather*}
[{G}^{+} {\,} {}_{\lambda} {\,} {G}^{-}] = {T} + {\frac{1}{2}}{\partial}{U} + {\lambda}{U} + {\frac{c}{6}}{\lambda}^{2},
{\quad}
[{G}^{+} {\,} {}_{\lambda} {\,} {\Phi}] = {F},
{\quad}
[{G}^{-} {\,} {}_{\lambda} {\,} {G}^{-}] = -2{F},
\\
[{G}^{-} {\,} {}_{\lambda} {\,} {F}] = {\partial}{\Phi} + 3{\lambda}{\Phi},
{\quad}
[{U} {\,} {}_{\lambda} {\,} {G}^{+}] = {G}^{+},
{\quad}
[{U} {\,} {}_{\lambda} {\,} {G}^{-}] = -{G}^{-},
\\
[{U} {\,} {}_{\lambda} {\,} {\Phi}] = -3{\Phi},
{\quad}
[{U} {\,} {}_{\lambda} {\,} {F}] = -2{F},
{\quad}
[{U} {\,} {}_{\lambda} {\,} {U}] = {\frac{c}{3}}{\lambda},
\end{gather*}
and we call a collection of vectors ${T}, {G}^{\pm}, {U}, {\Phi}, {F}$ (or simply ${G}^{\pm}$) satisfying the above OPEs of the $N = {\frac{3}{2}}$ superconformal algebra an {\textit{$N = {\frac{3}{2}}$ superconformal structure}} of ${V}$.
\label{def: N=3/2 sconf VA}
\end{definition}

\vskip 1mm

The following Lemma \ref{lem: shift of N=2 sconf} describes a {\textit{shift}} of an $N = 2$ superconformal structure. The shift of a conformal vector is discussed in Section 5 of \cite{Arakawa15} and Proposition 4.4 of \cite{CKLW18}. In particular, the shift of an $N = 1$ superconformal vector is given in Proposition 1.1 of \cite{Getzler95} (see also Proposition 4.3 of \cite{SY25}), and the shift of an $N = {\frac{3}{2}}$ superconformal structure appears in Proposition 1.4 of \cite{Getzler95}.

\vskip 1mm

\begin{lemma}[\cite{Getzler95}]\rm
Let $V$ be a vertex algebra with an $N = 2$ superconformal structure ${T}$, ${G}^{\pm}$, ${U}$ of central charge ${c} \in {\mathbb{C}}$. Suppose that ${v}_{+}, {v}_{-} \in {V}$ are odd primary vectors of conformal weight ${\frac{1}{2}}$ satisfying the following relations:
\begin{align*}
[{G}^{a} {\,} {}_{\lambda} {\,} {v}_{b}] = {\delta}_{ab}{\,}({G}^{a}_{(0)}{v}_{a} + {c}_{a}{\lambda}),
{\quad}
[{U} {\,} {}_{\lambda} {\,} {v}_{\pm}] = {\mp}{v}_{\pm},
{\quad}
[{v}_{a} {\,} {}_{\lambda} {\,} {v}_{b}] = {\delta}_{ab}{\,}{d}_{a},
{\qquad}
{\text{for ${a}, {b} \in {\{} +, - {\}}$,}}
\end{align*}
for some ${c}_{+}, {c}_{-}, {d}_{+}, {d}_{-} \in {\mathbb{C}}$, and such that the operators $({{G}^{\pm}}_{(0)}{v}_{\pm})_{(0)}$ act diagonally on $V$. Then the following vectors:
\begin{align*}
{T} + {\frac{1}{2}}{\mu}_{+}{\partial}{G}^{+}_{(0)}{v}_{+} + {\frac{1}{2}}{\mu}_{-}{\partial}{G}^{-}_{(0)}{v}_{-},
{\quad}
{G}^{\pm} + {\mu}_{\mp}{\partial}{v}_{\mp},
{\quad}
{U} + {\mu}_{+}{G}^{+}_{(0)}{v}_{+} - {\mu}_{-}{G}^{-}_{(0)}{v}_{-},
\end{align*}
form an $N = 2$ superconformal structure of $V$ with central charge ${c} + {6}{\mu}_{+}{c}_{+} + {6}{\mu}_{-}{c}_{-}$, for any ${\mu}_{+}, {\mu}_{-} \in {\mathbb{C}}$.
\label{lem: shift of N=2 sconf}
\end{lemma}

\vskip 1mm

For example, a $1$-parameter deformation of the standard $N = 2$ superconformal structure of the $bc$--${\beta}{\gamma}$ system can be found in \cite{FGLS92} (and will be discussed again in detail in Section \ref{section: N=2 sconf of SUSY semi inf cohom}). As another example, shifts of the $N = 2$ superconformal structure of the super-affine vertex algebra associated with a Lie superalgebra admitting a super Manin triple structure can be found in Proposition 6.4.2 of \cite{Creutzig09} and in equations (28), (29) of \cite{CRS09}.

\vskip 1mm

We now present Lemma \ref{lem: equiv def of N=2 sconf}, which gives an equivalent definition of $N = 2$ superconformal structures by starting from a given $N = 1$ superconformal vector and constructing a compatible $N = 2$ superconformal structure.

\vskip 1mm

\begin{lemma}\label{lem: equiv def of N=2 sconf}\rm
Let $V$ be a vertex algebra with an $N = 1$ superconformal vector $G$ of central charge $c \in {\mathbb{C}}$. Suppose that $U \in V$ is an even primary vector of conformal weight $1$ satisfying the following OPEs:
\begin{align}
[{G} {\,} {}_{\lambda} {\,} {U}] = {G}_{(0)}{U},
{\quad}
[{U} {\,} {}_{\lambda} {\,} {U}] = {\frac{c}{3}}{\lambda},
{\quad}
{U}_{(0)}({G}_{(0)}{U}) = - {G}.
\label{eq: equiv def of N=2 sconf}
\end{align}
Then the vectors ${G}^{\pm} := {\frac{1}{2}}({G} \mp {G}_{(0)}{U})$ form an $N = 2$ superconformal structure of $V$ with central charge $c$.
\end{lemma}

\begin{proof}
First, if we denote by ${T}$ the conformal vector corresponding to ${G}$, then by using sesquilinearity, skew-symmetry, and the Jacobi identity of vertex algebras, we obtain the following identities:
\begin{gather}
[{T} {\,} {}_{\lambda} {\,} {G}_{(0)}{U}] = ({\partial} +{\frac{3}{2}}{\lambda}){G}_{(0)}{U},
{\quad}
[{G} {\,} {}_{\lambda} {\,} {G}_{(0)}{U}] = ({\partial} + {2}{\lambda}){U},
\label{eq: equiv def of N=2 sconf AB}
\\
[{U} {\,} {}_{\lambda} {\,} {G}_{(0)}{U}] = -{G},
{\quad}
[{G}_{(0)}{U} {\,} {}_{\lambda} {\,} {G}_{(0)}{U}] = -(2{T} + {\frac{c}{3}}{\lambda}^{2}).
\label{eq: equiv def of N=2 sconf CD}
\end{gather}
For example, to prove the first identity in \eqref{eq: equiv def of N=2 sconf CD}, we note that by assumption $[{G} {\,} {}_{\lambda} {\,} {U}] = {G}_{(0)}{U}$, and by skew-symmetry we also have $[{U} {\,} {}_{\lambda} {\,} {G}] = -{G}_{(0)}{U}$. Hence, from the Jacobi identity, we have:
\begin{align*}
[{U} {\,} {}_{\lambda} {\,} {G}_{(0)}{U}] 
= [{U} {\,} {}_{\lambda} {\,} [{G} {\,} {}_{\gamma} {\,} {U}]] 
= [[{U} {\,} {}_{\lambda} {\,} {G}] {\,} {}_{\lambda + \gamma} {\,} {U}] 
+ [{G} {\,} {}_{\gamma} {\,} [{U} {\,} {}_{\lambda} {\,} {U}]] 
= - [{G}_{(0)}{U} {\,} {}_{\lambda + \gamma} {\,} {U}],
\end{align*}
so that substituting ${\gamma} = -{\lambda}$ yields $[{U} {\,} {}_{\lambda} {\,} {G}_{(0)}{U}] = - ({G}_{(0)}{U})_{(0)}{U}$. That is, since the ${\lambda}$-bracket $[{U} {\,} {}_{\lambda} {\,} {G}_{(0)}{U}]$ contains only the constant term in ${\lambda}$, the assumption ${U}_{(0)}({G}_{(0)}{U}) = - {G}$ implies the desired identity $[{U} {\,} {}_{\lambda} {\,} {G}_{(0)}{U}] = -{G}$. Also, the second equality in equation \eqref{eq: equiv def of N=2 sconf CD} follows from the second identity in \eqref{eq: equiv def of N=2 sconf AB} and the first identity in \eqref{eq: equiv def of N=2 sconf CD}.

\vskip 1mm

If we now consider the odd vectors defined by ${G}^{\pm} := {\frac{1}{2}}({G} \mp {G}_{(0)}{U})$, we see from \eqref{eq: equiv def of N=2 sconf AB} and \eqref{eq: equiv def of N=2 sconf CD} that the vectors ${G}^{\pm}$ satisfy all the conditions of an $N = 2$ superconformal structure in Definition \ref{def: N=2 sconf VA}. For example, we have:
\begin{align*}
[{G}^{+} {\,} {}_{\lambda} {\,} {G}^{-}] 
&= {\frac{1}{4}}[{G} - {G}_{(0)}{U} {\,} {}_{\lambda} {\,} {G} + {G}_{(0)}{U}] 
\\
&= {\frac{1}{4}} {\big(}[{G} {\,} {}_{\lambda} {\,} {G}] - [{G}_{(0)}{U} {\,} {}_{\lambda} {\,} {G}] + [{G} {\,} {}_{\lambda} {\,} {G}_{(0)}{U}] - [{G}_{(0)}{U} {\,} {}_{\lambda} {\,} {G}_{(0)}{U}]{\big)}
\\
&= {T} + {\frac{1}{2}}{\partial}{U} + {\lambda}{U} + {\frac{c}{6}}{\lambda}^{2},
\end{align*}
which follows from the second equalities in \eqref{eq: equiv def of N=2 sconf AB} and \eqref{eq: equiv def of N=2 sconf CD}.
\end{proof}

\vskip 1mm

Conversely, suppose that ${T}, {G}^{\pm}, {U}$ form an $N = 2$ superconformal structure satisfying the conditions of Definition \ref{def: N=2 sconf VA}. Then one can see that ${G} := {G}^{+} + {G}^{-}$ and ${U}$ satisfy the conditions of Lemma \ref{lem: equiv def of N=2 sconf}. Therefore, Lemma \ref{lem: equiv def of N=2 sconf} provides an equivalent definition of an $N = 2$ superconformal structure, and we will also refer to a pair ${G}$ and ${U}$ satisfying the conditions of Lemma \ref{lem: equiv def of N=2 sconf} as an $N = 2$ superconformal structure.

\vskip 1mm

We also introduce an equivalent definition of $N = {\frac{3}{2}}$ superconformal structures in terms of ${G}$ and ${U}$, in the same manner as in Lemma \ref{lem: equiv def of N=2 sconf}.

\vskip 1mm

\begin{proposition}\label{prop: equiv def of N=3/2 sconf}\rm
Let $V$ be a vertex algebra with an $N = 1$ superconformal vector $G$ of central charge $c \in {\mathbb{C}}$. Suppose that $U \in V$ is an even primary vector of conformal weight $1$ satisfying the following OPEs:
\begin{align}
\label{eq: eq1 in equiv def of N=3/2}
[{G} {\,} {}_{\lambda} {\,} {U}] = {G}_{(0)}{U},
{\quad}
[{U} {\,} {}_{\lambda} {\,} {U}] = {\frac{c}{3}}{\lambda},
{\quad}
[{U} {\,} {}_{\lambda} {\,} {G}_{(0)}{U}] = {U}_{(0)}({G}_{(0)}{U}),
\end{align}
and that ${U}_{(0)}({G}_{(0)}{U}) \neq - {G}$. If we set ${\Phi} := - {\frac{1}{8}}({G} + {U}_{(0)}({G}_{(0)}{U}))$, and the following additional OPEs hold:
\begin{align}
\label{eq: eq2 in equiv def of N=3/2}
[{G} {\,} {}_{\lambda} {\,} {\Phi}] = {G}_{(0)}{\Phi},
{\quad}
[{G}_{(0)}{U} {\,} {}_{\lambda} {\,} {\Phi}] = - {G}_{(0)}{\Phi},
{\quad}
[{U} {\,} {}_{\lambda} {\,} {\Phi}] = -3{\Phi},
{\quad}
[{\Phi} {\,} {}_{\lambda} {\,} {\Phi}] = 0,
\end{align}
then ${T}, {G}^{\pm}, {U}, {\Phi}, {F}$ give rise to an $N = {\frac{3}{2}}$ superconformal structure on $V$, where ${T}$ is the conformal vector corresponding to ${G}$, and ${G}^{\pm}$, ${F}$ are given by
\begin{align*}
{G}^{+} := {\frac{1}{2}}({G} - {G}_{(0)}{U}) + {\Phi},
{\quad}
{G}^{-} := {\frac{1}{2}}({G} + {G}_{(0)}{U}) - {2}{\Phi},
{\quad}
{F} := {G}_{(0)}{\Phi}.
\end{align*}
\end{proposition}

\begin{proof}
The proof is given in Appendix \ref{section: Pf of Prop in Sec N=2 and N=3/2 sconf VA}.
\end{proof}

\vskip 1mm

Conversely, suppose that ${T}, {G}^{\pm}, {U}, {\Phi}, {F}$ form an $N = {\frac{3}{2}}$ superconformal structure in the sense of Definition \ref{def: N=3/2 sconf VA}. If we set ${G} := {G}^{+} + {G}^{-} + {\Phi}$, then ${G}$ and ${U}$ satisfy all the conditions of Proposition \ref{prop: equiv def of N=3/2 sconf}. Therefore, we can also refer to ${G}$ and ${U}$ satisfying the conditions of Proposition \ref{prop: equiv def of N=3/2 sconf} as an $N = {\frac{3}{2}}$ superconformal structure.

\vskip 4mm
\section{$N = {\frac{3}{2}}$ superconformal structure of super--affine vertex algebras}\label{section: N=3/2 sconf of super affine VA}
\vskip 2mm

Note that in Section 4 of \cite{Getzler95}, the super-affine vertex algebra associated to a Lie algebra that admits a Manin pair structure carries an $N = {\frac{3}{2}}$ superconformal structure. In this section, rather than directly verifying the OPEs of the $N = {\frac{3}{2}}$ superconformal algebra (see Definition \ref{def: N=3/2 sconf VA}) in the case where a Lie superalgebra ${\mathfrak{g}}$ admits a super Manin pair structure, we proceed as follows. We introduce a $U(1)$-current ${U}_{\mathfrak{g}}$ in the super-affine vertex algebra associated to ${\mathfrak{g}}$. Together with the Kac--Todorov superconformal vector ${G}_{\mathfrak{g}}$, we will show that the pair ${G}_{\mathfrak{g}}$ and ${U}_{\mathfrak{g}}$ satisfies all the conditions appearing in Lemma \ref{lem: equiv def of N=2 sconf}, except for the last relation in equation \eqref{eq: equiv def of N=2 sconf}. The remaining relation corresponds to the expression ${U}_{{\mathfrak{g}} {\,} (0)}({G}_{{\mathfrak{g}} {\,} (0)}{U}_{\mathfrak{g}})$, whose result will be computed in Proposition \ref{prop: sec N=3/2 computation of UGU}. Thereafter, we consider a Lie superalgebra ${\mathfrak{g}} \oplus {\mathfrak{h}}$ equipped with the super Manin pair structure given in \eqref{eq: sec sMP decomposition of example sMP2}. Under Assumption \ref{assum: second assumption} and the unimodularity of ${\mathfrak{h}}$, we conclude the discussion of the $N = {\frac{3}{2}}$ superconformal structure on the super-affine vertex algebra associated to ${\mathfrak{g}} \oplus {\mathfrak{h}}$ in Theorem \ref{theorem: N=3/2 of saff}.

\vskip 1mm

We first recall the definition of the super-affine vertex algebra, which is the main object of this section, following \cite{KT85} (see also Example 2.6 and Example 5.9 of \cite{HK07}). Let ${\mathfrak{g}}$ be a finite-dimensional Lie superalgebra with a consistent supersymmetric invariant bilinear form ${\kappa}_{\mathfrak{g}}$. The {\textit{super-affine vertex algebra}} associated to ${\mathfrak{g}}$, denoted by ${V}^{{\kappa}_{\mathfrak{g}}}_{{\,\,} {\text{s}}}({\mathfrak{g}})$, is a vertex algebra generated by ${\overline{\mathfrak{g}}} \cup {\mathfrak{g}}$ with the following OPEs:
\begin{align*}
[{a} {\,} {}_{\lambda} {\,} {b}] = [a, b] + {\kappa}_{\mathfrak{g}}({a}, {b}){\lambda},
{\quad}
[{a} {\,} {}_{\lambda} {\,} {\overline{b}}] = (-1)^{a}{\overline{[a, b]}},
{\quad}
[{\overline{a}} {\,} {}_{\lambda} {\,} {b}] = {\overline{[a, b]}},
{\quad}
[{\overline{a}} {\,} {}_{\lambda} {\,} {\overline{b}}] = (-1)^{a}{\kappa}_{\mathfrak{g}}({a}, {b}),
\end{align*}
for $a, b \in {\mathfrak{g}}$, where ${\overline{\mathfrak{g}}} = {\{} {\overline{a}} {\,} {|} {\,} {a} \in {\mathfrak{g}} {\}}$ is the parity-reversed vector superspace of ${\mathfrak{g}}$, i.e. $p({\overline{a}}) = p(a) + 1$ (see, for example, Proposition 3.9.3 of \cite{LL04} and equation (2.4) of \cite{CL19}, for the definition of generators of a vertex algebra). Note that, from Lemma 3 of \cite{KT85} (see also Proposition B.2 of \cite{Sato22}), the super-affine vertex algebra is isomorphic to the tensor product of the affine vertex algebra and the $bc$--${\beta}{\gamma}$ system.

\vskip 1mm

Suppose that ${\kappa}_{\mathfrak{g}}$ is nondegenerate. Let $\{ {a}_{i} \}_{i \in G}$ be a parity-homogeneous basis of ${\mathfrak{g}}$, and let $\{ {b}_{i} \}_{i \in G}$ be the dual basis with respect to ${\kappa}_{\mathfrak{g}}$, that is, ${\kappa}_{\mathfrak{g}}({a}_{i}, {b}_{j}) = {\delta}_{ij}$ and ${\kappa}_{\mathfrak{g}}({b}_{i}, {a}_{j}) = (-1)^{i}{\delta}_{ij}$ for $i,j \in G$. Then, from \cite{KT85}, the following vector
\begin{align}
{G}_{\mathfrak{g}} = 
{\sum}_{i \in G} (-1)^{i}{\overline{a}}_{i}{b}_{i} 
+ {\frac{1}{3}}{\sum}_{i,j,k \in G} (-1)^{ij + i} {\kappa}_{\mathfrak{g}}({b}_{i}, [{b}_{j}, {b}_{k}]){\overline{a}}_{i}({\overline{a}}_{j}{\overline{a}}_{k}),
\label{eq: KT sconf vector}
\end{align}
is an $N = 1$ superconformal vector of ${V}^{{\kappa}_{\mathfrak{g}}}_{{\,\,} {\text{s}}}({\mathfrak{g}})$, called the {\textit{Kac--Todorov superconformal vector}} (see also Section 3 of \cite{Mohammedi94A}, equation (3) of \cite{Rhedin96}, and Section 2.2 of \cite{FS97A} for the Lie algebra case, and Example 5.9 of \cite{HK07} for the Lie superalgebra case). If we denote by ${T}_{\mathfrak{g}}$ the conformal vector corresponding to ${G}_{\mathfrak{g}}$ in \eqref{eq: KT sconf vector}, the OPEs with the generators of the super-affine vertex algebra can be computed using the non-commutative Wick formula of vertex algebras and the Jacobi identity for Lie superalgebras as follows:
\begin{gather}
[{T}_{\mathfrak{g}} {\,} {}_{\lambda} {\,} {v}] = ({\partial} + {\lambda}){v},
{\quad}
[{T}_{\mathfrak{g}} {\,} {}_{\lambda} {\,} {\overline{v}}] = ({\partial} + {\frac{1}{2}}{\lambda}){\overline{v}},
\label{eq: OPE T with generator of saff}
\\
[{G}_{\mathfrak{g}} {\,} {}_{\lambda} {\,} {v}] = ({\partial} + {\lambda}){\overline{v}},
{\quad}
[{G}_{\mathfrak{g}} {\,} {}_{\lambda} {\,} {\overline{v}}] = {v},
\label{eq: OPE G with generator of saff}
\end{gather}
for ${v} \in {\mathfrak{g}}$ (see Example 2.13 of \cite{HZ10} for the case where ${\mathfrak{g}}$ is a Lie algebra).

\vskip 1mm

From now on, until the end of Proposition \ref{prop: sec N=3/2 computation of UGU}, we assume that ${\mathfrak{g}}$ admits a super Manin pair structure (see Definition \ref{def: sMP}). Following the notation of Section \ref{section: sMP}, let $\{ {y}_{i} \}_{i \in I}$ be a parity-homogeneous basis of ${\mathfrak{a}}^{+}$ and let $\{ {y}^{i} \}_{i \in I}$ be the basis of ${\mathfrak{a}}^{-}$ satisfying ${\kappa}_{\mathfrak{g}}({y}_{i}, {y}^{j}) = {\delta}_{ij}$ for $i,j \in I$.

\vskip 1mm

In order to express the superconformal vector in \eqref{eq: KT sconf vector} in terms of the basis from the super Manin pair structure, consider the dual bases $\{ {y}_{i}, {y}^{i} \}_{i \in I}$ and $\{ {y}^{i}, (-1)^{i}{y}_{i} \}_{i \in I}$ of ${\mathfrak{g}}$. Then the first term in \eqref{eq: KT sconf vector} can be written as:
\begin{align*}
{\sum}_{i \in I} {\big(}(-1)^{i}{y}^{i}{\overline{y}}_{i} + {y}_{i}{\overline{y}}^{i}{\big)},
\end{align*}
from the quasi-commutativity of vertex algebras. For the second cubic term in \eqref{eq: KT sconf vector}, when $({a}_{i}, {a}_{j},{a}_{k}) = ({y}^{i}, {y}^{j}, {y}^{k})$, we have $({b}_{i}, {b}_{j}, {b}_{k}) = ((-1)^{i}{y}_{i}, (-1)^{j}{y}_{j}, (-1)^{k}{y}_{k})$, so that ${\kappa}_{\mathfrak{g}}({b}_{i}, [{b}_{j}, {b}_{k}]) = 0$. Next, in the case where one of ${a}_{i}, {a}_{j},{a}_{k}$ is ${y}_{i}$ and the other two are of the form ${y}^{j}$, the second term in \eqref{eq: KT sconf vector} is given by:
\begin{align}
{\frac{1}{3}}{\,} {\sum}_{i,j,k \in I} {\big(}(-1)^{ij + i} {C}^{i}_{jk} {\overline{y}}_{i}({\overline{y}}^{j}{\overline{y}}^{k}) + (-1)^{ij + j} {C}^{k}_{ij} {\overline{y}}^{i}({\overline{y}}^{j}{\overline{y}}_{k}) + (-1)^{ij + jk + k + 1} {C}^{j}_{ik} {\overline{y}}^{i}({\overline{y}}_{j}{\overline{y}}^{k}){\big)}.
\label{eq: cubic term1 in KT}
\end{align}
From the non-commutative Wick formula, we know that
\begin{align}
[{\overline{y}}^{i} {\,} {}_{\lambda} {\,} {\overline{y}}^{j}{\overline{y}}_{k}] = (-1)^{(i+1)(j+1)}{\delta}_{ik}{\overline{y}}^{j},
\label{eq: lambd brac for 1bar and 2bar}
\end{align}
for $i, j, k \in I$. Hence, using equation \eqref{eq: lambd brac for 1bar and 2bar}, quasi-commutativity and quasi-associativity, we obtain the following relation:
\begin{align}
{\overline{y}}^{i}({\overline{y}}^{j}{\overline{y}}_{k}) 
=&{\,}(-1)^{(i+1)(j+k)}({\overline{y}}^{j}{\overline{y}}_{k}){\overline{y}}^{i} + (-1)^{(i+1)(j+1)}{\delta}_{ik}{\partial}{\overline{y}}^{j}
\notag
\\
=&{\,}(-1)^{ij + jk + ki + 1}({\overline{y}}_{k}{\overline{y}}^{j}){\overline{y}}^{i} + (-1)^{(i+1)(j+1)}{\delta}_{ik}{\partial}{\overline{y}}^{j}
\label{eq: re-ordering of bar cubic1}
\\
=&{\,}(-1)^{ij + jk + ki + 1}{\overline{y}}_{k}({\overline{y}}^{j}{\overline{y}}^{i}),
\notag
\end{align}
so that we also have
\begin{align}
{\overline{y}}^{i}({\overline{y}}_{j}{\overline{y}}^{k}) 
= (-1)^{(j+1)(k+1)}{\overline{y}}^{i}({\overline{y}}^{k}{\overline{y}}_{j}) 
= (-1)^{ij + ik + j + k}{\overline{y}}_{j}({\overline{y}}^{k}{\overline{y}}^{i}),
\label{eq: re-ordering of bar cubic2}
\end{align}
from the above equality \eqref{eq: re-ordering of bar cubic1}. Therefore, using the relations \eqref{eq: re-ordering of bar cubic1} and \eqref{eq: re-ordering of bar cubic2}, the term in \eqref{eq: cubic term1 in KT} can be simplified as follows:
\begin{align*}
{\sum}_{i,j,k \in I} (-1)^{ki} {C}^{i}_{jk}{\overline{y}}_{i}({\overline{y}}^{j}{\overline{y}}^{k}).
\end{align*}
For the case where two of ${a}_{i}$ are of the form ${y}_{i}$ and one is of the form ${y}^{i}$, we similarly obtain from the non-commutative Wick formula that
\begin{align}
[{\overline{y}}_{i} {\,} {}_{\lambda} {\,} {\overline{y}}^{j}{\overline{y}}_{k}] = (-1)^{i}{\delta}_{ij}{\overline{y}}_{k},
\label{eq: lambd brac for sub1bar and 2bar}
\end{align}
and hence, via quasi-commutativity and quasi-associativity, we know that
\begin{align*}
{\overline{y}}_{i}({\overline{y}}^{j}{\overline{y}}_{k}) = (-1)^{(i+1)(j+k)}{\overline{y}}^{j}({\overline{y}}_{k}{\overline{y}}_{i}),
{\quad}
{\overline{y}}_{i}({\overline{y}}_{j}{\overline{y}}^{k}) = (-1)^{ik + jk + i + j}{\overline{y}}^{k}({\overline{y}}_{i}{\overline{y}}_{j}),
\end{align*}
so that the second term in \eqref{eq: KT sconf vector} becomes
\begin{align*}
{\sum}_{i,j,k \in I} (-1)^{ij} {F}^{jk}_{i}{\overline{y}}^{i}({\overline{y}}_{j}{\overline{y}}_{k}).
\end{align*}
Finally, when all ${a}_{i}$ are of the form ${y}_{i}$, the second term in \eqref{eq: KT sconf vector} is given by
\begin{align*}
{\frac{1}{3}}{\,}{\sum}_{i,j,k \in I} (-1)^{ij}{\mathring{F}}^{jk}_{i}{\overline{y}}_{i}({\overline{y}}_{j}{\overline{y}}_{k}),
\end{align*}
and hence the whole expression of the superconformal vector in \eqref{eq: KT sconf vector}, written in the basis from the super Manin pair structure, is:
\begin{align}
\begin{split}
{G}_{\mathfrak{g}} =& {\,} {\sum}_{i \in I} {\big(}(-1)^{i}{y}^{i}{\overline{y}}_{i} + {y}_{i}{\overline{y}}^{i}{\big)} + {\sum}_{i,j,k \in I} (-1)^{ki} {C}^{i}_{jk}{\overline{y}}_{i}({\overline{y}}^{j}{\overline{y}}^{k})
\\
&+ {\sum}_{i,j,k \in I} (-1)^{ij} {F}^{jk}_{i}{\overline{y}}^{i}({\overline{y}}_{j}{\overline{y}}_{k}) + {\frac{1}{3}}{\,}{\sum}_{i,j,k \in I} (-1)^{ij}{\mathring{F}}^{jk}_{i}{\overline{y}}_{i}({\overline{y}}_{j}{\overline{y}}_{k}),
\end{split}
\label{eq: sec N=3/2 KT sconf sMP basis expression}
\end{align}
with central charge ${\frac{3}{2}}{\,}{\text{sdim}}{\,}{\mathfrak{g}} + {3}{\,}{\text{str}}({\mathcal{D}}_{\mathfrak{g}})$.

\vskip 1mm

Moreover, we introduce a $U(1)$-current as follows:
\begin{align}
{U}_{\mathfrak{g}} := {\sum}_{i \in I} (-1)^{i}{\overline{y}}^{i}{\overline{y}}_{i} + {\rho},
\label{eq: def of U_g in saff}
\end{align}
where ${\rho}$ is defined in \eqref{eq: def of the element rho}. Then, the OPEs of ${U}_{\mathfrak{g}}$ with the conformal vector ${T}_{\mathfrak{g}}$ and the superconformal vector ${G}_{\mathfrak{g}}$ are given in the following lemma.

\vskip 1mm

\begin{lemma}\label{lem: sec N=3/2 OPE of T and U, G and U}\rm
In the super-affine vertex algebra ${V}^{{\kappa}_{\mathfrak{g}}}_{{\,\,} {\text{s}}}({\mathfrak{g}})$, the following OPEs hold:
\begin{align*}
[{T}_{\mathfrak{g}} {\,} {}_{\lambda} {\,} {U}_{\mathfrak{g}}] = ({\partial} + {\lambda}){U}_{\mathfrak{g}},
{\quad}
[{G}_{\mathfrak{g}} {\,} {}_{\lambda} {\,} {U}_{\mathfrak{g}}] = {\sum}_{i \in I} (-1)^{i}{y}^{i}{\overline{y}}_{i} - {\sum}_{i \in I} {y}_{i}{\overline{y}}^{i} + {2}{\partial}{\overline{\rho}}.
\end{align*}
\end{lemma}

\begin{proof}
Note that the first OPE follows from the non-commutative Wick formula together with \eqref{eq: OPE T with generator of saff}. For the second OPE, from the non-commutative Wick formula and \eqref{eq: OPE G with generator of saff}, we obtain
\begin{align*}
[{G}_{\mathfrak{g}} {\,} {}_{\lambda} {\,} {\overline{y}}^{i}{\overline{y}}_{i}] = {y}^{i}{\overline{y}}_{i} + (-1)^{i+1}{\overline{y}}^{i}{y}_{i} - {\lambda}{\overline{[{y}_{i}, {y}^{i}]}},
\end{align*}
so that we have the following equalities:
\begin{align*}
[{G}_{\mathfrak{g}} {\,} {}_{\lambda} {\,} {U}_{\mathfrak{g}}] 
&= [{G}_{\mathfrak{g}} {\,} {}_{\lambda} {\,} {\sum}_{i \in I} (-1)^{i}{\overline{y}}^{i}{\overline{y}}_{i}] + [{G}_{\mathfrak{g}} {\,} {}_{\lambda} {\,} {\rho}]
\\
&= {\sum}_{i \in I} (-1)^{i}{y}^{i}{\overline{y}}_{i} - {\sum}_{i \in I} {\overline{y}}^{i}{y}_{i} - {\lambda}{\overline{\rho}} + ({\partial} + {\lambda}){\overline{\rho}}
\\
&= {\sum}_{i \in I} (-1)^{i}{y}^{i}{\overline{y}}_{i} - {\sum}_{i \in I} {\overline{y}}^{i}{y}_{i} + {\partial}{\overline{\rho}}
\\
&= {\sum}_{i \in I} (-1)^{i}{y}^{i}{\overline{y}}_{i} - {\sum}_{i \in I} {y}_{i}{\overline{y}}^{i} + {2}{\partial}{\overline{\rho}},
\end{align*}
where the last equality is obtained from quasi-commutativity.
\end{proof}

\vskip 1mm

We now list the OPEs in the super-affine vertex algebra required to compute ${U}_{{\mathfrak{g}} {\,} (0)}({G}_{{\mathfrak{g}} {\,} (0)}{U}_{\mathfrak{g}})$. The definitions of ${\mathcal{D}}^{j}_{i}$ and ${\mathcal{E}}^{ij}$ are given in \eqref{eq: def of map D and D^ij} and \eqref{eq: def of E^ij}, respectively.

\vskip 1mm

\begin{lemma}\label{lem: key OPE lem1 in sec N=3/2}\rm
In ${V}^{{\kappa}_{\mathfrak{g}}}_{{\,\,} {\text{s}}}({\mathfrak{g}})$, the following OPEs hold:
\begin{align*}
[{y}_{r} {\,} {}_{\lambda} {\,} {\sum}_{i \in I} (-1)^{i}{\overline{y}}^{i}{\overline{y}}_{i}] &= {\sum}_{i, j \in I} (-1)^{j}{F}^{ij}_{r}{\overline{y}}_{j}{\overline{y}}_{i} + {\lambda}{\,}{\sum}_{i \in I} (-1)^{i}{C}^{i}_{ir},
\\
[{y}^{r} {\,} {}_{\lambda} {\,} {\sum}_{i \in I} (-1)^{i}{\overline{y}}^{i}{\overline{y}}_{i}] &= {\sum}_{i, j \in I} (-1)^{i}{\mathring{F}}^{ij}_{r}{\overline{y}}_{j}{\overline{y}}_{i} + {\sum}_{i, j \in I} (-1)^{i+1}{C}^{r}_{ji}{\overline{y}}^{i}{\overline{y}}^{j} + {\lambda}{\,}{\sum}_{i \in I} (-1)^{i}{F}^{ri}_{i},
\\
[{\overline{y}}_{r} {\,} {}_{\lambda} {\,} {\sum}_{i \in I} (-1)^{i}{\overline{y}}^{i}{\overline{y}}_{i}] &= {\overline{y}}_{r},
\\
[{\overline{y}}^{r} {\,} {}_{\lambda} {\,} {\sum}_{i \in I} (-1)^{i}{\overline{y}}^{i}{\overline{y}}_{i}] &= -{\overline{y}}^{r},
\\
[{y}_{r} {\,} {}_{\lambda} {\,} {\rho}] &= {\sum}_{k \in I} {\mathcal{D}}^{k}_{r}{y}_{k} + {\lambda}{\,}{\sum}_{i \in I} (-1)^{i}{C}^{i}_{ri},
\\
[{y}^{r} {\,} {}_{\lambda} {\,} {\rho}] &= - {\sum}_{k \in I} {\mathcal{D}}^{r}_{k}{y}^{k} - {\sum}_{k \in I} {\mathcal{E}}^{rk}{y}_{k} + {\lambda}{\,}{\sum}_{i \in I} (-1)^{i}{F}^{ir}_{i},
\\
[{\overline{y}}_{r} {\,} {}_{\lambda} {\,} {\rho}] &= {\sum}_{k \in I} {\mathcal{D}}^{k}_{r}{\overline{y}}_{k},
\\
[{\overline{y}}^{r} {\,} {}_{\lambda} {\,} {\rho}] &= - {\sum}_{k \in I} {\mathcal{D}}^{r}_{k}{\overline{y}}^{k} - {\sum}_{k \in I} {\mathcal{E}}^{rk}{\overline{y}}_{k},
\end{align*}
for ${r} \in {I}$.
\end{lemma}

\begin{proof}
From the non-commutative Wick formula, we have:
\begin{align*}
[{y}_{r} {\,} {}_{\lambda} {\,} {\sum}_{i \in I} (-1)^{i}{\overline{y}}^{i}{\overline{y}}_{i}] 
=& {\,} {\sum}_{i \in I} (-1)^{i}[{y}_{r} {\,} {}_{\lambda} {\,} {\overline{y}}^{i}]{\overline{y}}_{i} 
+ {\sum}_{i \in I} (-1)^{ir + i + r}{\overline{y}}^{i}[{y}_{r} {\,} {}_{\lambda} {\,} {\overline{y}}_{i}]
\\
&+ {\sum}_{i \in I} (-1)^{i} \int^{\lambda}_{0} [[{y}_{r} {\,} {}_{\lambda} {\,} {\overline{y}}^{i}] {\,} {}_{\gamma} {\,} {\overline{y}}_{i}] {\,} {d}{\gamma}
\\
=& {\,} {\sum}_{i \in I} (-1)^{i + r}{\overline{[{y}_{r}, {y}^{i}]}}{\overline{y}}_{i} 
+ {\sum}_{i \in I} (-1)^{ir + i}{\overline{y}}^{i}{\overline{[{y}_{r}, {y}_{i}]}}
\\
&+ {\sum}_{i \in I} (-1)^{i + r} \int^{\lambda}_{0} [{\overline{[{y}_{r}, {y}^{i}]}} {\,} {}_{\gamma} {\,} {\overline{y}}_{i}] {\,} {d}{\gamma}
\\
=& {\,} {\sum}_{i \in I} (-1)^{i + r} {\big(}{\sum}_{j \in I} {F}^{ij}_{r}{\overline{y}}_{j} + {\sum}_{j \in I} {C}^{i}_{jr}{\overline{y}}^{j}{\big)}{\overline{y}}_{i}
\\
&+ {\sum}_{i \in I} (-1)^{ir + i}{\overline{y}}^{i} {\big(}{\sum}_{j \in I} {C}^{j}_{ri}{\overline{y}}_{j}{\big)} 
+ {\sum}_{i \in I} \int^{\lambda}_{0} {\kappa}_{\mathfrak{g}}([{y}_{r}, {y}^{i}], {y}_{i}) {\,} {d}{\gamma}
\\
=& {\,} {\sum}_{i, j \in I} (-1)^{j}{F}^{ij}_{r}{\overline{y}}_{j}{\overline{y}}_{i} + {\lambda}{\,}{\sum}_{i \in I} (-1)^{i}{C}^{i}_{ir},
\end{align*}
so that the first OPE follows. The second OPE follows from a similar computation using the non-commutative Wick formula and the equality in \eqref{eq: cyclic str const MP}, while the third and fourth OPEs follow from \eqref{eq: lambd brac for sub1bar and 2bar} and \eqref{eq: lambd brac for 1bar and 2bar}, respectively.

\vskip 1mm

The fifth OPE follows from a direct computation together with \eqref{eq: JI 1 modified for MP} and \eqref{eq: JI 3 modified for MP}, and similarly, the sixth OPE is obtained from the following computation:
\begin{align*}
[{y}^{r} {\,} {}_{\lambda} {\,} {\rho}] 
=& {\,} {\sum}_{i, j \in I} (-1)^{i}{F}^{ij}_{i} [{y}^{r} {\,} {}_{\lambda} {\,} {y}_{j}] + {\sum}_{i, j \in I} (-1)^{i}{C}^{i}_{ji} [{y}^{r} {\,} {}_{\lambda} {\,} {y}^{j}]
\\
=& {\,} {\sum}_{i, j \in I} (-1)^{i + 1}{F}^{ij}_{i} {\big(}{\sum}_{k \in I} {F}^{rk}_{j}{y}_{k} + {\sum}_{k \in I} {C}^{r}_{kj}{y}^{k}{\big)} + {\lambda}{\,}{\sum}_{i \in I} (-1)^{i}{F}^{ir}_{i}
\\
&+ {\sum}_{i, j \in I} (-1)^{i}{C}^{i}_{ji} {\big(}{\sum}_{k \in I} {F}^{rj}_{k}{y}^{k} + {\sum}_{k \in I} {\mathring{F}}^{rj}_{k}{y}_{k}{\big)}
\\
=& {\,} {\sum}_{k \in I} {\big(} {\sum}_{i, j \in I} {\big(}(-1)^{i}{C}^{i}_{ji}{F}^{rj}_{k} + (-1)^{i + 1}{C}^{r}_{kj}{F}^{ij}_{i}{\big)}{\big)}{y}^{k}
\\
&+ {\sum}_{k \in I} {\big(} {\sum}_{i, j \in I} {\big(}(-1)^{i + 1}{F}^{ij}_{i}{F}^{rk}_{j} + (-1)^{i}{C}^{i}_{ji}{\mathring{F}}^{rj}_{k}{\big)}{\big)}{y}_{k} + {\lambda}{\,}{\sum}_{i \in I} (-1)^{i}{F}^{ir}_{i}
\\
=& {\,} {\sum}_{k \in I} {\big(} {\sum}_{i, j \in I} {\big(}(-1)^{i + 1}{C}^{i}_{ji}{F}^{jr}_{k} + (-1)^{i + 1}{C}^{r}_{jk}{F}^{ji}_{i}{\big)}{\big)}{y}^{k}
\\
&+ {\sum}_{k \in I} {\big(} {\sum}_{i, j \in I} {\big(}(-1)^{i + 1}{F}^{ij}_{i}{F}^{rk}_{j} + (-1)^{i + 1}{C}^{i}_{ji}{\mathring{F}}^{rk}_{j}{\big)}{\big)}{y}_{k} + {\lambda}{\,}{\sum}_{i \in I} (-1)^{i}{F}^{ir}_{i}
\\
=& {\,} - {\sum}_{k \in I} {\mathcal{D}}^{r}_{k}{y}^{k} - {\sum}_{k \in I} {\mathcal{E}}^{rk}{y}_{k} + {\lambda}{\,}{\sum}_{i \in I} (-1)^{i}{F}^{ir}_{i},
\end{align*}
where the last equality is a consequence of \eqref{eq: JI 3 modified for MP}. Moreover, the seventh and eighth OPEs follow directly from the fifth and sixth OPEs, respectively.
\end{proof}

\vskip 1mm

From the above Lemma \ref{lem: key OPE lem1 in sec N=3/2}, we can show that ${U}_{\mathfrak{g}}$, defined in \eqref{eq: def of U_g in saff}, is indeed a $U(1)$-current of the super-affine vertex algebra associated with a Lie superalgebra admitting a super Manin pair structure, as follows.

\vskip 1mm

\begin{lemma}\label{lem: sec N=3/2 OPE of U and U}\rm
In ${V}^{{\kappa}_{\mathfrak{g}}}_{{\,\,} {\text{s}}}({\mathfrak{g}})$, the OPE of ${U}_{\mathfrak{g}}$ with itself is given by:
\begin{align*}
[{U}_{\mathfrak{g}} {\,} {}_{\lambda} {\,} {U}_{\mathfrak{g}}] = {\frac{1}{3}}{\big(}{\frac{3}{2}}{\,}{\text{sdim}}{\,}{\mathfrak{g}} + {3}{\,}{\text{str}}({\mathcal{D}}_{\mathfrak{g}}){\big)}{\lambda}.
\end{align*}
\end{lemma}

\begin{proof}
From Lemma \ref{lem: key OPE lem1 in sec N=3/2}, we know that:
\begin{align*}
[{\sum}_{i \in I} (-1)^{i}{\overline{y}}^{i}{\overline{y}}_{i} {\,} {}_{\lambda} {\,} {\sum}_{r \in I} (-1)^{r}{\overline{y}}^{r}{\overline{y}}_{r}] 
=& {\,} {\sum}_{r \in I} (-1)^{r} [{\sum}_{i \in I} (-1)^{i}{\overline{y}}^{i}{\overline{y}}_{i} {\,} {}_{\lambda} {\,} {\overline{y}}^{r}]{\overline{y}}_{r}
\\
&+ {\sum}_{r \in I} (-1)^{r}{\overline{y}}^{r} [{\sum}_{i \in I} (-1)^{i}{\overline{y}}^{i}{\overline{y}}_{i} {\,} {}_{\lambda} {\,} {\overline{y}}_{r}]
\\
&+ {\sum}_{r \in I} (-1)^{r} \int^{\lambda}_{0} [[{\sum}_{i \in I} (-1)^{i}{\overline{y}}^{i}{\overline{y}}_{i} {\,} {}_{\lambda} {\,} {\overline{y}}^{r}] {\,} {}_{\gamma} {\,} {\overline{y}}_{r}] {\,} {d}{\gamma}
\\
=& {\,} {\sum}_{r \in I} (-1)^{r} {\overline{y}}^{r}{\overline{y}}_{r} 
+ {\sum}_{r \in I} (-1)^{r + 1} {\overline{y}}^{r}{\overline{y}}_{r} 
\\
&+ {\sum}_{r \in I} (-1)^{r} \int^{\lambda}_{0} [{\overline{y}}^{r} {\,} {}_{\gamma} {\,} {\overline{y}}_{r}] {\,} {d}{\gamma}
\\
=& {\,} {\lambda}{\,}{\sum}_{r \in I} (-1)^{r} = {\big(}{\frac{1}{2}}{\,}{\text{sdim}}{\,}{\mathfrak{g}}{\big)}{\lambda},
\end{align*}
and we also have:
\begin{align*}
[{\rho} {\,} {}_{\lambda} {\,} {\sum}_{r \in I} (-1)^{r}{\overline{y}}^{r}{\overline{y}}_{r}] 
=& {\,} {\sum}_{r \in I} (-1)^{r} [{\rho} {\,} {}_{\lambda} {\,} {\overline{y}}^{r}]{\overline{y}}_{r} 
+ {\sum}_{r \in I} (-1)^{r}{\overline{y}}^{r} [{\rho} {\,} {}_{\lambda} {\,} {\overline{y}}_{r}]
\\
&+ {\sum}_{r \in I} (-1)^{r} \int^{\lambda}_{0} [[{\rho} {\,} {}_{\lambda} {\,} {\overline{y}}^{r}] {\,} {}_{\gamma} {\,} {\overline{y}}_{r}] {\,} {d}{\gamma}
\\
=& {\,} {\sum}_{r \in I} (-1)^{r} {\big(}{\sum}_{k \in I} {\mathcal{D}}^{r}_{k}{\overline{y}}^{k} + {\sum}_{k \in I} {\mathcal{E}}^{rk}{\overline{y}}_{k}{\big)}{\overline{y}}_{r}
\\
&+ {\sum}_{r \in I} (-1)^{r + 1}{\overline{y}}^{r} {\big(}{\sum}_{k \in I} {\mathcal{D}}^{k}_{r}{\overline{y}}_{k}{\big)}
\\
&+ {\sum}_{r \in I} (-1)^{r} \int^{\lambda}_{0} [{\sum}_{k \in I} {\mathcal{D}}^{r}_{k}{\overline{y}}^{k} + {\sum}_{k \in I} {\mathcal{E}}^{rk}{\overline{y}}_{k} {\,} {}_{\gamma} {\,} {\overline{y}}_{r}] {\,} {d}{\gamma}
\\
=& {\,} {\sum}_{k, r \in I} (-1)^{r}{\mathcal{E}}^{rk}{\overline{y}}_{k}{\overline{y}}_{r} 
+ {\lambda}{\,}{\sum}_{r \in I} (-1)^{r}{\mathcal{D}}^{r}_{r},
\end{align*}
so that
\begin{align*}
[{\sum}_{r \in I} (-1)^{r}{\overline{y}}^{r}{\overline{y}}_{r} {\,} {}_{\lambda} {\,} {\rho}] = - {\sum}_{k, r \in I} (-1)^{r}{\mathcal{E}}^{rk}{\overline{y}}_{k}{\overline{y}}_{r} 
+ {\lambda}{\,}{\sum}_{r \in I} (-1)^{r}{\mathcal{D}}^{r}_{r},
\end{align*}
from the skew-symmetry of vertex algebras. Moreover, by \eqref{eq: def of map D and D^ij}, we have:
\begin{align*}
[{\rho} {\,} {}_{\lambda} {\,} {\rho}] = [{\rho}, {\rho}] + {\kappa}_{\mathfrak{g}}({\rho}, {\rho}){\lambda} = - {\text{str}}({\mathcal{D}}_{\mathfrak{g}}){\lambda},
\end{align*}
hence, the OPE of ${U}_{\mathfrak{g}}$ with itself is as follows:
\begin{align*}
[{U}_{\mathfrak{g}} {\,} {}_{\lambda} {\,} {U}_{\mathfrak{g}}] = {\frac{1}{3}}{\big(}{\frac{3}{2}}{\,}{\text{sdim}}{\,}{\mathfrak{g}} + {3}{\,}{\text{str}}({\mathcal{D}}_{\mathfrak{g}}){\big)}{\lambda}.
\end{align*}
\end{proof}

\vskip 1mm

Let us also record a result that follows immediately from Lemma \ref{lem: key OPE lem1 in sec N=3/2}, which will be needed later in Section \ref{section: N=2 sconf of SUSY semi inf cohom} to verify that a $U(1)$-current on the supersymmetric semi-infinite complex is indeed BRST-invariant.

\vskip 1mm

\begin{lemma}\label{lem: sec N=3/2 OPE generator and U}\rm
In ${V}^{{\kappa}_{\mathfrak{g}}}_{{\,\,} {\text{s}}}({\mathfrak{g}})$, the following OPEs hold:
\begin{align*}
[{y}_{r} {\,} {}_{\lambda} {\,} {U}_{\mathfrak{g}}] &= {\sum}_{k \in I} {\mathcal{D}}^{k}_{r}{y}_{k} + {\sum}_{i, j \in I} (-1)^{j}{F}^{ij}_{r}{\overline{y}}_{j}{\overline{y}}_{i},
\\
[{y}^{r} {\,} {}_{\lambda} {\,} {U}_{\mathfrak{g}}] &= - {\sum}_{k \in I} {\mathcal{D}}^{r}_{k}{y}^{k} - {\sum}_{k \in I} {\mathcal{E}}^{rk}{y}_{k} + {\sum}_{i, j \in I} (-1)^{i}{\mathring{F}}^{ij}_{r}{\overline{y}}_{j}{\overline{y}}_{i} + {\sum}_{i, j \in I} (-1)^{i+1}{C}^{r}_{ji}{\overline{y}}^{i}{\overline{y}}^{j},
\\
[{\overline{y}}_{r} {\,} {}_{\lambda} {\,} {U}_{\mathfrak{g}}] &= {\overline{y}}_{r} + {\sum}_{k \in I} {\mathcal{D}}^{k}_{r}{\overline{y}}_{k},
\\
[{\overline{y}}^{r} {\,} {}_{\lambda} {\,} {U}_{\mathfrak{g}}] &= - {\overline{y}}^{r} - {\sum}_{k \in I} {\mathcal{D}}^{r}_{k}{\overline{y}}^{k} - {\sum}_{k \in I} {\mathcal{E}}^{rk}{\overline{y}}_{k},
\end{align*}
for ${r} \in {I}$.
\end{lemma}

\begin{proof}
It follows directly from Lemma \ref{lem: key OPE lem1 in sec N=3/2}.
\end{proof}

\vskip 1mm

Therefore, by Lemma \ref{lem: sec N=3/2 OPE of T and U, G and U} and Lemma \ref{lem: sec N=3/2 OPE of U and U}, we have verified that ${G}_{\mathfrak{g}}$ and ${U}_{\mathfrak{g}}$ satisfy all the conditions of Lemma \ref{lem: equiv def of N=2 sconf} except for the last relation in \eqref{eq: equiv def of N=2 sconf}. The remaining computation of ${U}_{{\mathfrak{g}} {\,} (0)}({G}_{{\mathfrak{g}} {\,} (0)}{U}_{\mathfrak{g}})$ is as follows:

\vskip 1mm

\begin{proposition}\label{prop: sec N=3/2 computation of UGU}\rm
In ${V}^{{\kappa}_{\mathfrak{g}}}_{{\,\,} {\text{s}}}({\mathfrak{g}})$, the OPE of ${U}_{\mathfrak{g}}$ and ${G}_{{\mathfrak{g}} {\,} (0)}{U}_{\mathfrak{g}}$ is given by
\begin{align*}
[{U}_{\mathfrak{g}} {\,} {}_{\lambda} {\,} {G}_{{\mathfrak{g}} {\,} (0)}{U}_{\mathfrak{g}}] = - {G}_{\mathfrak{g}} 
+ {\frac{4}{3}}{\,}{\sum}_{i, j, k \in I} (-1)^{ij}{\mathring{F}}^{jk}_{i}{\overline{y}}_{i}({\overline{y}}_{j}{\overline{y}}_{k}) 
- {2}{\,} {\sum}_{i, j \in I} {\mathcal{E}}^{ij}{y}_{i}{\overline{y}}_{j}.
\end{align*}
\end{proposition}

\begin{proof}
From Lemma \ref{lem: sec N=3/2 OPE generator and U}, we know that:
\begin{align*}
[{U}_{\mathfrak{g}} {\,} {}_{\lambda} {\,} {\sum}_{r \in I} (-1)^{r}{y}^{r}{\overline{y}}_{r}] 
=& {\,} {\sum}_{r \in I} (-1)^{r} [{U}_{\mathfrak{g}} {\,} {}_{\lambda} {\,} {y}^{r}]{\overline{y}}_{r} 
+ {\sum}_{r \in I} (-1)^{r} {y}^{r}[{U}_{\mathfrak{g}} {\,} {}_{\lambda} {\,} {\overline{y}}_{r}] 
\\
&+ {\sum}_{r \in I} (-1)^{r} \int^{\lambda}_{0} [[{U}_{\mathfrak{g}} {\,} {}_{\lambda} {\,} {y}^{r}] {\,} {}_{\gamma} {\,} {\overline{y}}_{r}] {\,} {d}{\gamma}
\\
=& {\,} {\sum}_{r \in I} (-1)^{r + 1}{y}^{r}{\overline{y}}_{r} 
+ {\sum}_{k, r \in I} (-1)^{r}{\mathcal{E}}^{rk}{y}_{k}{\overline{y}}_{r}
\\
&+ {\sum}_{i, j, r \in I} (-1)^{j + 1}{\mathring{F}}^{ij}_{r}({\overline{y}}_{j}{\overline{y}}_{i}){\overline{y}}_{r} 
+ {\sum}_{i, j, r \in I} (-1)^{j}{C}^{r}_{ji}({\overline{y}}^{i}{\overline{y}}^{j}){\overline{y}}_{r}
\\
&+ {\lambda}{\,} {\sum}_{k, r \in I} (-1)^{k + 1}{\mathcal{D}}^{r}_{k}{\overline{[{y}_{r}, {y}^{k}]}} 
- 2{\lambda}{\,} {\sum}_{i, r \in I} (-1)^{r}{C}^{r}_{ir}{\overline{y}}^{i}
\\
&+ {\lambda}{\,} {\sum}_{s \in I} {\big(}{\sum}_{k, r \in I} (-1)^{k + 1}{\mathcal{E}}^{kr}{C}^{s}_{kr}{\big)}{\overline{y}}_{s}.
\end{align*}
Moreover, since the following equalities hold:
\begin{align*}
{\sum}_{k, r \in I}& (-1)^{k}{\mathcal{E}}^{kr}{C}^{s}_{kr} 
\\
=& {\,} {\sum}_{k, r, m, n \in I} (-1)^{mn + k} {\big(}{C}^{k}_{mn}{\mathring{F}}^{mn}_{r} - (-1)^{m + n}{C}^{r}_{mn}{\mathring{F}}^{mn}_{k}{\big)}{C}^{s}_{kr}
\\
=& {\,} 2{\,}{\sum}_{k, r, m, n \in I} (-1)^{mn + k}{C}^{s}_{kr}{C}^{k}_{mn}{\mathring{F}}^{mn}_{r}
\\
=& {\,} {\frac{2}{3}}{\,}{\sum}_{k, r, m, n \in I} (-1)^{mn + k}{C}^{s}_{kr}{C}^{k}_{mn}{\mathring{F}}^{mn}_{r} 
+ {\frac{2}{3}}{\,}{\sum}_{k, r, m, n \in I} (-1)^{nr + k} {C}^{s}_{km}{C}^{k}_{nr}{\mathring{F}}^{nr}_{m}
\\
&+ {\frac{2}{3}}{\,}{\sum}_{k, r, m, n \in I} (-1)^{mr + k} {C}^{s}_{kn}{C}^{k}_{rm}{\mathring{F}}^{rm}_{n}
\\
=& {\,} {\frac{2}{3}}{\,}{\sum}_{k, r, m, n \in I} {\big(}(-1)^{mn + r}{C}^{s}_{kr}{C}^{k}_{mn} + (-1)^{nr}{C}^{s}_{km}{C}^{k}_{nr} + (-1)^{mr + n + r}{C}^{s}_{kn}{C}^{k}_{rm}{\big)}{\mathring{F}}^{mn}_{r}
\\
=& {\,} 0,
\end{align*}
where the first equality follows from \eqref{eq: sec sMP equality for map E}, and the equality before the last and the final one follow from \eqref{eq: cyclic str const MP} and \eqref{eq: JI 1 for MP}, respectively, we have:
\begin{align*}
[{U}_{\mathfrak{g}} {\,} {}_{\lambda} {\,} {\sum}_{r \in I} (-1)^{r}{y}^{r}{\overline{y}}_{r}] 
=& {\,} {\sum}_{r \in I} (-1)^{r + 1}{y}^{r}{\overline{y}}_{r} 
+ {\sum}_{k, r \in I} (-1)^{r}{\mathcal{E}}^{rk}{y}_{k}{\overline{y}}_{r}
\\
&+ {\sum}_{i, j, r \in I} (-1)^{j + 1}{\mathring{F}}^{ij}_{r}({\overline{y}}_{j}{\overline{y}}_{i}){\overline{y}}_{r} 
+ {\sum}_{i, j, r \in I} (-1)^{j}{C}^{r}_{ji}({\overline{y}}^{i}{\overline{y}}^{j}){\overline{y}}_{r}
\\
&+ {\lambda}{\,} {\sum}_{k, r \in I} (-1)^{k + 1}{\mathcal{D}}^{r}_{k}{\overline{[{y}_{r}, {y}^{k}]}} 
- 2{\lambda}{\,} {\sum}_{i, r \in I} (-1)^{r}{C}^{r}_{ir}{\overline{y}}^{i}.
\end{align*}

\vskip 1mm

Likewise, by Lemma \ref{lem: sec N=3/2 OPE generator and U}, we have:
\begin{align*}
[{U}_{\mathfrak{g}} {\,} {}_{\lambda} {\,} {\sum}_{r \in I} {y}_{r}{\overline{y}}^{r}] 
=& {\,} {\sum}_{r \in I} {y}_{r}{\overline{y}}^{r} + {\sum}_{k, r \in I} {\mathcal{E}}^{rk}{y}_{r}{\overline{y}}_{k} + {\sum}_{i, j, r \in I} (-1)^{j + 1}{F}^{ij}_{r}({\overline{y}}_{j}{\overline{y}}_{i}){\overline{y}}^{r}
\\
&+ {\lambda}{\,} {\sum}_{k, r \in I} (-1)^{k + 1}{\mathcal{D}}^{k}_{r}{\overline{[{y}_{k}, {y}^{r}]}} 
- {2}{\lambda}{\,} {\sum}_{i, r \in I} (-1)^{r}{F}^{ri}_{r}{\overline{y}}_{i},
\end{align*}
and by Lemma \ref{lem: key OPE lem1 in sec N=3/2}, together with the fact that $[{\rho} {\,} {}_{\lambda} {\,} {\overline{\rho}}] = 0$, we also have:
\begin{align*}
[{U}_{\mathfrak{g}} {\,} {}_{\lambda} {\,} {2}{\partial}{\overline{\rho}}] 
= {2}({\partial} + {\lambda}){\,}{\sum}_{r, s \in I} (-1)^{s + 1}{F}^{sr}_{s}{\overline{y}}_{r} + {2}({\partial} + {\lambda}){\,}{\sum}_{r, s \in I} (-1)^{s}{C}^{s}_{rs}{\overline{y}}^{r}.
\end{align*}

\vskip 1mm

Therefore, we have the following computation of OPE:
\begin{align*}
[{U}_{\mathfrak{g}} {\,} {}_{\lambda} {\,} {G}_{{\mathfrak{g}} {\,} (0)}{U}_{\mathfrak{g}}] 
=& {\,} {\sum}_{r \in I} (-1)^{r + 1}{y}^{r}{\overline{y}}_{r} 
+ {\sum}_{k, r \in I} (-1)^{r}{\mathcal{E}}^{rk}{y}_{k}{\overline{y}}_{r}
\\
&+ {\sum}_{i, j, r \in I} (-1)^{j + 1}{\mathring{F}}^{ij}_{r}({\overline{y}}_{j}{\overline{y}}_{i}){\overline{y}}_{r} 
+ {\sum}_{i, j, r \in I} (-1)^{j}{C}^{r}_{ji}({\overline{y}}^{i}{\overline{y}}^{j}){\overline{y}}_{r}
\\
&+ {\lambda}{\,} {\sum}_{k, r \in I} (-1)^{k + 1}{\mathcal{D}}^{r}_{k}{\overline{[{y}_{r}, {y}^{k}]}} 
- 2{\lambda}{\,} {\sum}_{i, r \in I} (-1)^{r}{C}^{r}_{ir}{\overline{y}}^{i}
\\
&- {\sum}_{r \in I} {y}_{r}{\overline{y}}^{r} - {\sum}_{k, r \in I} {\mathcal{E}}^{rk}{y}_{r}{\overline{y}}_{k} - {\sum}_{i, j, r \in I} (-1)^{j + 1}{F}^{ij}_{r}({\overline{y}}_{j}{\overline{y}}_{i}){\overline{y}}^{r}
\\
&- {\lambda}{\,} {\sum}_{k, r \in I} (-1)^{k + 1}{\mathcal{D}}^{k}_{r}{\overline{[{y}_{k}, {y}^{r}]}} 
+ {2}{\lambda}{\,} {\sum}_{i, r \in I} (-1)^{r}{F}^{ri}_{r}{\overline{y}}_{i}
\\
&+ {2}({\partial} + {\lambda}){\,}{\sum}_{r, s \in I} (-1)^{s + 1}{F}^{sr}_{s}{\overline{y}}_{r} + {2}({\partial} + {\lambda}){\,}{\sum}_{r, s \in I} (-1)^{s}{C}^{s}_{rs}{\overline{y}}^{r}
\\
=& {\,} - {\sum}_{r \in I} (-1)^{r}{y}^{r}{\overline{y}}_{r} - {\sum}_{r \in I} {y}_{r}{\overline{y}}^{r} + {2}{\,} {\sum}_{k, r \in I} (-1)^{r}{\mathcal{E}}^{rk}{y}_{k}{\overline{y}}_{r}
\\
&+ {\sum}_{i, j, r \in I} (-1)^{j}{C}^{r}_{ji}({\overline{y}}^{i}{\overline{y}}^{j}){\overline{y}}_{r} + {\sum}_{i, j, r \in I} (-1)^{j}{F}^{ij}_{r}({\overline{y}}_{j}{\overline{y}}_{i}){\overline{y}}^{r}
\\
&+ {\sum}_{i, j, r \in I} (-1)^{j + 1}{\mathring{F}}^{ij}_{r}({\overline{y}}_{j}{\overline{y}}_{i}){\overline{y}}_{r}
\\
&+ {2}{\,}{\sum}_{r, s \in I} (-1)^{s + 1}{F}^{sr}_{s}{\partial}{\overline{y}}_{r} + {2}{\,}{\sum}_{r, s \in I} (-1)^{s}{C}^{s}_{rs}{\partial}{\overline{y}}^{r},
\end{align*}
since ${G}_{{\mathfrak{g}} {\,} (0)}{U}_{\mathfrak{g}} = {\sum}_{i \in I} (-1)^{i}{y}^{i}{\overline{y}}_{i} - {\sum}_{i \in I} {y}_{i}{\overline{y}}^{i} + {2}{\partial}{\overline{\rho}}$ by Lemma \ref{lem: sec N=3/2 OPE of T and U, G and U}.

\vskip 1mm

In order to simplify the result of the above OPE computation, we use the following identities derived from quasi-commutativity:
\begin{align*}
{\overline{y}}_{r}({\overline{y}}^{i}{\overline{y}}^{j}) &= (-1)^{(r+1)(i+j)}({\overline{y}}^{i}{\overline{y}}^{j}){\overline{y}}_{r} + (-1)^{r}{\delta}_{ir}{\partial}{\overline{y}}^{j} + (-1)^{ir + i + 1}{\delta}_{jr}{\partial}{\overline{y}}^{i},
\\
{\overline{y}}^{r}({\overline{y}}_{j}{\overline{y}}_{i}) &= (-1)^{(r+1)(i+j)}({\overline{y}}_{j}{\overline{y}}_{i}){\overline{y}}^{r} + {\delta}_{jr}{\partial}{\overline{y}}_{i} + (-1)^{(r+1)(j+1)}{\delta}_{ir}{\partial}{\overline{y}}_{j},
\\
({\overline{y}}_{j}{\overline{y}}_{i}){\overline{y}}_{r} &= (-1)^{(r+1)(i+j)}{\overline{y}}_{r}({\overline{y}}_{j}{\overline{y}}_{i}),
\end{align*}
from which it follows that the OPE of ${U}_{\mathfrak{g}}$ and ${G}_{{\mathfrak{g}} {\,} (0)}{U}_{\mathfrak{g}}$ is given by:
\begin{align*}
[{U}_{\mathfrak{g}} {\,} {}_{\lambda} {\,} {G}_{{\mathfrak{g}} {\,} (0)}{U}_{\mathfrak{g}}] 
=& {\,} - {\sum}_{r \in I} (-1)^{r}{y}^{r}{\overline{y}}_{r} - {\sum}_{r \in I} {y}_{r}{\overline{y}}^{r} + {2}{\,} {\sum}_{k, r \in I} (-1)^{r}{\mathcal{E}}^{rk}{y}_{k}{\overline{y}}_{r}
\\
& {\,} - {\sum}_{i, j, r \in I} (-1)^{jr}{C}^{r}_{ij}{\overline{y}}_{r}({\overline{y}}^{i}{\overline{y}}^{j}) 
- {\sum}_{i, j, r \in I} (-1)^{ri}{F}^{ij}_{r}{\overline{y}}^{r}({\overline{y}}_{i}{\overline{y}}_{j})
\\
& {\,} + {\sum}_{i, j, r \in I} (-1)^{ri}{\mathring{F}}^{ij}_{r}{\overline{y}}_{r}({\overline{y}}_{i}{\overline{y}}_{j})
\\
=& {\,} - {G}_{\mathfrak{g}} 
+ {\frac{4}{3}}{\,}{\sum}_{i, j, r \in I} (-1)^{ri}{\mathring{F}}^{ij}_{r}{\overline{y}}_{r}({\overline{y}}_{i}{\overline{y}}_{j}) 
+ {2}{\,} {\sum}_{k, r \in I} (-1)^{r}{\mathcal{E}}^{rk}{y}_{k}{\overline{y}}_{r},
\end{align*}
where the last equality is obtained from \eqref{eq: sec N=3/2 KT sconf sMP basis expression}.
\end{proof}

\vskip 1mm

Note that if a Lie superalgebra ${\mathfrak{g}}$ admits a super Manin triple structure, then the structure constants ${\mathring{F}}^{jk}_{i}$ all vanish, and hence ${\mathcal{E}}^{ij}$ also vanish by \eqref{eq: sec sMP equality for map E}. Therefore, the above ${G}_{\mathfrak{g}}$ and ${U}_{\mathfrak{g}}$ give rise to an $N = 2$ superconformal structure of the super-affine vertex algebra, as given in Proposition 6.4.1 of \cite{Creutzig09} and Section 4 of \cite{CRS09}. The $N = 2$ superconformal structures in the case where a Lie algebra ${\mathfrak{g}}$ admits a Manin triple structure can be found in \cite{Parkhomenko92} (see also \cite{Mohammedi94B} and \cite{Figueroa-O'Farrill96}).

\vskip 1mm

From now on, we study the $N = {\frac{3}{2}}$ superconformal structure on the super-affine vertex algebra ${V}^{{\kappa}_{{\mathfrak{g}} \oplus {\mathfrak{h}}}}_{{\,\,} {\text{s}}}({\mathfrak{g}} \oplus {\mathfrak{h}})$ associated to ${\mathfrak{g}} \oplus {\mathfrak{h}}$ admitting the super Manin pair structure given in \eqref{eq: sec sMP decomposition of example sMP2}, under Assumption \ref{assum: second assumption} and the unimodularity of ${\mathfrak{h}}$. First, consider ${G}_{{\mathfrak{g}} \oplus {\mathfrak{h}}}$ and ${U}_{{\mathfrak{g}} \oplus {\mathfrak{h}}}$ given in \eqref{eq: sec N=3/2 KT sconf sMP basis expression} and \eqref{eq: def of U_g in saff}, respectively. Then all the OPEs in \eqref{eq: eq1 in equiv def of N=3/2} hold by Lemma \ref{lem: sec N=3/2 OPE of T and U, G and U}, Lemma \ref{lem: sec N=3/2 OPE of U and U} and Proposition \ref{prop: sec N=3/2 computation of UGU}. Here, by Lemma \ref{lem: sec sMP F and ring F vanish} and Lemma \ref{lem: sec sMP coeff E vanish}, the OPE in Proposition \ref{prop: sec N=3/2 computation of UGU} can be rewritten as follows:
\begin{align*}
[{U}_{{\mathfrak{g}} \oplus {\mathfrak{h}}} {\,} {}_{\lambda} {\,} {G}_{{{\mathfrak{g}} \oplus {\mathfrak{h}}} {\,} (0)}{U}_{{\mathfrak{g}} \oplus {\mathfrak{h}}}] = - {G}_{{\mathfrak{g}} \oplus {\mathfrak{h}}} - {8}{\Phi}_{{\mathfrak{g}} \oplus {\mathfrak{h}}},
\end{align*}
where ${\Phi}_{{\mathfrak{g}} \oplus {\mathfrak{h}}}$ is given by:
\begin{align*}
{\Phi}_{{\mathfrak{g}} \oplus {\mathfrak{h}}} = - {\frac{1}{6}}{\,}{\sum}_{i, j, k \in H} (-1)^{ij}{\mathring{F}}^{jk}_{i}{\overline{y}}_{i}({\overline{y}}_{j}{\overline{y}}_{k}).
\end{align*}
If we now verify the remaining OPEs \eqref{eq: eq2 in equiv def of N=3/2} in Proposition \ref{prop: equiv def of N=3/2 sconf}, we obtain the following theorem.

\vskip 1mm

\begin{theorem}\label{theorem: N=3/2 of saff}\rm
Under Assumption \ref{assum: second assumption}, suppose that ${\mathfrak{h}}$ is unimodular. Then ${G}_{{\mathfrak{g}} \oplus {\mathfrak{h}}}$ and ${U}_{{\mathfrak{g}} \oplus {\mathfrak{h}}}$ form an $N = {\frac{3}{2}}$ superconformal structure on the super-affine vertex algebra ${V}^{{\kappa}_{{\mathfrak{g}} \oplus {\mathfrak{h}}}}_{{\,\,} {\text{s}}}({\mathfrak{g}} \oplus {\mathfrak{h}})$ with central charge ${\frac{3}{2}}{\,}{\text{sdim}}{\,}{\mathfrak{g}} + {\frac{3}{2}}{\,}{\text{sdim}}{\,}{\mathfrak{h}} + {3}{\,}{\text{str}}({\mathcal{D}}_{{\mathfrak{g}} \oplus {\mathfrak{h}}})$.
\end{theorem}

\begin{proof}
First, $[{\Phi}_{{\mathfrak{g}} \oplus {\mathfrak{h}}} {\,} {}_{\lambda} {\,} {\Phi}_{{\mathfrak{g}} \oplus {\mathfrak{h}}}] = 0$ holds trivially. On the other hand, from Lemma \ref{lem: sec sMP F and ring F vanish} and Lemma \ref{lem: sec N=3/2 OPE generator and U}, we have:
\begin{align}
\label{eq: sec N=3/2 U and oveline y}
[{U}_{{\mathfrak{g}} \oplus {\mathfrak{h}}} {\,} {}_{\lambda} {\,} {\overline{y}}_{i}] = - {\overline{y}}_{i},
\end{align}
for $i \in H$, and hence
\begin{align*}
[{U}_{{\mathfrak{g}} \oplus {\mathfrak{h}}} {\,} {}_{\lambda} {\,} {\Phi}_{{\mathfrak{g}} \oplus {\mathfrak{h}}}] = - {3}{\Phi}_{{\mathfrak{g}} \oplus {\mathfrak{h}}},
\end{align*}
also holds. Now, by using \eqref{eq: OPE G with generator of saff}, we have:
\begin{align*}
[{G}_{{\mathfrak{g}} \oplus {\mathfrak{h}}} {\,} {}_{\lambda} {\,} {\Phi}_{{\mathfrak{g}} \oplus {\mathfrak{h}}}] 
=& {\,} - {\frac{1}{6}}{\,}{\sum}_{i, j, k \in H} (-1)^{ij}{\mathring{F}}^{jk}_{i}{\big(}{y}_{i}({\overline{y}}_{j}{\overline{y}}_{k}) + (-1)^{i + 1}{\overline{y}}_{i}({y}_{j}{\overline{y}}_{k}) + (-1)^{i + j}{\overline{y}}_{i}({\overline{y}}_{j}{y}_{k}){\big)}
\\
& {\,} - {\frac{1}{6}}{\lambda}{\,}{\sum}_{i, j, k, r \in H} (-1)^{ij + i + j + 1}{C}^{r}_{jk}{\mathring{F}}^{jk}_{i}{\overline{y}}_{i}{\overline{y}}_{r}
\\
& {\,} - {\frac{1}{6}}{\lambda}{\,}{\sum}_{i, j, k, r \in H} (-1)^{ij + i}{C}^{r}_{ij}{\mathring{F}}^{jk}_{i}{\overline{y}}_{r}{\overline{y}}_{k}
\\
& {\,} - {\frac{1}{6}}{\lambda}{\,}{\sum}_{i, j, k, r \in H} {C}^{r}_{ik}{\mathring{F}}^{jk}_{i}{\overline{y}}_{j}{\overline{y}}_{r}
\\
=& {\,} - {\frac{1}{6}}{\,}{\sum}_{i, j, k \in H} (-1)^{ij}{\mathring{F}}^{jk}_{i}{\big(}{y}_{i}({\overline{y}}_{j}{\overline{y}}_{k}) + (-1)^{i + 1}{\overline{y}}_{i}({y}_{j}{\overline{y}}_{k}) + (-1)^{i + j}{\overline{y}}_{i}({\overline{y}}_{j}{y}_{k}){\big)}
\\
& {\,} - {\frac{1}{2}}{\lambda}{\,}{\sum}_{i, j, k, r \in H} (-1)^{ij + i + j + 1}{C}^{r}_{jk}{\mathring{F}}^{jk}_{i}{\overline{y}}_{i}{\overline{y}}_{r},
\end{align*}
where the last equality follows from quasi-commutativity and \eqref{eq: cyclic str const MP}. However, by exchanging the indices ${i}$ and ${r}$ and using quasi-commutativity, we have:
\begin{align*}
{\sum}_{i, j, k, r \in H} (-1)^{(i + 1)(j + 1)}{C}^{r}_{jk}{\mathring{F}}^{jk}_{i}{\overline{y}}_{i}{\overline{y}}_{r} 
&= {\sum}_{i, j, k, r \in H} (-1)^{(r + 1)(j + 1)}{C}^{i}_{jk}{\mathring{F}}^{jk}_{r}{\overline{y}}_{r}{\overline{y}}_{i}
\\
&= {\sum}_{i, r \in H} {\big(}{\sum}_{j, k \in H}(-1)^{jk}{C}^{i}_{jk}{\mathring{F}}^{jk}_{r}{\big)}{\overline{y}}_{i}{\overline{y}}_{r},
\end{align*}
and hence:
\begin{align*}
- {\frac{1}{2}}{\lambda}{\,}{\sum}_{i, j, k, r \in H} &(-1)^{(i + 1)(j + 1)}{C}^{r}_{jk}{\mathring{F}}^{jk}_{i}{\overline{y}}_{i}{\overline{y}}_{r}
\\
&= - {\frac{1}{4}}{\lambda}{\,}{\sum}_{i, r \in H} {\big(}{\sum}_{j, k \in H} (-1)^{jk}({C}^{i}_{jk}{\mathring{F}}^{jk}_{r} - (-1)^{j + k}{C}^{r}_{jk}{\mathring{F}}^{jk}_{i}){\big)}{\overline{y}}_{i}{\overline{y}}_{r}
\\
&= - {\frac{1}{4}}{\lambda}{\,}{\sum}_{i, r \in H} {\mathcal{E}}^{ir}{\overline{y}}_{i}{\overline{y}}_{r} = 0,
\end{align*}
which follows from Lemma \ref{lem: sec sMP coeff E vanish}. Thus, the result of the computation of $[{G}_{{\mathfrak{g}} \oplus {\mathfrak{h}}} {\,} {}_{\lambda} {\,} {\Phi}_{{\mathfrak{g}} \oplus {\mathfrak{h}}}]$ can be rewritten as follows:
\begin{align*}
[{G}_{{\mathfrak{g}} \oplus {\mathfrak{h}}} {\,} {}_{\lambda} {\,} {\Phi}_{{\mathfrak{g}} \oplus {\mathfrak{h}}}] = {G}_{{{\mathfrak{g}} \oplus {\mathfrak{h}}} {\,} (0)}{\Phi}_{{\mathfrak{g}} \oplus {\mathfrak{h}}},
\end{align*}
where
\begin{align*}
{G}_{{{\mathfrak{g}} \oplus {\mathfrak{h}}} {\,} (0)}{\Phi}_{{\mathfrak{g}} \oplus {\mathfrak{h}}} = - {\frac{1}{6}}{\,}{\sum}_{i, j, k \in H} (-1)^{ij}{\mathring{F}}^{jk}_{i}{\big(}{y}_{i}({\overline{y}}_{j}{\overline{y}}_{k}) + (-1)^{i + 1}{\overline{y}}_{i}({y}_{j}{\overline{y}}_{k}) + (-1)^{i + j}{\overline{y}}_{i}({\overline{y}}_{j}{y}_{k}){\big)}.
\end{align*}

\vskip 1mm

On the other hand, from Lemma \ref{lem: sec sMP F and ring F vanish} and Lemma \ref{lem: sec N=3/2 OPE generator and U}, we have:
\begin{align}
\label{eq: sec N=3/2 U and y}
[{U}_{{\mathfrak{g}} \oplus {\mathfrak{h}}} {\,} {}_{\lambda} {\,} {y}_{i}] = 0,
\end{align}
for $i \in H$. Hence, using \eqref{eq: sec N=3/2 U and oveline y} and \eqref{eq: sec N=3/2 U and y}, we obtain
\begin{align*}
[{U}_{{\mathfrak{g}} \oplus {\mathfrak{h}}} {\,} {}_{\lambda} {\,} {y}_{i}({\overline{y}}_{j}{\overline{y}}_{k})] &= -{2}{y}_{i}({\overline{y}}_{j}{\overline{y}}_{k}),
\\
[{U}_{{\mathfrak{g}} \oplus {\mathfrak{h}}} {\,} {}_{\lambda} {\,} {\overline{y}}_{i}({y}_{j}{\overline{y}}_{k})] &= -{2}{\overline{y}}_{i}({y}_{j}{\overline{y}}_{k}) - {\lambda}{\overline{[{y}_{i}, {y}_{j}]}}{\overline{y}}_{k},
\\
[{U}_{{\mathfrak{g}} \oplus {\mathfrak{h}}} {\,} {}_{\lambda} {\,} {\overline{y}}_{i}({\overline{y}}_{j}{y}_{k})] &= -{2}{\overline{y}}_{i}({\overline{y}}_{j}{y}_{k}) - {\lambda}{\overline{y}}_{i}{\overline{[{y}_{j}, {y}_{k}]}} - (-1)^{(i + 1)(j + 1)}{\lambda}{\overline{y}}_{j}{\overline{[{y}_{i}, {y}_{k}]}},
\end{align*}
for $i, j, k \in H$. Therefore, by combining the above identities, we obtain:
\begin{align*}
[{U}_{{\mathfrak{g}} \oplus {\mathfrak{h}}} {\,} {}_{\lambda} {\,} {G}_{{{\mathfrak{g}} \oplus {\mathfrak{h}}} {\,} (0)}{\Phi}_{{\mathfrak{g}} \oplus {\mathfrak{h}}}] 
=& {\,} {\frac{1}{3}}{\,}{\sum}_{i, j, k \in H} (-1)^{ij}{\mathring{F}}^{jk}_{i}{\big(}{y}_{i}({\overline{y}}_{j}{\overline{y}}_{k}) + (-1)^{i + 1}{\overline{y}}_{i}({y}_{j}{\overline{y}}_{k}) + (-1)^{i + j}{\overline{y}}_{i}({\overline{y}}_{j}{y}_{k}){\big)}
\\
& - {\frac{1}{2}}{\lambda}{\,}{\sum}_{i, j, k, r \in H} (-1)^{ij + i + j + 1}{C}^{r}_{jk}{\mathring{F}}^{jk}_{i}{\overline{y}}_{i}{\overline{y}}_{r}
\\
=& {\,} - {2}{G}_{{{\mathfrak{g}} \oplus {\mathfrak{h}}} {\,} (0)}{\Phi}_{{\mathfrak{g}} \oplus {\mathfrak{h}}}.
\end{align*}
It follows from the above results that
\begin{align*}
[{\Phi}_{{\mathfrak{g}} \oplus {\mathfrak{h}}} {\,} {}_{\lambda} {\,} {G}_{{{\mathfrak{g}} \oplus {\mathfrak{h}}} {\,} (0)}{U}_{{\mathfrak{g}} \oplus {\mathfrak{h}}}] 
&= [{\Phi}_{{\mathfrak{g}} \oplus {\mathfrak{h}}} {\,} {}_{\lambda} {\,} [{G}_{{\mathfrak{g}} \oplus {\mathfrak{h}}} {\,} {}_{\gamma} {\,} {U}_{{\mathfrak{g}} \oplus {\mathfrak{h}}}]] 
\\
&= [[{\Phi}_{{\mathfrak{g}} \oplus {\mathfrak{h}}} {\,} {}_{\lambda} {\,} {G}_{{\mathfrak{g}} \oplus {\mathfrak{h}}}] {\,} {}_{{\lambda} + {\gamma}} {\,} {U}_{{\mathfrak{g}} \oplus {\mathfrak{h}}}] - [{G}_{{\mathfrak{g}} \oplus {\mathfrak{h}}} {\,} {}_{\gamma} {\,} [{\Phi}_{{\mathfrak{g}} \oplus {\mathfrak{h}}} {\,} {}_{\lambda} {\,} {U}_{{\mathfrak{g}} \oplus {\mathfrak{h}}}]]
\\
&= [{G}_{{{\mathfrak{g}} \oplus {\mathfrak{h}}} {\,} (0)}{\Phi}_{{\mathfrak{g}} \oplus {\mathfrak{h}}} {\,} {}_{{\lambda} + {\gamma}} {\,} {U}_{{\mathfrak{g}} \oplus {\mathfrak{h}}}] - {3}[{G}_{{\mathfrak{g}} \oplus {\mathfrak{h}}} {\,} {}_{\gamma} {\,} {\Phi}_{{\mathfrak{g}} \oplus {\mathfrak{h}}}] 
\\
&= - {G}_{{{\mathfrak{g}} \oplus {\mathfrak{h}}} {\,} (0)}{\Phi}_{{\mathfrak{g}} \oplus {\mathfrak{h}}}.
\end{align*}
Consequently, from Proposition \ref{prop: equiv def of N=3/2 sconf}, ${G}_{{\mathfrak{g}} \oplus {\mathfrak{h}}}$ and ${U}_{{\mathfrak{g}} \oplus {\mathfrak{h}}}$ form an $N = {\frac{3}{2}}$ superconformal structure on the super-affine vertex algebra ${V}^{{\kappa}_{{\mathfrak{g}} \oplus {\mathfrak{h}}}}_{{\,\,} {\text{s}}}({\mathfrak{g}} \oplus {\mathfrak{h}})$.
\end{proof}

\vskip 4mm
\section{Supersymmetric semi--infinite cohomology}\label{section: SUSY semi inf cohom}
\vskip 2mm

In this section, following the construction of the semi-infinite cohomology given in \cites{Feigin84, FGZ86} (see also Section 3.1 of \cite{CLNS25}), we consider the supersymmetric semi-infinite cohomology for a ${V}^{0}_{{\,\,} {\text{s}}}({\mathfrak{h}})$-module $M$, where ${\mathfrak{h}}$ is a Lie superalgebra. Concrete examples of the supersymmetric semi-infinite cohomology can be found in \cites{Rhedin96, FS97A} and \cites{Kuramoto94, MR94, MRS21}. In the latter part of this section, we also introduce the relative supersymmetric semi-infinite cohomology.

\vskip 1mm

To begin, for convenience, we first recall the notion of a supersymmetric vertex algebra, introduced as an ${N}_{K} = 1$ supersymmetric vertex algebra in \cite{HK07} and as an $N = 1$ vertex superalgebra in Section 5.9 of \cite{Kac98}. A {\textit{supersymmetric vertex algebra}} is a vertex algebra ${V}$ equipped with an odd endomorphism ${D} \in {\text{End}}(V)$ such that ${D}$ is an odd derivation of the ${n}$-th products for all ${n} \in {\mathbb{Z}}$ and satisfies ${D}^{2} = {\partial}$, where ${\partial}$ is the infinitesimal translation operator of ${V}$ (see (5.9.9) and (5.9.10) of \cite{Kac98}).

\vskip 1mm

For a supersymmetric vertex algebra $V$, we consider an odd ${\mathbb{C}}$-linear map ${V} \otimes {V} \to {\mathbb{C}}[{\lambda}, {\chi}] \otimes {V}$, called the {\textit{${\Lambda}$-bracket}}, defined by
\begin{align*}
{a} \otimes {b} {\,} \mapsto {\,} [{a} {\,} {}_{\Lambda} {\,} {b}] :=& {\,} [Da {\,} {}_{\lambda} {\,} {b}] + {\chi}[{a} {\,} {}_{\lambda} {\,} {b}]
\\
=& {\,} {\sum}_{n \in {\mathbb{Z}}_{\geq 0}} {\frac{{\lambda}^{n}}{n!}}{(Da)}_{(n)}{b} + {\chi}{\,}{\sum}_{n \in {\mathbb{Z}}_{\geq 0}} {\frac{{\lambda}^{n}}{n!}}{a}_{(n)}{b},
\end{align*}
where ${\lambda}$ is an even indeterminate and ${\chi}$ is an odd indeterminate satisfying ${\chi}^{2} = - {\lambda}$. Then, for parity-homogeneous elements $a, b, c \in V$, the ${\Lambda}$-bracket satisfies the following properties:
\begin{enumerate}[]
\item (Sesquilinearity)
{\quad}
$
[{D}{a} {\,} {}_{\Lambda} {\,} {b}] = {\chi}[{a} {\,} {}_{\Lambda} {\,} {b}],
{\quad}
[{a} {\,} {}_{\Lambda} {\,} {D}{b}] = (-1)^{a + 1}({D} + {\chi})[{a} {\,} {}_{\Lambda} {\,} {b}],
$
\item (Skew-symmetry)
{\quad}
$
[{b} {\,} {}_{\Lambda} {\,} {a}] = (-1)^{ab}[{a} {\,} {}_{-{\nabla}-{\Lambda}} {\,} {b}],
$
\item (Jacobi identity) 
{\quad}
$
[{a} {\,} {}_{\Lambda} {\,} [{b} {\,} {}_{\Gamma} {\,} {c}]] = (-1)^{a + 1}[[{a} {\,} {}_{\Lambda} {\,} {b}] {\,} {}_{{\Lambda} + {\Gamma}} {\,} {c}] + (-1)^{(a+1)(b+1)}[{b} {\,} {}_{\Gamma} {\,} [{a} {\,} {}_{\Lambda} {\,} {c}]],
$
\item (Non-commutative Wick formula)
\begin{equation*}
[{a} {\,} {}_{\Lambda} {\,} {b}{c}] = [{a} {\,} {}_{\Lambda} {\,} {b}]{c} + (-1)^{(a+1)b}{b}[{a} {\,} {}_{\Lambda} {\,} {c}] + \int^{\Lambda}_{0} [[{a} {\,} {}_{\Lambda} {\,} {b}] {\,} {}_{\Gamma} {\,} {c}] {\,} {d}{\Gamma},
\end{equation*}
\end{enumerate}
where $\int^{\Lambda}_{0} {\,} {d}{\Gamma}$ is computed as ${\partial}_{\eta} \int^{\lambda}_{0} {\,} {d}{\gamma}$. Here, for the ${\Lambda}$-bracket, we use the notation
\begin{align*}
{\Lambda} := ({\lambda}, {\chi}),
{\quad}
- {\nabla} - {\Lambda} := (- {\partial} - {\lambda}, - {D} - {\chi}),
{\quad}
{\Gamma} := ({\gamma}, {\eta}),
{\quad}
{\Lambda} + {\Gamma} := ({\lambda} + {\gamma}, {\chi} + {\eta}),
\end{align*}
where ${\gamma}$ is an even indeterminate and ${\eta}$ is an odd indeterminate satisfying ${\eta}^{2} = - {\gamma}$. Moreover, ${\partial}$, ${D}$, ${\lambda}$, ${\chi}$, ${\gamma}$, and ${\eta}$ are subject to the following relations:
\begin{align*}
{D}{\chi} + {\chi}{D} = {2}{\lambda},
{\quad}
{D}{\eta} + {\eta}{D} = {2}{\gamma},
{\quad}
{\chi}{\eta} + {\eta}{\chi} = 0,
\end{align*}
and all remaining pairs supercommute with each other. In addition, ${\chi}$ and ${\eta}$ supercommute with the operators ${a}_{(n)}$ for ${a} \in V$ and ${n} \in {\mathbb{Z}}$. For more details on the ${\Lambda}$-bracket, the reader is referred to Definition 4.10 and Definition 4.19 of \cite{HK07}. For the definition of supersymmetric vertex algebras in terms of superfields, see Definition 4.13 of \cite{HK07}.

\vskip 1mm

Note that, if a vertex algebra ${V}$ has an $N = 1$ superconformal vector ${G}$, then the odd derivation $D := {G}_{(0)}$ gives rise to a supersymmetric vertex algebra structure on $V$. For example, the Kac--Todorov superconformal vector induces a supersymmetric vertex algebra structure on the super-affine vertex algebra, with an odd derivation ${D}$ sending generators ${\overline{a}}$ to ${a}$ and ${a}$ to ${\partial}{\overline{a}}$ (see \eqref{eq: OPE G with generator of saff}). Even when the bilinear form defining the super-affine vertex algebra is degenerate, the odd derivation ${D}$ of the ${n}$-th products defined by ${D}({\overline{a}}) = a$ and ${D}({a}) = {\partial}{\overline{a}}$ on the generators still gives a supersymmetric vertex algebra structure on the super-affine vertex algebra.

\vskip 1mm

Another example of a supersymmetric vertex algebra is the following vertex algebra, called the ${b}{c}$--${\beta}{\gamma}$ system, which is one of the ingredients used to construct the complex of supersymmetric semi-infinite cohomology.

\vskip 1mm

Throughout this section, let ${\mathfrak{h}}$ be a finite-dimensional Lie superalgebra, and let $\{ {u}_{i} \}_{i \in H}$ be a parity-homogeneous basis of ${\mathfrak{h}}$. We denote the structure constants of ${\mathfrak{h}}$ with respect to the basis $\{ {u}_{i} \}_{i \in H}$ by
\begin{align*}
[{u}_{i}, {u}_{j}] = {\sum}_{k \in H} {h}^{k}_{ij}{u}_{k},
\end{align*}
for $i, j \in H$. Denote by ${\textstyle{\bigwedge}}^{{\frac{\infty}{2}} + {\bullet}}_{{\,\,} {\text{s}}}({\mathfrak{h}})$ the vertex algebra consisting of ${\text{dim}}{\,}{\mathfrak{h}}$ copies of the ${b}{c}$--${\beta}{\gamma}$ system, that is, the vertex algebra generated by $\{ {\phi}_{i}, {D}{\phi}_{i}, {\phi}^{i}, {D}{\phi}^{i} \}_{i \in H}$ such that the only nonzero OPEs among the generators are given by
\begin{align*}
[{D}{\phi}_{i} {\,} {}_{\lambda} {\,} {\phi}^{j}] = {\delta}_{ij},
{\quad}
[{D}{\phi}^{i} {\,} {}_{\lambda} {\,} {\phi}_{j}] = {\delta}_{ij},
\end{align*}
for $i, j \in H$. Equivalently, in terms of the ${\Lambda}$-bracket, the above OPEs are encoded by
\begin{align*}
[{\phi}_{i} {\,} {}_{\Lambda} {\,} {\phi}^{j}] = {\delta}_{ij},
\end{align*}
for $i, j \in H$. Here, the parity of each generator is given by $p({\phi}_{i}) = p({D}{\phi}^{i}) = p({u}_{i})$ and $p({\phi}^{i}) = p({D}{\phi}_{i}) = p({u}_{i}) + 1$ for $i \in H$. Then we have the following vertex algebra homomorphism.

\vskip 1mm

\begin{lemma}\label{lem: sec SUSY semi inf cohom J wedge homo}\rm
The map ${J}_{\bigwedge} : {V}^{0}_{{\,\,} {\text{s}}}({\mathfrak{h}}) \rightarrow {\textstyle{\bigwedge}}^{{\frac{\infty}{2}} + {\bullet}}_{{\,\,} {\text{s}}}({\mathfrak{h}})$ defined on the generators by
\begin{align*}
{J}_{\bigwedge}({\overline{u}}_{i}) &= {\sum}_{j, k \in H} (-1)^{jk}{h}^{k}_{ij}{\phi}^{j}{\phi}_{k},
\\
{J}_{\bigwedge}({u}_{i}) &= {\sum}_{j, k \in H} (-1)^{jk}{\big(}{h}^{k}_{ij}({D}{\phi}^{j}){\phi}_{k} + (-1)^{j+1} {h}^{k}_{ij}{\phi}^{j}{D}{\phi}_{k}{\big)},
\end{align*}
for $i \in H$, is a vertex algebra homomorphism.
\end{lemma}

\begin{proof}
From the non-commutative Wick formula and the Jacobi identity for Lie superalgebras, we have:
\begin{align*}
[{J}_{\bigwedge}({\overline{u}}_{a}) {\,} {}_{\Lambda} {\,} {J}_{\bigwedge}({\overline{u}}_{b})] 
=& {\,} [{\sum}_{i, j \in H} (-1)^{ij}{h}^{j}_{ai}{\phi}^{i}{\phi}_{j} {\,} {}_{\Lambda} {\,} {\sum}_{r, s \in H} (-1)^{rs}{h}^{s}_{br}{\phi}^{r}{\phi}_{s}]
\\
=& {\,} {\sum}_{i, j, s \in H} (-1)^{ij + js + i + 1}{h}^{j}_{ai}{h}^{s}_{bj}{\phi}^{i}{\phi}_{s} + {\sum}_{i, j, r \in H} (-1)^{jr + i + j}{h}^{j}_{ai}{h}^{i}_{br}{\phi}^{r}{\phi}_{j}
\\
=& {\,} {\sum}_{x, y, z \in H} {\big(}(-1)^{xz + yz + x + 1}{h}^{z}_{ax}{h}^{y}_{bz} + (-1)^{xy + y + z}{h}^{y}_{az}{h}^{z}_{bx}{\big)}{\phi}^{x}{\phi}_{y}
\\
=& {\,} {\sum}_{x, y, z \in H} (-1)^{a + xy}{h}^{z}_{ab}{h}^{y}_{zx}{\phi}^{x}{\phi}_{y}
\\
=& {\,} (-1)^{a}{J}_{\bigwedge}({\sum}_{z \in H} {h}^{z}_{ab}{\overline{u}}_{z})
\\
=& {\,} (-1)^{a}{J}_{\bigwedge}({\overline{[{u}_{a}, {u}_{b}]}}) 
= {J}_{\bigwedge}([{u}_{a} {\,} {}_{\lambda} {\,} {\overline{u}}_{b}]) 
= {J}_{\bigwedge}([{\overline{u}}_{a} {\,} {}_{\Lambda} {\,} {\overline{u}}_{b}]),
\end{align*}
for $a, b \in H$. Hence, using ${J}_{\bigwedge}{D} = {D}{J}_{\bigwedge}$, we have:
\begin{align*}
[{J}_{\bigwedge}({u}_{a}) {\,} {}_{\lambda} {\,} {J}_{\bigwedge}({\overline{u}}_{b})] 
= {J}_{\bigwedge}([{u}_{a} {\,} {}_{\lambda} {\,} {\overline{u}}_{b}]),
{\quad}
[{J}_{\bigwedge}({\overline{u}}_{a}) {\,} {}_{\lambda} {\,} {J}_{\bigwedge}({\overline{u}}_{b})] 
= {J}_{\bigwedge}([{\overline{u}}_{a} {\,} {}_{\lambda} {\,} {\overline{u}}_{b}]),
\end{align*}
for $a, b \in H$. Moreover, by the sesquilinearity of the ${\Lambda}$-bracket, we obtain:
\begin{align*}
[{J}_{\bigwedge}({\overline{u}}_{a}) {\,} {}_{\Lambda} {\,} {J}_{\bigwedge}({u}_{b})] 
&= [{J}_{\bigwedge}({\overline{u}}_{a}) {\,} {}_{\Lambda} {\,} {D}{J}_{\bigwedge}({\overline{u}}_{b})]
\\
&= ({D} + {\chi}){J}_{\bigwedge}({\overline{[{u}_{a}, {u}_{b}]}}) 
= {J}_{\bigwedge}([{u}_{a} {\,} {}_{\lambda} {\,} {u}_{b}]) + {\chi}{J}_{\bigwedge}([{\overline{u}}_{a} {\,} {}_{\lambda} {\,} {u}_{b}]),
\end{align*}
for $a, b \in H$, which gives the desired result.
\end{proof}

\vskip 1mm

Recall that the $n$-th products on the tensor product of vertex algebras ${V}_{1} \otimes {V}_{2}$ are given by
\begin{align*}
({a} \otimes {b})_{(n)} := {\sum}_{m \in {\mathbb{Z}}} {a}_{(m)} \otimes {b}_{(n - m - 1)},
\end{align*}
for $a \in {V}_{1}$ and $b \in {V}_{2}$ (see (4.3.2) of \cite{Kac98}). Then, for supersymmetric vertex algebras $({V}_{1}, {D}_{1})$ and $({V}_{2}, {D}_{2})$, the odd endomorphism ${D}_{1} \otimes {\text{id}}_{{V}_{2}} + {\text{id}}_{{V}_{1}} \otimes {D}_{2}$ gives rise to a supersymmetric vertex algebra structure on ${V}_{1} \otimes {V}_{2}$.

\vskip 1mm

We now consider the supersymmetric vertex algebra ${V}^{0}_{{\,\,} {\text{s}}}({\mathfrak{h}}) \otimes {\textstyle{\bigwedge}}^{{\frac{\infty}{2}} + {\bullet}}_{{\,\,} {\text{s}}}({\mathfrak{h}})$ and introduce the following even vector in ${V}^{0}_{{\,\,} {\text{s}}}({\mathfrak{h}}) \otimes {\textstyle{\bigwedge}}^{{\frac{\infty}{2}} + {\bullet}}_{{\,\,} {\text{s}}}({\mathfrak{h}})$:
\begin{align}
\label{eq: sec SUSY semi inf cohom def of Q}
Q = {\sum}_{i \in H} (-1)^{i} {\overline{u}}_{i} \otimes {\phi}^{i} - {\frac{1}{2}}{\,}{\sum}_{i, j, k \in H} (-1)^{jk + j}{h}^{i}_{jk} {\ket{0}} \otimes {\phi}_{i}({\phi}^{j}{\phi}^{k}).
\end{align}
Then we obtain the following OPE of ${Q}$ with itself:

\vskip 1mm

\begin{proposition}\label{prop: sec SUSY semi inf OPE Q and Q}\rm
In ${V}^{0}_{{\,\,} {\text{s}}}({\mathfrak{h}}) \otimes {\textstyle{\bigwedge}}^{{\frac{\infty}{2}} + {\bullet}}_{{\,\,} {\text{s}}}({\mathfrak{h}})$, we have $[{Q} {\,} {}_{\Lambda} {\,} {Q}] = 0$.
\end{proposition}

\begin{proof}
From the Jacobi identity for Lie superalgebras and the following ${\Lambda}$-bracket relations
\begin{align*}
[{\overline{u}}_{i} \otimes {\phi}^{i} {\,} {}_{\Lambda} {\,} {\overline{u}}_{r} \otimes {\phi}^{r}] 
=& {\,} (-1)^{r} {\overline{[{u}_{r}, {u}_{i}]}} \otimes {\phi}^{i}{\phi}^{r},
\\
[{\overline{u}}_{i} \otimes {\phi}^{i} {\,} {}_{\Lambda} {\,} {\ket{0}} \otimes {\phi}_{r}({\phi}^{s}{\phi}^{t})] 
=& {\,} (-1)^{r + 1}{\delta}_{ir} {\overline{u}}_{i} \otimes {\phi}^{s}{\phi}^{t},
\\
[{\ket{0}} \otimes {\phi}_{i}({\phi}^{j}{\phi}^{k}) {\,} {}_{\Lambda} {\,} {\ket{0}} \otimes {\phi}_{r}({\phi}^{s}{\phi}^{t})] 
=& {\,} (-1)^{(i + j + k)(r + s + t) + it + i + r}{\delta}_{is} {\ket{0}} \otimes {\phi}_{r}{\big(}{\phi}^{t}({\phi}^{j}{\phi}^{k}){\big)}
\\
&+ (-1)^{(i + j + k)(r + s + t) + r + s + 1}{\delta}_{it} {\ket{0}} \otimes {\phi}_{r}{\big(}{\phi}^{s}({\phi}^{j}{\phi}^{k}){\big)}
\\
&+ (-1)^{jk + i + j}{\delta}_{jr} {\ket{0}} \otimes {\phi}_{i}{\big(}{\phi}^{k}({\phi}^{s}{\phi}^{t}){\big)}
\\
&+ (-1)^{i + j + 1}{\delta}_{kr} {\ket{0}} \otimes {\phi}_{i}{\big(}{\phi}^{j}({\phi}^{s}{\phi}^{t}){\big)},
\end{align*}
for $i, j, k, r, s, t \in H$, we obtain the desired result.
\end{proof}

\vskip 1mm

Let ${M}$ be a module over a vertex algebra ${V}$ (for the definition of a module over a vertex algebra, see, for example, Section 4.1 of \cite{CKLW18} and Section A.3 of \cite{Sato22}). Then, in particular, the following identity follows from the Borcherds identity for vertex algebra modules:
\begin{align}
\label{eq: sec SUSY semi inf cohom BI}
({a}_{(0)}{b})^{M}_{(0)} = {a}^{M}_{(0)}{b}^{M}_{(0)} + (-1)^{ab + 1}{b}^{M}_{(0)}{a}^{M}_{(0)},
\end{align}
for $a, b \in V$. Moreover, let ${V}_{1}$ and ${V}_{2}$ be vertex algebras, and let ${M}_{i}$ be a ${V}_{i}$-module for $i = 1, 2$. Then ${M}_{1} \otimes {M}_{2}$ is a ${V}_{1} \otimes {V}_{2}$-module via the action
\begin{align*}
({a} \otimes {b})^{{M}_{1} \otimes {M}_{2}}_{(n)} = {\sum}_{m \in {\mathbb{Z}}} {a}^{{M}_{1}}_{(m)} \otimes {b}^{{M}_{2}}_{(n - m - 1)},
\end{align*}
for $a \in {V}_{1}$, $b \in {V}_{2}$, and $n \in {\mathbb{Z}}$ (see Section 4.6 of \cite{LL04}).

\vskip 1mm

We now introduce the complex and differential for supersymmetric semi-infinite cohomology. For a ${V}^{0}_{{\,\,} {\text{s}}}({\mathfrak{h}})$-module ${M}$, we set 
\begin{align}
\label{eq: sec SUSY semi inf cohom def of complex}
{C}^{{\frac{\infty}{2}} + {\bullet}}_{{\,\,} {\text{s}}}({\mathfrak{h}}, {M}) := {M} \otimes {\textstyle{\bigwedge}}^{{\frac{\infty}{2}} + {\bullet}}_{{\,\,} {\text{s}}}({\mathfrak{h}}),
\end{align}
and consider ${C}^{{\frac{\infty}{2}} + {\bullet}}_{{\,\,} {\text{s}}}({\mathfrak{h}}, {M})$ as a ${V}^{0}_{{\,\,} {\text{s}}}({\mathfrak{h}}) \otimes {\textstyle{\bigwedge}}^{{\frac{\infty}{2}} + {\bullet}}_{{\,\,} {\text{s}}}({\mathfrak{h}})$-module. Define the differential on ${C}^{{\frac{\infty}{2}} + {\bullet}}_{{\,\,} {\text{s}}}({\mathfrak{h}}, {M})$ by
\begin{align}
\label{eq: sec SUSY semi inf cohom def of differential}
{d}_{\text{s}} 
:= {DQ}^{{C}^{{\frac{\infty}{2}} + {\bullet}}_{{\,\,} {\text{s}}}({\mathfrak{h}}, {M})}_{(0)},
\end{align}
as a ${V}^{0}_{{\,\,} {\text{s}}}({\mathfrak{h}}) \otimes {\textstyle{\bigwedge}}^{{\frac{\infty}{2}} + {\bullet}}_{{\,\,} {\text{s}}}({\mathfrak{h}})$-module action, that is:
\begin{align}
\begin{split}
{d}_{\text{s}} 
=& {\,} {\sum}_{i \in H}{\sum}_{m \in {\mathbb{Z}}} (-1)^{i} {u}^{M}_{i {\,} (m)} \otimes {\phi}^{i}_{(-m-1)} 
- {\sum}_{i \in H}{\sum}_{m \in {\mathbb{Z}}} {\overline{u}}^{M}_{i {\,} (m)} \otimes {D}{\phi}^{i}_{(-m-1)}
\\
&- {\frac{1}{2}}{\,}{\sum}_{i, j, k \in H} (-1)^{jk + j}{h}^{i}_{jk} {\,}{\text{id}}_{M} \otimes ({D}({\phi}_{i}{\phi}^{j}{\phi}^{k}))_{(0)},
\end{split}
\label{eq: sec SUSY semi inf cohom mode expression of differential}
\end{align}
where $Q$ is given by \eqref{eq: sec SUSY semi inf cohom def of Q}. Then we have ${d}^{2}_{\text{s}} = 0$, which follows from Proposition \ref{prop: sec SUSY semi inf OPE Q and Q}, the sesquilinearity of the ${\Lambda}$-bracket, and \eqref{eq: sec SUSY semi inf cohom BI}.

\vskip 1mm

We also introduce a cohomological grading on ${\textstyle{\bigwedge}}^{{\frac{\infty}{2}} + {\bullet}}_{{\,\,} {\text{s}}}({\mathfrak{h}})$ as follows:
\begin{align*}
{\text{deg}}{\,}({D}^{m}{\phi}_{i}) := -1,
{\quad}
{\text{deg}}{\,}({D}^{m}{\phi}^{i}) := 1,
\end{align*}
for $i \in H$ and $m \in {\mathbb{Z}}_{\geq 0}$. For the normally ordered product, we set
\begin{align*}
{\text{deg}}{\,}({\varphi}_{1}{\varphi}_{2}) := {\text{deg}}{\,}({\varphi}_{1}) + {\text{deg}}{\,}({\varphi}_{2}),
\end{align*}
for ${\varphi}_{1}, {\varphi}_{2} \in {\textstyle{\bigwedge}}^{{\frac{\infty}{2}} + {\bullet}}_{{\,\,} {\text{s}}}({\mathfrak{h}})$. Thus, with respect to this cohomological grading, we have the following decomposition:
\begin{align*}
{\textstyle{\bigwedge}}^{{\frac{\infty}{2}} + {\bullet}}_{{\,\,} {\text{s}}}({\mathfrak{h}}) = {\bigoplus}_{n \in {\mathbb{Z}}}{\,} {\textstyle{\bigwedge}}^{{\frac{\infty}{2}} + {n}}_{{\,\,} {\text{s}}}({\mathfrak{h}}),
\end{align*}
where ${\textstyle{\bigwedge}}^{{\frac{\infty}{2}} + {n}}_{{\,\,} {\text{s}}}({\mathfrak{h}}) := {\{} {\varphi} \in {\textstyle{\bigwedge}}^{{\frac{\infty}{2}} + {\bullet}}_{{\,\,} {\text{s}}}({\mathfrak{h}}) {\,} {|} {\,} {\text{deg}}{\,}({\varphi}) = n {\}}$ for $n \in {\mathbb{Z}}$. Consequently, the space ${C}^{{\frac{\infty}{2}} + {\bullet}}_{{\,\,} {\text{s}}}({\mathfrak{h}}, {M})$ is endowed with a cohomological grading given by
\begin{align*}
{C}^{{\frac{\infty}{2}} + {\bullet}}_{{\,\,} {\text{s}}}({\mathfrak{h}}, {M}) = {\bigoplus}_{n \in {\mathbb{Z}}}{\,} {C}^{{\frac{\infty}{2}} + {n}}_{{\,\,} {\text{s}}}({\mathfrak{h}}, {M}),
\end{align*}
where ${C}^{{\frac{\infty}{2}} + {n}}_{{\,\,} {\text{s}}}({\mathfrak{h}}, {M}) := {M} \otimes {\textstyle{\bigwedge}}^{{\frac{\infty}{2}} + {n}}_{{\,\,} {\text{s}}}({\mathfrak{h}})$ for $n \in {\mathbb{Z}}$. Moreover, ${d}_{\text{s}}$ increases the cohomological grading by $1$, that is:
\begin{align*}
{d}_{\text{s}}({v} \otimes {\varphi}) \in {C}^{{\frac{\infty}{2}} + {n + 1}}_{{\,\,} {\text{s}}}({\mathfrak{h}}, {M}),    
\end{align*}
for ${v} \otimes {\varphi} \in {C}^{{\frac{\infty}{2}} + {n}}_{{\,\,} {\text{s}}}({\mathfrak{h}}, {M})$ and $n \in {\mathbb{Z}}$.

\vskip 1mm

\begin{definition}\label{def: SUSY semi-inf coh}\rm
Let ${\mathfrak{h}}$ be a finite-dimensional Lie superalgebra, and let ${M}$ be a ${V}^{0}_{{\,\,} {\text{s}}}({\mathfrak{h}})$-module. The {\textit{supersymmetric semi-infinite cohomology}} of ${\widehat{\mathfrak{h}}}$ with coefficients in ${M}$ is the cohomology of the complex $({C}^{{\frac{\infty}{2}} + {\bullet}}_{{\,\,} {\text{s}}}({\mathfrak{h}}, {M}), {d}_{\text{s}})$, denoted by
\begin{align*}
{H}^{{\frac{\infty}{2}} + {\bullet}}_{{\,\,} {\text{s}}}({\mathfrak{h}}, {M}) = {\bigoplus}_{n \in {\mathbb{Z}}}{\,} {H}^{{\frac{\infty}{2}} + {n}}_{{\,\,} {\text{s}}}({\mathfrak{h}}, {M}),
\end{align*}
where ${C}^{{\frac{\infty}{2}} + {\bullet}}_{{\,\,} {\text{s}}}({\mathfrak{h}}, {M})$ and ${d}_{\text{s}}$ are defined in \eqref{eq: sec SUSY semi inf cohom def of complex} and \eqref{eq: sec SUSY semi inf cohom def of differential}, respectively.
\end{definition}

\vskip 1mm

To define the relative version of supersymmetric semi-infinite cohomology, we now consider the following lemma, which is a supersymmetric analogue of Proposition 1.4 of \cite{FGZ86}.

\vskip 1mm

\begin{lemma}\rm
In ${V}^{0}_{{\,\,} {\text{s}}}({\mathfrak{h}}) \otimes {\textstyle{\bigwedge}}^{{\frac{\infty}{2}} + {\bullet}}_{{\,\,} {\text{s}}}({\mathfrak{h}})$, we have
\begin{align}
&[{Q} {\,} {}_{\Lambda} {\,} {\ket{0}} \otimes {\phi}_{r}] 
= - {\big(}{\overline{u}}_{r} \otimes {\ket{0}} + {\ket{0}} \otimes {J}_{\bigwedge}({\overline{u}}_{r}){\big)},
\label{eq: sec SUSY semi inf cohom lem Q and phi}
\\
&[{Q} {\,} {}_{\Lambda} {\,} {\overline{u}}_{r} \otimes {\ket{0}} + {\ket{0}} \otimes {J}_{\bigwedge}({\overline{u}}_{r})] 
= 0,
\label{eq: sec SUSY semi inf cohom lem Q and J}
\end{align}
for $r \in H$, where ${J}_{\bigwedge}$ is given in Lemma \ref{lem: sec SUSY semi inf cohom J wedge homo}.
\end{lemma}

\begin{proof}
First, in ${\textstyle{\bigwedge}}^{{\frac{\infty}{2}} + {\bullet}}_{{\,\,} {\text{s}}}({\mathfrak{h}})$, we know that
\begin{align}
[{\phi}_{i}({\phi}^{j}{\phi}^{k}) {\,} {}_{\Lambda} {\,} {\phi}_{r}] 
= (-1)^{r(i + j + k)}{\big(}(-1)^{(r + 1)i}{\delta}_{rj}{\phi}_{i}{\phi}^{k} + (-1)^{(r + 1)(i + j + 1)}{\delta}_{rk}{\phi}_{i}{\phi}^{j}{\big)},
\label{eq: sec SUSY semi inf cohom phi phi phi and phi}
\end{align}
for $i, j, k, r \in H$. Also, in ${V}^{0}_{{\,\,} {\text{s}}}({\mathfrak{h}}) \otimes {\textstyle{\bigwedge}}^{{\frac{\infty}{2}} + {\bullet}}_{{\,\,} {\text{s}}}({\mathfrak{h}})$, we have
\begin{align*}
[{\overline{u}}_{i} \otimes {\phi}^{i} {\,} {}_{\Lambda} {\,} {\ket{0}} \otimes {\phi}_{r}] 
= (-1)^{r + 1}{\delta}_{ri} {\overline{u}}_{i} \otimes {\ket{0}},
\end{align*}
for $i, r \in H$. Hence, \eqref{eq: sec SUSY semi inf cohom lem Q and phi} follows directly from the above ${\Lambda}$-bracket relations. The OPE in \eqref{eq: sec SUSY semi inf cohom lem Q and J} follows from \eqref{eq: sec SUSY semi inf cohom lem Q and phi}, the Jacobi identity of the ${\Lambda}$-bracket, and Proposition \ref{prop: sec SUSY semi inf OPE Q and Q}.
\end{proof}

\vskip 1mm

For ${u} \in {\mathfrak{h}}$, we use the following notation:
\begin{align*}
{J}_{u} := {u} \otimes {\ket{0}} + {\ket{0}} \otimes {J}_{\bigwedge}(u) \in {V}^{0}_{{\,\,} {\text{s}}}({\mathfrak{h}}) \otimes {\textstyle{\bigwedge}}^{{\frac{\infty}{2}} + {\bullet}}_{{\,\,} {\text{s}}}({\mathfrak{h}}),
\end{align*}
and consider the invariant subspace defined by
\begin{align}
\label{eq: sec SUSY semi inf cohom def of subcomplex}
{C}^{{\frac{\infty}{2}} + {\bullet}}_{{\,\,} {\text{s, rel}}}({\mathfrak{h}}, {M}) := {\{} {v} \in {C}^{{\frac{\infty}{2}} + {\bullet}}_{{\,\,} {\text{s}}}({\mathfrak{h}}, {M}) {\,} {|} {\,} {J}_{{u}_{i} (0)}{v} = 0 {\text{ and }} ({\ket{0}} \otimes {D}{\phi}_{i})_{(0)}{v} = 0 {\text{ for all }} {i} \in H {\}},
\end{align}
where ${J}_{{u}_{i} (0)}$ and $({\ket{0}} \otimes {D}{\phi}_{i})_{(0)}$ act via the ${V}^{0}_{{\,\,} {\text{s}}}({\mathfrak{h}}) \otimes {\textstyle{\bigwedge}}^{{\frac{\infty}{2}} + {\bullet}}_{{\,\,} {\text{s}}}({\mathfrak{h}})$-module action on ${C}^{{\frac{\infty}{2}} + {\bullet}}_{{\,\,} {\text{s}}}({\mathfrak{h}}, {M})$. Then, in ${V}^{0}_{{\,\,} {\text{s}}}({\mathfrak{h}}) \otimes {\textstyle{\bigwedge}}^{{\frac{\infty}{2}} + {\bullet}}_{{\,\,} {\text{s}}}({\mathfrak{h}})$, we have:
\begin{gather}
{D}{Q}_{(0)}{J}_{{u}_{i}} = 0,
\label{eq: sec SUSY semi inf cohom DQ and J}
\\
{D}{Q}_{(0)}({\ket{0}} \otimes {D}{\phi}_{i}) = {J}_{{u}_{i}},
\label{eq: sec SUSY semi inf cohom DQ and Dphi}
\end{gather}
for ${i} \in H$, which follow from the sesquilinearity of the ${\Lambda}$-bracket together with \eqref{eq: sec SUSY semi inf cohom lem Q and J} and \eqref{eq: sec SUSY semi inf cohom lem Q and phi}, respectively. Thus, from \eqref{eq: sec SUSY semi inf cohom BI} together with \eqref{eq: sec SUSY semi inf cohom DQ and J} and \eqref{eq: sec SUSY semi inf cohom DQ and Dphi}, respectively, we have:
\begin{gather*}
{J}_{{u}_{i} (0)}{d}_{\text{s}} = (-1)^{i}{d}_{\text{s}}{J}_{{u}_{i} (0)},
\\
({\ket{0}} \otimes {D}{\phi}_{i})_{(0)}{d}_{\text{s}} = (-1)^{i + 1}{d}_{\text{s}}({\ket{0}} \otimes {D}{\phi}_{i})_{(0)} + (-1)^{i}{J}_{{u}_{i} (0)},
\end{gather*}
for $i \in H$, as the ${V}^{0}_{{\,\,} {\text{s}}}({\mathfrak{h}}) \otimes {\textstyle{\bigwedge}}^{{\frac{\infty}{2}} + {\bullet}}_{{\,\,} {\text{s}}}({\mathfrak{h}})$-module action on ${C}^{{\frac{\infty}{2}} + {\bullet}}_{{\,\,} {\text{s}}}({\mathfrak{h}}, {M})$. Hence, we have
\begin{align*}
{d}_{\text{s}}(v) \in {C}^{{\frac{\infty}{2}} + {\bullet}}_{{\,\,} {\text{s, rel}}}({\mathfrak{h}}, {M}),    
\end{align*}
for ${v} \in {C}^{{\frac{\infty}{2}} + {\bullet}}_{{\,\,} {\text{s, rel}}}({\mathfrak{h}}, {M})$, so that the invariant subspace defined in \eqref{eq: sec SUSY semi inf cohom def of subcomplex} is a subcomplex of the complex given in \eqref{eq: sec SUSY semi inf cohom def of complex} with respect to the differential ${d}_{\text{s}}$, which leads to the following definition.

\vskip 1mm

\begin{definition}\label{def: rel SUSY semi-inf coh}\rm
The cohomology of the subcomplex ${C}^{{\frac{\infty}{2}} + {\bullet}}_{{\,\,} {\text{s, rel}}}({\mathfrak{h}}, {M})$ in \eqref{eq: sec SUSY semi inf cohom def of subcomplex} with the differential ${d}_{\text{s}}$ is denoted by
\begin{align*}
{H}^{{\frac{\infty}{2}} + {\bullet}}_{{\,\,} {\text{s, rel}}}({\mathfrak{h}}, {M}) = {\bigoplus}_{n \in {\mathbb{Z}}}{\,} {H}^{{\frac{\infty}{2}} + {n}}_{{\,\,} {\text{s, rel}}}({\mathfrak{h}}, {M}),
\end{align*}
and is called the {\textit{relative supersymmetric semi-infinite cohomology}} of ${\widehat{\mathfrak{h}}}$ with coefficients in ${M}$.
\end{definition}

\vskip 4mm
\section{$N = 2$ superconformal structure of supersymmetric semi--infinite cohomology}\label{section: N=2 sconf of SUSY semi inf cohom}
\vskip 2mm

We have studied the super Manin pair structure on ${\mathfrak{g}} \oplus {\mathfrak{h}}$ arising from both the Kazama--Suzuki $N = 2$ condition (see Assumption \ref{assum: second assumption}) and the unimodularity of ${\mathfrak{h}}$, and have shown that the corresponding super-affine vertex algebra admits an $N = {\frac{3}{2}}$ superconformal structure (see Theorem \ref{theorem: N=3/2 of saff}). In this section, we will show that the above $N = {\frac{3}{2}}$ superconformal structure on the super-affine vertex algebra, combined with the $N = 2$ superconformal structure on the ghost system of the supersymmetric semi-infinite complex, gives rise to BRST-invariant states and induces an $N = 2$ superconformal structure on the supersymmetric semi-infinite cohomology of ${\widehat{\mathfrak{h}}}$ with coefficients in the super-affine vertex algebra associated to ${\mathfrak{g}} \oplus {\mathfrak{h}}$.

\vskip 1mm

Recall Assumption \ref{assum: first assumption}, which states that ${\mathfrak{g}}$ is a finite-dimensional Lie superalgebra with a nondegenerate consistent supersymmetric invariant bilinear form ${\kappa}_{\mathfrak{g}}$, and that there exists an injective Lie superalgebra homomorphism ${\iota} : {\mathfrak{h}} \rightarrow {\mathfrak{g}}$ such that the bilinear form ${\kappa}_{\mathfrak{h}} := {\kappa}_{\mathfrak{g}} \circ ({\iota} \times {\iota})$ is nondegenerate on ${\mathfrak{h}}$. Under Assumption \ref{assum: first assumption}, consider the Lie superalgebra ${\mathfrak{g}} \oplus {\mathfrak{h}}$ together with the bilinear form ${\kappa}_{{\mathfrak{g}} \oplus {\mathfrak{h}}}$ defined in \eqref{eq: sec sMP bilinear form on g+h}. Then the super-affine vertex algebra ${V}^{{\kappa}_{{\mathfrak{g}} \oplus {\mathfrak{h}}}}_{{\,\,} {\text{s}}}({\mathfrak{g}} \oplus {\mathfrak{h}})$ is a ${V}^{0}_{{\,\,} {\text{s}}}({\mathfrak{h}})$-module via the diagonal vertex algebra homomorphism ${V}^{0}_{{\,\,} {\text{s}}}({\mathfrak{h}}) \rightarrow {V}^{{\kappa}_{{\mathfrak{g}} \oplus {\mathfrak{h}}}}_{{\,\,} {\text{s}}}({\mathfrak{g}} \oplus {\mathfrak{h}})$ defined on the generators by
\begin{align}
{h} \mapsto {\iota}({h}) + {h},
{\quad}
{\overline{h}} \mapsto {\overline{{\iota}({h})}} + {\overline{h}},
\label{eq: sec N=2 of SUSY semi inf diag homo}
\end{align}
for ${h} \in {\mathfrak{h}}$. Note that the conformal grading on ${V}^{{\kappa}_{{\mathfrak{g}} \oplus {\mathfrak{h}}}}_{{\,\,} {\text{s}}}({\mathfrak{g}} \oplus {\mathfrak{h}})$ is given by the conformal vector ${T}_{{\mathfrak{g}} \oplus {\mathfrak{h}}}$ (see \eqref{eq: OPE T with generator of saff}).

\vskip 1mm

Now, let $\{ {u}_{i} \}_{i \in H}$ be a parity-homogeneous basis of ${\mathfrak{h}}$. For the ghost sector ${\textstyle{\bigwedge}}^{{\frac{\infty}{2}} + {\bullet}}_{{\,\,} {\text{s}}}({\mathfrak{h}})$, we recall the following shifts of the standard $N = 2$ superconformal structure on the ${b}{c}$--${\beta}{\gamma}$ system:

\vskip 1mm

\begin{proposition}[\cite{FGLS92}]\label{prop: sec N=2 of SUSY semi inf N=2 shift of ghost}\rm
In ${\textstyle{\bigwedge}}^{{\frac{\infty}{2}} + {\bullet}}_{{\,\,} {\text{s}}}({\mathfrak{h}})$, the following vectors
\begin{align*}
{G}_{\wedge} 
=& {\,} {\sum}_{i \in {H}_{0}} ({\partial}{\phi}_{i}){\phi}^{i} + {\sum}_{i \in {H}_{1}} {\phi}_{i}{\partial}{\phi}^{i} + {\sum}_{i \in {H}} {D}{\phi}_{i}{D}{\phi}^{i} + {\sum}_{i \in {H}} {s}_{i}{\partial}({\phi}_{i}{\phi}^{i}),
\\
{U}_{\wedge} 
=& {\,} {\sum}_{i \in {H}_{0}} ({D}{\phi}_{i}){\phi}^{i} - {\sum}_{i \in {H}_{1}} {\phi}_{i}{D}{\phi}^{i} + {\sum}_{i \in {H}} {s}_{i}{D}({\phi}_{i}{\phi}^{i}),
\end{align*}
form an $N = 2$ superconformal structure on ${\textstyle{\bigwedge}}^{{\frac{\infty}{2}} + {\bullet}}_{{\,\,} {\text{s}}}({\mathfrak{h}})$ with central charge ${\sum}_{i \in {H}} ({6}{s}_{i} + 3)$, for any ${s}_{i} \in {\mathbb{C}}$. Here, the index sets ${H}_{0}$ and ${H}_{1}$ are given by:
\begin{align*}
{H}_{0} = {\{} i \in H {\,} {|} {\,} p({u}_{i}) = 0 {\}},
{\quad}
{H}_{1} = {\{} i \in H {\,} {|} {\,} p({u}_{i}) = 1 {\}}.
\end{align*}
\end{proposition}

\vskip 1mm

We will use the following $N = 2$ superconformal structure on ${\textstyle{\bigwedge}}^{{\frac{\infty}{2}} + {\bullet}}_{{\,\,} {\text{s}}}({\mathfrak{h}})$ throughout the rest of this section. Taking the shift ${s}_{i} = - {\frac{1}{2}}{\big(}1 + (-1)^{i}{\big)}$ for all $i \in H$ in the $N = 2$ superconformal structure of Proposition \ref{prop: sec N=2 of SUSY semi inf N=2 shift of ghost}, we obtain:
\begin{align}
\begin{split}
{G}_{\wedge} 
=& {\,} {\sum}_{i \in {H}} {\big(}(-1)^{i + 1}{\phi}_{i}{\partial}{\phi}^{i} + {D}{\phi}_{i}{D}{\phi}^{i}{\big)},
\\
{U}_{\wedge} 
=& {\,} - {\sum}_{i \in {H}} {\phi}_{i}{D}{\phi}^{i},
\end{split}
\label{eq: sec N=2 of SUSY semi inf def of G and U in ghost}
\end{align}
which form an $N = 2$ superconformal structure on ${\textstyle{\bigwedge}}^{{\frac{\infty}{2}} + {\bullet}}_{{\,\,} {\text{s}}}({\mathfrak{h}})$ with central charge $-{3}{\,}{\text{sdim}}{\,}{\mathfrak{h}}$. Moreover, the conformal grading on ${\textstyle{\bigwedge}}^{{\frac{\infty}{2}} + {\bullet}}_{{\,\,} {\text{s}}}({\mathfrak{h}})$ is given by the fact that the generators are primary of conformal weights
\begin{align*}
{\Delta}({\phi}_{i}) = {\frac{1}{2}},
{\quad}
{\Delta}({D}{\phi}_{i}) = 1,
{\quad}
{\Delta}({\phi}^{i}) = 0,
{\quad}
{\Delta}({D}{\phi}^{i}) = {\frac{1}{2}},
\end{align*}
for $i \in H$ (see also Proposition 4.9 of \cite{Yoon23}).

\vskip 1mm

The following lemma will be used to verify that the fields defining an $N = {\frac{3}{2}}$ superconformal structure on the supersymmetric semi-infinite complex are BRST-invariant.

\vskip 1mm

\begin{lemma}\label{lem: sec N=2 of SUSY semi inf diff with G_gh and U_gh}\rm
In ${\textstyle{\bigwedge}}^{{\frac{\infty}{2}} + {\bullet}}_{{\,\,} {\text{s}}}({\mathfrak{h}})$, the following OPEs hold:
\begin{align*}
[{\phi}^{x} {\,} {}_{\lambda} {\,} {G}_{\wedge}] 
=& {\,} (-1)^{x}{D}{\phi}^{x},
\\
[{D}{\phi}^{x} {\,} {}_{\lambda} {\,} {G}_{\wedge}] 
=& {\,} (-1)^{x + 1}{\partial}{\phi}^{x},
\\
[{D}({\phi}_{x}{\phi}^{y}{\phi}^{z}) {\,} {}_{\lambda} {\,} {G}_{\wedge}] 
=& {\,} (-1)^{x + y + z + 1}{\lambda}({\phi}_{x}{\phi}^{y}{\phi}^{z}) + (-1)^{xy + x + 1}{\delta}_{xz}({\partial} + {\lambda}){D}{\phi}^{y} 
\\
&+ (-1)^{z}{\delta}_{xy}({\partial} + {\lambda}){D}{\phi}^{z},
\\
[{\phi}^{x} {\,} {}_{\lambda} {\,} {U}_{\wedge}] 
=& {\,} 0,
\\
[{D}{\phi}^{x} {\,} {}_{\lambda} {\,} {U}_{\wedge}] 
=& {\,} - {D}{\phi}^{x},
\\
[{D}({\phi}_{x}{\phi}^{y}{\phi}^{z}) {\,} {}_{\lambda} {\,} {U}_{\wedge}] 
=& {\,} {\delta}_{xy}({\partial} + {\lambda}){\phi}^{z} + (-1)^{(x + 1)(y + 1)}{\delta}_{xz}({\partial} + {\lambda}){\phi}^{y},
\end{align*}
for $x, y, z \in H$, where ${G}_{\wedge}$ and ${U}_{\wedge}$ are given in \eqref{eq: sec N=2 of SUSY semi inf def of G and U in ghost}.
\end{lemma}

\begin{proof}
The result follows from direct computations.
\end{proof}

\vskip 1mm

Moreover, Lemma \ref{lem: sec N=2 of SUSY semi inf rel complex and G and U} below will be used later to check that the fields defining an $N = {\frac{3}{2}}$ superconformal structure on the supersymmetric semi-infinite complex belong to the relative supersymmetric semi-infinite complex.

\vskip 1mm

\begin{lemma}\label{lem: sec N=2 of SUSY semi inf rel complex and G and U}\rm
In ${\textstyle{\bigwedge}}^{{\frac{\infty}{2}} + {\bullet}}_{{\,\,} {\text{s}}}({\mathfrak{h}})$, we have:
\begin{gather*}
[{D}({\phi}^{x}{\phi}_{y}) {\,} {}_{\lambda} {\,} {G}_{\wedge}] 
= (-1)^{x + y}{\lambda}({\phi}^{x}{\phi}_{y}),
{\quad}
[{D}{\phi}_{x} {\,} {}_{\lambda} {\,} {G}_{\wedge}] 
= (-1)^{x + 1}{\lambda}{\phi}_{x},
\\
[{D}({\phi}^{x}{\phi}_{y}) {\,} {}_{\lambda} {\,} {U}_{\wedge}] 
= {\delta}_{xy}{\lambda},
{\quad}
[{D}{\phi}_{x} {\,} {}_{\lambda} {\,} {U}_{\wedge}] 
= 0,
\end{gather*}
for $x, y \in H$, where ${G}_{\wedge}$ and ${U}_{\wedge}$ are given in \eqref{eq: sec N=2 of SUSY semi inf def of G and U in ghost}.
\end{lemma}

\begin{proof}
The result can be verified by direct computations.
\end{proof}

\vskip 1mm

In the remainder of this section, assume further that ${\mathfrak{h}}$ is unimodular and that Assumption \ref{assum: second assumption} holds, that is, the orthogonal complement of ${\iota}({\mathfrak{h}})$ in ${\mathfrak{g}}$ with respect to ${\kappa}_{\mathfrak{g}}$ admits a vector superspace decomposition ${\iota}({\mathfrak{h}})^{\perp} = {\mathfrak{t}}^{+} \oplus {\mathfrak{t}}^{-}$ satisfying:
\begin{align*}
{\kappa}_{\mathfrak{g}}({\mathfrak{t}}^{\pm}, {\,} {\mathfrak{t}}^{\pm}) = 0,
{\quad}
[{\mathfrak{t}}^{\pm}, {\iota}({\mathfrak{h}}) \oplus {\mathfrak{t}}^{\pm}] \subseteq {\mathfrak{t}}^{\pm}.
\end{align*}
Under these assumptions, consider the super Manin pair structure on ${\mathfrak{g}} \oplus {\mathfrak{h}}$ given in \eqref{eq: sec sMP decomposition of example sMP2}, namely:
\begin{align*}
{\mathfrak{g}} \oplus {\mathfrak{h}} = {\mathfrak{a}}^{+}_{{\mathfrak{g}} \oplus {\mathfrak{h}}} \oplus {\mathfrak{a}}^{-}_{{\mathfrak{g}} \oplus {\mathfrak{h}}},
\end{align*}
where $\{ {y}_{i} \}_{i \in I}$ and $\{ {y}^{i} \}_{i \in I}$ are parity-homogeneous bases of ${\mathfrak{a}}^{+}_{{\mathfrak{g}} \oplus {\mathfrak{h}}}$ and ${\mathfrak{a}}^{-}_{{\mathfrak{g}} \oplus {\mathfrak{h}}}$, respectively, satisfying 
\begin{align*}
{\kappa}_{{\mathfrak{g}} \oplus {\mathfrak{h}}}({y}_{i}, {y}^{j}) = {\delta}_{ij},    
\end{align*}
for $i, j \in I$. Here, we follow the notation and constructions associated with the super Manin pair structure on ${\mathfrak{g}} \oplus {\mathfrak{h}}$ introduced in Section \ref{section: sMP} after Assumption \ref{assum: first assumption}.

\vskip 1mm

We also recall that the $N = {\frac{3}{2}}$ superconformal structure ${G}_{{\mathfrak{g}} \oplus {\mathfrak{h}}}$ and ${U}_{{\mathfrak{g}} \oplus {\mathfrak{h}}}$ on the super-affine vertex algebra ${V}^{{\kappa}_{{\mathfrak{g}} \oplus {\mathfrak{h}}}}_{{\,\,} {\text{s}}}({\mathfrak{g}} \oplus {\mathfrak{h}})$ are given by
\begin{align}
{G}_{{\mathfrak{g}} \oplus {\mathfrak{h}}} 
=& {\,} {\sum}_{i \in I} {\big(}(-1)^{i}{y}^{i}{\overline{y}}_{i} + {y}_{i}{\overline{y}}^{i}{\big)} + {\sum}_{i,j,k \in I} (-1)^{ki} {C}^{i}_{jk}{\overline{y}}_{i}({\overline{y}}^{j}{\overline{y}}^{k})
\notag
\\
&+ {\sum}_{i,j,k \in I} (-1)^{ij} {F}^{jk}_{i}{\overline{y}}^{i}({\overline{y}}_{j}{\overline{y}}_{k}) + {\frac{1}{3}}{\,}{\sum}_{i,j,k \in I} (-1)^{ij}{\mathring{F}}^{jk}_{i}{\overline{y}}_{i}({\overline{y}}_{j}{\overline{y}}_{k}),
\label{eq: sec N=2 of SUSY semi inf def of G and U in matter}
\\
{U}_{{\mathfrak{g}} \oplus {\mathfrak{h}}} 
=& {\,} {\sum}_{i \in I} (-1)^{i}{\overline{y}}^{i}{\overline{y}}_{i} + {\sum}_{i,j \in I} (-1)^{i}{F}^{ij}_{i}{y}_{j} + {\sum}_{i,j \in I} (-1)^{i}{C}^{i}_{ji}{y}^{j}.
\notag
\end{align}

\vskip 1mm

Combining the vectors in \eqref{eq: sec N=2 of SUSY semi inf def of G and U in matter} with the $N = 2$ superconformal structure on ${\textstyle{\bigwedge}}^{{\frac{\infty}{2}} + {\bullet}}_{{\,\,} {\text{s}}}({\mathfrak{h}})$ given in \eqref{eq: sec N=2 of SUSY semi inf def of G and U in ghost}, we introduce the following vectors
\begin{align}
\begin{split}
{G}_{\mathcal{C}} :=& {\,} {G}_{{\mathfrak{g}} \oplus {\mathfrak{h}}} \otimes {\ket{0}} + {\ket{0}} \otimes {G}_{\wedge},
\\
{U}_{\mathcal{C}} :=& {\,} {U}_{{\mathfrak{g}} \oplus {\mathfrak{h}}} \otimes {\ket{0}} + {\ket{0}} \otimes {U}_{\wedge},
\end{split}
\label{eq: sec N=2 of SUSY semi inf def of G and U in complex}
\end{align}
in the supersymmetric semi-infinite complex:
\begin{align}
{C}^{{\frac{\infty}{2}} + {\bullet}}_{{\,\,} {\text{s}}}({\mathfrak{h}}, {V}^{{\kappa}_{{\mathfrak{g}} \oplus {\mathfrak{h}}}}_{{\,\,} {\text{s}}}({\mathfrak{g}} \oplus {\mathfrak{h}})) 
= {V}^{{\kappa}_{{\mathfrak{g}} \oplus {\mathfrak{h}}}}_{{\,\,} {\text{s}}}({\mathfrak{g}} \oplus {\mathfrak{h}}) \otimes {\textstyle{\bigwedge}}^{{\frac{\infty}{2}} + {\bullet}}_{{\,\,} {\text{s}}}({\mathfrak{h}}).
\label{eq: sec N=2 of SUSY semi inf def of complex}
\end{align}
Then it follows from Theorem \ref{theorem: N=3/2 of saff} and Proposition \ref{prop: sec N=2 of SUSY semi inf N=2 shift of ghost} that ${G}_{\mathcal{C}}$ and ${U}_{\mathcal{C}}$ form an $N = {\frac{3}{2}}$ superconformal structure on the supersymmetric semi-infinite complex with central charge ${\frac{3}{2}}{\,}{\text{sdim}}{\,}{\mathfrak{g}} - {\frac{3}{2}}{\,}{\text{sdim}}{\,}{\mathfrak{h}} + {3}{\,}{\text{str}}({\mathcal{D}}_{{\mathfrak{g}} \oplus {\mathfrak{h}}})$. Moreover, in this case, we have
\begin{align}
{U}_{{\mathcal{C}} (0)}({G}_{{\mathcal{C}} (0)}{U}_{\mathcal{C}}) = -{G}_{\mathcal{C}} - {8}{\Phi}_{\mathcal{C}},
\label{eq: sec N=2 of SUSY semi inf UGU and -G -8Phi}
\end{align}
where ${\Phi}_{\mathcal{C}}$ is given by
\begin{align}
{\Phi}_{\mathcal{C}} = - {\frac{1}{6}}{\,}{\sum}_{i, j, k \in H} (-1)^{ij}{\mathring{F}}^{jk}_{i}{\overline{y}}_{i}({\overline{y}}_{j}{\overline{y}}_{k}) \otimes {\ket{0}}.
\label{eq: sec N=2 of SUSY semi inf constuction of Phi}
\end{align}
This follows from Proposition \ref{prop: sec N=3/2 computation of UGU}, Lemma \ref{lem: sec sMP F and ring F vanish}, and Lemma \ref{lem: sec sMP coeff E vanish}.

\vskip 1mm

Note that, since we are considering the ${V}^{0}_{{\,\,} {\text{s}}}({\mathfrak{h}})$-module action on ${V}^{{\kappa}_{{\mathfrak{g}} \oplus {\mathfrak{h}}}}_{{\,\,} {\text{s}}}({\mathfrak{g}} \oplus {\mathfrak{h}})$ given by \eqref{eq: sec N=2 of SUSY semi inf diag homo}, the differential ${d}_{\text{s}}$ on the complex in \eqref{eq: sec N=2 of SUSY semi inf def of complex} can be written as follows (see also \eqref{eq: sec SUSY semi inf cohom mode expression of differential} and \eqref{eq: sec sMP dual bases y for h}):
\begin{align}
\begin{split}
{d}_{\text{s}} 
=& {\,} {\sum}_{i \in H}{\sum}_{m \in {\mathbb{Z}}} (-1)^{i} {y}_{i (m)} \otimes {\phi}^{i}_{(-m-1)} 
- {\sum}_{i \in H}{\sum}_{m \in {\mathbb{Z}}} {\overline{y}}_{i (m)} \otimes {D}{\phi}^{i}_{(-m-1)}
\\
&- {\frac{1}{2}}{\,}{\sum}_{i, j, k \in H} (-1)^{jk + j}{h}^{i}_{jk} {\,}{\text{id}} \otimes ({D}({\phi}_{i}{\phi}^{j}{\phi}^{k}))_{(0)}.
\end{split}
\label{eq: sec N=2 of SUSY semi inf diff in complex}
\end{align}
Then we can show that ${G}_{\mathcal{C}}$ and ${U}_{\mathcal{C}}$ are BRST-invariant as follows:

\vskip 1mm

\begin{proposition}\label{prop: sec N=2 of SUSY semi inf G and U in kernel}\rm
In ${C}^{{\frac{\infty}{2}} + {\bullet}}_{{\,\,} {\text{s}}}({\mathfrak{h}}, {V}^{{\kappa}_{{\mathfrak{g}} \oplus {\mathfrak{h}}}}_{{\,\,} {\text{s}}}({\mathfrak{g}} \oplus {\mathfrak{h}}))$, we have:
\begin{align*}
{d}_{\text{s}}({G}_{\mathcal{C}}) = 0,
{\quad}
{d}_{\text{s}}({U}_{\mathcal{C}}) = 0,
\end{align*}
where ${G}_{\mathcal{C}}$ and ${U}_{\mathcal{C}}$ are given in \eqref{eq: sec N=2 of SUSY semi inf def of G and U in complex}.
\end{proposition}

\begin{proof}
First, by Lemma \ref{lem: sec sMP F and ring F vanish} and Lemma \ref{lem: sec N=3/2 OPE generator and U}, we have \eqref{eq: sec N=3/2 U and oveline y} and \eqref{eq: sec N=3/2 U and y}. Together with Lemma \ref{lem: sec N=2 of SUSY semi inf diff with G_gh and U_gh}, we obtain the following computation:
\begin{align*}
{d}_{\text{s}}({U}_{\mathcal{C}}) 
=& {\,} {\sum}_{i \in H}{\sum}_{m \in {\mathbb{Z}}_{\geq 0}} (-1)^{i} {y}_{i (m)}{U}_{{\mathfrak{g}} \oplus {\mathfrak{h}}} \otimes {\phi}^{i}_{(-m-1)}{\ket{0}}
\\
&- {\sum}_{i \in H}{\sum}_{m \in {\mathbb{Z}}_{\geq 0}} {\overline{y}}_{i (m)}{U}_{{\mathfrak{g}} \oplus {\mathfrak{h}}} \otimes {D}{\phi}^{i}_{(-m-1)}{\ket{0}}
\\
&+ {\sum}_{i \in H}{\sum}_{m \in {\mathbb{Z}}_{\geq 0}} (-1)^{i}{y}_{i (- m - 1)}{\ket{0}} \otimes {\phi}^{i}_{(m)}{U}_{\wedge}
\\
&- {\sum}_{i \in H}{\sum}_{m \in {\mathbb{Z}}_{\geq 0}} {\overline{y}}_{i (- m - 1)}{\ket{0}} \otimes {D}{\phi}^{i}_{(m)}{U}_{\wedge}
\\
&- {\frac{1}{2}}{\,}{\sum}_{i, j, k \in H} (-1)^{jk + j}{h}^{i}_{jk} {\,} {\ket{0}} \otimes ({D}({\phi}_{i}{\phi}^{j}{\phi}^{k}))_{(0)}{U}_{\wedge}
\\
=& {\,} - {\sum}_{i \in H} {\overline{y}}_{i} \otimes {D}{\phi}^{i} - {\sum}_{i \in H} {\overline{y}}_{i} \otimes (-{D}{\phi}^{i})
\\
&+ {\sum}_{i, j \in H} (-1)^{i}{h}^{i}_{ji}{\,} {\ket{0}} \otimes {\partial}{\phi}^{j}
\\
=& {\,} 0,
\end{align*}
where the last equality follows from the unimodularity of ${\mathfrak{h}}$. Similarly, by \eqref{eq: OPE G with generator of saff}, Lemma \ref{lem: sec N=2 of SUSY semi inf diff with G_gh and U_gh}, and the unimodularity of ${\mathfrak{h}}$, we obtain ${d}_{\text{s}}({G}_{\mathcal{C}}) = 0$.
\end{proof}

\vskip 1mm

We have shown that ${G}_{\mathcal{C}}$ and ${U}_{\mathcal{C}}$ form an $N = {\frac{3}{2}}$ superconformal structure on the complex in \eqref{eq: sec N=2 of SUSY semi inf def of complex}. Moreover, by Proposition \ref{prop: sec N=2 of SUSY semi inf G and U in kernel}, they lie in ${\text{ker}}{\,}{d}_{\text{s}}$. Thus, to obtain the $N = 2$ superconformal structure on the supersymmetric semi-infinite cohomology induced by ${G}_{\mathcal{C}}$ and ${U}_{\mathcal{C}}$, it remains to show that ${\Phi}_{\mathcal{C}}$ in \eqref{eq: sec N=2 of SUSY semi inf UGU and -G -8Phi} is BRST-trivial.

\vskip 1mm

\begin{lemma}\label{lem: sec N=2 of SUSY semi inf BRST triv of Phi}\rm
In the complex ${C}^{{\frac{\infty}{2}} + {\bullet}}_{{\,\,} {\text{s}}}({\mathfrak{h}}, {V}^{{\kappa}_{{\mathfrak{g}} \oplus {\mathfrak{h}}}}_{{\,\,} {\text{s}}}({\mathfrak{g}} \oplus {\mathfrak{h}}))$, we have
\begin{align*}
{\Phi}_{\mathcal{C}} = {d}_{\text{s}}{\big(}{\frac{1}{6}}{\,}{\sum}_{r, s, t \in H} (-1)^{st}{\mathring{F}}^{rs}_{t}{\overline{y}}_{r}{\overline{y}}_{s} \otimes {\phi}_{t}{\big)},
\end{align*}
where ${\Phi}_{\mathcal{C}}$ is given in \eqref{eq: sec N=2 of SUSY semi inf constuction of Phi}.
\end{lemma}

\begin{proof}
For $i, r, s \in H$, we have:
\begin{align}
[{y}_{i} {\,} {}_{\lambda} {\,} {\overline{y}}_{r}{\overline{y}}_{s}] 
= {\sum}_{w \in H} (-1)^{i}{C}^{w}_{ir}{\overline{y}}_{w}{\overline{y}}_{s} + {\sum}_{w \in H} (-1)^{ir}{C}^{w}_{is}{\overline{y}}_{r}{\overline{y}}_{w},
\label{eq: sec N=2 of SUSY semi inf y and bary bary}
\end{align}
in ${V}^{{\kappa}_{{\mathfrak{g}} \oplus {\mathfrak{h}}}}_{{\,\,} {\text{s}}}({\mathfrak{g}} \oplus {\mathfrak{h}})$. Hence, from \eqref{eq: sec N=2 of SUSY semi inf diff in complex}, \eqref{eq: sec N=2 of SUSY semi inf y and bary bary}, \eqref{eq: sec SUSY semi inf cohom phi phi phi and phi}, and \eqref{eq: sec sMP str const C}, we obtain:
\begin{align*}
{d}_{\text{s}}({\overline{y}}_{r}{\overline{y}}_{s} \otimes {\phi}_{t}) 
=& {\,} (-1)^{t(r + s) + 1}{\overline{y}}_{t}({\overline{y}}_{r}{\overline{y}}_{s}) \otimes {\ket{0}}
\\
&+ {\sum}_{i, w \in H} (-1)^{(i + 1)(r + s)}{C}^{w}_{ir}{\overline{y}}_{w}{\overline{y}}_{s} \otimes {\phi}^{i}{\phi}_{t}
\\
&+ {\sum}_{i, w \in H} (-1)^{is + i + r + s}{C}^{w}_{is}{\overline{y}}_{r}{\overline{y}}_{w} \otimes {\phi}^{i}{\phi}_{t}
\\
&+ {\sum}_{i, j \in H} (-1)^{ij + r + s + 1}{C}^{j}_{ti}{\overline{y}}_{r}{\overline{y}}_{s} \otimes {\phi}^{i}{\phi}_{j},
\end{align*}
for $r, s, t \in H$. Thus, using the identity obtained by applying \eqref{eq: sec sMP str const F zero} to \eqref{eq: JI B for MP} with respect to the super Manin pair structure \eqref{eq: sec sMP decomposition of example sMP1}, we have:
\begin{align*}
{d}_{\text{s}}{\big(}{\sum}_{r, s, t \in H} (-1)^{st}{\mathring{F}}^{rs}_{t}{\overline{y}}_{r}{\overline{y}}_{s} \otimes {\phi}_{t}{\big)} = {6}{\Phi}_{\mathcal{C}},
\end{align*}
which gives the desired result.
\end{proof}

\vskip 1mm

Therefore, we obtain the following theorem, which is the main result of this section:

\vskip 1mm

\begin{theorem}\label{theorem: sec N=2 of SUSY semi inf N=2 on cohom}\rm
Under Assumption \ref{assum: second assumption}, suppose that ${\mathfrak{h}}$ is unimodular. Then ${G}_{\mathcal{C}}$ and ${U}_{\mathcal{C}}$ given in \eqref{eq: sec N=2 of SUSY semi inf def of G and U in complex} induce an $N = 2$ superconformal structure on the supersymmetric semi-infinite cohomology
\begin{align}
{H}^{{\frac{\infty}{2}} + {\bullet}}_{{\,\,} {\text{s}}}({\mathfrak{h}}, {V}^{{\kappa}_{{\mathfrak{g}} \oplus {\mathfrak{h}}}}_{{\,\,} {\text{s}}}({\mathfrak{g}} \oplus {\mathfrak{h}})),
\label{eq: sec N=2 of SUSY semi inf def of cohom}
\end{align}
with central charge ${\frac{3}{2}}{\,}{\text{sdim}}{\,}{\mathfrak{g}} - {\frac{3}{2}}{\,}{\text{sdim}}{\,}{\mathfrak{h}} + {3}{\,}{\text{str}}({\mathcal{D}}_{{\mathfrak{g}} \oplus {\mathfrak{h}}})$.
\end{theorem}

\begin{proof}
We summarize the proof as follows. First, by Proposition \ref{prop: sec N=2 of SUSY semi inf G and U in kernel}, ${G}_{\mathcal{C}}$ and ${U}_{\mathcal{C}}$ lie in ${\text{ker}}{\,}{d}_{\text{s}}$. Moreover, by Theorem \ref{theorem: N=3/2 of saff} and Proposition \ref{prop: sec N=2 of SUSY semi inf N=2 shift of ghost}, ${G}_{\mathcal{C}}$ and ${U}_{\mathcal{C}}$ form an $N = {\frac{3}{2}}$ superconformal structure on the complex in \eqref{eq: sec N=2 of SUSY semi inf def of complex} with central charge ${\frac{3}{2}}{\,}{\text{sdim}}{\,}{\mathfrak{g}} - {\frac{3}{2}}{\,}{\text{sdim}}{\,}{\mathfrak{h}} + {3}{\,}{\text{str}}({\mathcal{D}}_{{\mathfrak{g}} \oplus {\mathfrak{h}}})$. In particular, by Lemma \ref{lem: sec sMP F and ring F vanish}, Lemma \ref{lem: sec sMP coeff E vanish}, Proposition \ref{prop: sec N=3/2 computation of UGU}, and Proposition \ref{prop: sec N=2 of SUSY semi inf N=2 shift of ghost}, we know that \eqref{eq: sec N=2 of SUSY semi inf UGU and -G -8Phi} holds. Finally, by \eqref{eq: sec N=2 of SUSY semi inf UGU and -G -8Phi} and Lemma \ref{lem: sec N=2 of SUSY semi inf BRST triv of Phi}, it follows that ${G}_{\mathcal{C}}$ and ${U}_{\mathcal{C}}$ satisfy all the conditions of Lemma \ref{lem: equiv def of N=2 sconf} on the cohomology in \eqref{eq: sec N=2 of SUSY semi inf def of cohom}.
\end{proof}

\vskip 1mm

That is, the Kazama--Suzuki $N = 2$ condition together with the unimodularity of ${\mathfrak{h}}$ gives rise to an ${N} = {\frac{3}{2}}$ superconformal symmetry on the super-affine vertex algebra, and when that symmetry is combined with a specific shift of the $N = 2$ superconformal structure on the ghost sector, it yields an $N = 2$ superconformal symmetry on the cohomology through the supersymmetric quantum Hamiltonian reduction.

\vskip 1mm

We now rewrite the $N = 2$ superconformal structure of Theorem \ref{theorem: sec N=2 of SUSY semi inf N=2 on cohom} on the supersymmetric semi-infinite cohomology without ghost terms. First, by \eqref{eq: sec SUSY semi inf cohom phi phi phi and phi} and \eqref{eq: sec N=2 of SUSY semi inf diff in complex}, we have:
\begin{align}
\begin{split}
{d}_{\text{s}}({\overline{y}}^{r} \otimes {\phi}_{r}) 
=& {\,} - {\ket{0}} \otimes {\phi}_{r}{D}{\phi}^{r} + (-1)^{r}{\overline{y}}^{r}{\overline{y}}_{r} \otimes {\ket{0}}
\\
&+ {\sum}_{i, j \in H} (-1)^{i + 1}{h}^{r}_{ji}{\overline{y}}^{j} \otimes {\phi}_{r}{\phi}^{i} 
+ {\sum}_{i, j \in H} (-1)^{j}{h}^{i}_{rj}{\overline{y}}^{r} \otimes {\phi}_{i}{\phi}^{j},
\end{split}
\label{eq: sec N=2 of SUSY semi inf image of bar y times phi}
\end{align}
for $r \in H$, since we know that:
\begin{align*}
[{y}_{i}, {y}^{r}] = {\sum}_{j \in H} {h}^{r}_{ji}{y}^{j},
\end{align*}
for $i, r \in H$. On the other hand, with respect to the super Manin pair structure \eqref{eq: sec sMP decomposition of example sMP1} on ${\iota}({\mathfrak{h}}) \oplus {\mathfrak{h}}$, consider the following $U(1)$-current, as given in \eqref{eq: def of U_g in saff}:
\begin{align*}
{U}_{{\iota}({\mathfrak{h}}) \oplus {\mathfrak{h}}} 
=& {\,} {\sum}_{i \in H} (-1)^{i}{\overline{y}}^{i}{\overline{y}}_{i} + {\sum}_{i \in H} (-1)^{i}[{y}_{i}, {y}^{i}]
\\
=& {\,} {\sum}_{i \in H} (-1)^{i}{\overline{y}}^{i}{\overline{y}}_{i},
\end{align*}
where the last equality follows from the unimodularity of ${\mathfrak{h}}$. Then, from \eqref{eq: sec N=2 of SUSY semi inf image of bar y times phi}, we have
\begin{align*}
{d}_{\text{s}}{\big(}{\sum}_{r \in H} {\overline{y}}^{r} \otimes {\phi}_{r}{\big)} = {U}_{{\iota}({\mathfrak{h}}) \oplus {\mathfrak{h}}} \otimes {\ket{0}} + {\ket{0}} \otimes {U}_{\wedge},
\end{align*}
and hence we obtain
\begin{align}
{U}_{\mathcal{C}} = ({U}_{{\mathfrak{g}} \oplus {\mathfrak{h}}} - {U}_{{\iota}({\mathfrak{h}}) \oplus {\mathfrak{h}}}) \otimes {\ket{0}} + {d}_{\text{s}}{\big(}{\sum}_{r \in H} {\overline{y}}^{r} \otimes {\phi}_{r}{\big)},
\label{eq: sec N=2 of SUSY semi inf without ghost part U_C}
\end{align}
which gives a representative of ${U}_{\mathcal{C}}$ in cohomology with no ghost contribution. For the $N = 1$ superconformal vector ${G}_{\mathcal{C}}$, one can use the Lie superalgebra version of the representative constructed in Section 4.3 of \cite{FS97A}.

\vskip 1mm

We conclude this section by observing that Theorem \ref{theorem: sec N=2 of SUSY semi inf N=2 on cohom} and \eqref{eq: sec N=2 of SUSY semi inf without ghost part U_C} also hold for the relative supersymmetric semi-infinite cohomology. First, note that the relative supersymmetric semi-infinite complex:
\begin{align}
{C}^{{\frac{\infty}{2}} + {\bullet}}_{{\,\,} {\text{s, rel}}}({\mathfrak{h}}, {V}^{{\kappa}_{{\mathfrak{g}} \oplus {\mathfrak{h}}}}_{{\,\,} {\text{s}}}({\mathfrak{g}} \oplus {\mathfrak{h}}))
\label{eq: sec N=2 of SUSY semi inf def of rel complex}
\end{align}
is a vertex subalgebra of the complex in \eqref{eq: sec N=2 of SUSY semi inf def of complex}. Indeed, by definition, the subcomplex in \eqref{eq: sec N=2 of SUSY semi inf def of rel complex} is the invariant subspace with respect to the zero modes of the following elements:
\begin{align*}
{y}_{i} \otimes {\ket{0}} + {\ket{0}} \otimes {J}_{\bigwedge}({u}_{i}),
{\quad}
{\ket{0}} \otimes {D}{\phi}_{i},
\end{align*}
for all $i \in H$. Moreover, by \eqref{eq: OPE G with generator of saff}, \eqref{eq: sec N=3/2 U and y}, and Lemma \ref{lem: sec N=2 of SUSY semi inf rel complex and G and U}, we also have:
\begin{align*}
{G}_{\mathcal{C}}, {U}_{\mathcal{C}} \in {C}^{{\frac{\infty}{2}} + {\bullet}}_{{\,\,} {\text{s, rel}}}({\mathfrak{h}}, {V}^{{\kappa}_{{\mathfrak{g}} \oplus {\mathfrak{h}}}}_{{\,\,} {\text{s}}}({\mathfrak{g}} \oplus {\mathfrak{h}})).
\end{align*}
On the other hand, from \eqref{eq: sec N=2 of SUSY semi inf y and bary bary} and the following OPE:
\begin{align*}
[{J}_{\bigwedge}({u}_{i}) {\,} {}_{\lambda} {\,} {\phi}_{t}] = {\sum}_{w \in H} {C}^{w}_{it}{\phi}_{w},
\end{align*}
for $i, t \in H$, we have:
\begin{align*}
{J}_{{u}_{i} (0)}({\overline{y}}_{r}{\overline{y}}_{s} \otimes {\phi}_{t}) 
=& {\,} {\sum}_{w \in H} (-1)^{i}{C}^{w}_{ir}{\overline{y}}_{w}{\overline{y}}_{s} \otimes {\phi}_{t} + {\sum}_{w \in H} (-1)^{ir}{C}^{w}_{is}{\overline{y}}_{r}{\overline{y}}_{w} \otimes {\phi}_{t}
\\
&+ {\sum}_{w \in H} (-1)^{ir + is}{C}^{w}_{it}{\overline{y}}_{r}{\overline{y}}_{s} \otimes {\phi}_{w},
\end{align*}
for $i, r, s, t \in H$. Therefore, using the identity obtained by applying \eqref{eq: sec sMP str const F zero} to \eqref{eq: JI B for MP} with respect to the super Manin pair structure \eqref{eq: sec sMP decomposition of example sMP1}, we have:
\begin{align*}
{J}_{{u}_{i} (0)}{\big(}{\sum}_{r, s, t \in H} (-1)^{st}{\mathring{F}}^{rs}_{t}{\overline{y}}_{r}{\overline{y}}_{s} \otimes {\phi}_{t}{\big)} = 0,
\end{align*}
for $i \in H$. Moreover, we have:
\begin{align*}
({\ket{0}} \otimes {D}{\phi}_{i})_{(0)}{\big(}{\sum}_{r, s, t \in H} (-1)^{st}{\mathring{F}}^{rs}_{t}{\overline{y}}_{r}{\overline{y}}_{s} \otimes {\phi}_{t}{\big)} = 0,
\end{align*}
for $i \in H$. Thus, ${G}_{\mathcal{C}}$ and ${U}_{\mathcal{C}}$ form an $N = {\frac{3}{2}}$ superconformal structure on the relative supersymmetric semi-infinite complex in \eqref{eq: sec N=2 of SUSY semi inf def of rel complex}, and both \eqref{eq: sec N=2 of SUSY semi inf UGU and -G -8Phi} and Lemma \ref{lem: sec N=2 of SUSY semi inf BRST triv of Phi} also hold in the relative version. Therefore, ${G}_{\mathcal{C}}$ and ${U}_{\mathcal{C}}$ induce an $N = 2$ superconformal structure on the relative supersymmetric semi-infinite cohomology. Finally, one can check that
\begin{align*}
{\sum}_{r \in H} {\overline{y}}^{r} \otimes {\phi}_{r} \in {C}^{{\frac{\infty}{2}} + {\bullet}}_{{\,\,} {\text{s, rel}}}({\mathfrak{h}}, {V}^{{\kappa}_{{\mathfrak{g}} \oplus {\mathfrak{h}}}}_{{\,\,} {\text{s}}}({\mathfrak{g}} \oplus {\mathfrak{h}})),
\end{align*}
so that \eqref{eq: sec N=2 of SUSY semi inf without ghost part U_C} also holds in the relative supersymmetric semi-infinite complex. Hence, we have shown the following result:

\vskip 1mm

\begin{corollary}\label{corollary: sec N=2 of SUSY semi inf N=2 on rel SUSY cohom}\rm
Under Assumption \ref{assum: second assumption}, suppose that ${\mathfrak{h}}$ is unimodular. Then the $N = {\frac{3}{2}}$ superconformal structure ${G}_{\mathcal{C}}$ and ${U}_{\mathcal{C}}$ on the relative supersymmetric semi-infinite complex in \eqref{eq: sec N=2 of SUSY semi inf def of rel complex} induce an $N = 2$ superconformal structure on the relative supersymmetric semi-infinite cohomology
\begin{align*}
{H}^{{\frac{\infty}{2}} + {\bullet}}_{{\,\,} {\text{s, rel}}}({\mathfrak{h}}, {V}^{{\kappa}_{{\mathfrak{g}} \oplus {\mathfrak{h}}}}_{{\,\,} {\text{s}}}({\mathfrak{g}} \oplus {\mathfrak{h}})),
\end{align*}
with central charge ${\frac{3}{2}}{\,}{\text{sdim}}{\,}{\mathfrak{g}} - {\frac{3}{2}}{\,}{\text{sdim}}{\,}{\mathfrak{h}} + {3}{\,}{\text{str}}({\mathcal{D}}_{{\mathfrak{g}} \oplus {\mathfrak{h}}})$.
\end{corollary}

\vskip 4mm
\section{Cohomological description of super coset models}\label{section: cohom description of super coset model}
\vskip 2mm

In this section, we investigate the supersymmetric semi-infinite cohomology of ${\widehat{\mathfrak{h}}}$ with coefficients in ${M}$ and its relative version, where the coefficient module is given as follows:
\begin{align*}
{M} := {V}^{{\kappa}_{\mathfrak{g}}}_{{\,\,} {\text{s}}}({\mathfrak{g}}) \otimes {V}^{-{\kappa}_{\mathfrak{h}}}_{{\,\,} {\text{s}}}({\mathfrak{h}}),
\end{align*}
by comparing them with the semi-infinite cohomology and the relative semi-infinite cohomology, respectively. Recall the definition of a coset vertex algebra. For a vertex subalgebra ${V}_{1}$ of ${V}_{2}$, the vertex subalgebra ${\text{Com}}({V}_{1}, {V}_{2})$ of ${V}_{2}$ defined by
\begin{align*}
{\text{Com}}({V}_{1}, {V}_{2}) := {\{} {v} \in {V}_{2} {\,} {|} {\,} {a}_{(n)}{v} = 0 {\text{ for all }} a \in {V}_{1}, n \in {\mathbb{Z}}_{\geq 0} {\}},
\end{align*}
is called the {\textit{coset}} (or {\textit{commutant}}) of ${V}_{1}$ in ${V}_{2}$. Then, the main result of this section shows that the {\textit{super coset model}}
\begin{align*}
{\text{Com}}({V}^{{\kappa}_{\mathfrak{h}}}_{{\,\,} {\text{s}}}({\mathfrak{h}}), {V}^{{\kappa}_{\mathfrak{g}}}_{{\,\,} {\text{s}}}({\mathfrak{g}}))
\end{align*}
can be realized as the relative supersymmetric semi-infinite cohomology of ${\widehat{\mathfrak{h}}}$ with coefficients in ${M}$.

\vskip 1mm

Throughout this section, we work in the following setting. Let ${\mathfrak{g}}$ be a finite-dimensional Lie superalgebra, and let ${\mathfrak{h}}$ be a Lie subalgebra of ${\mathfrak{g}}$. Suppose that ${\mathfrak{g}}$ is equipped with a nondegenerate consistent supersymmetric invariant bilinear form $({\,}{\cdot}{\,}, {\,}{\cdot}{\,})$ whose restriction to the Lie superalgebra ${\mathfrak{h}}$ is also nondegenerate. For a nonzero ${\kappa} \in {\mathbb{C}}$, consider the following bilinear forms
\begin{align*}
{\kappa}_{\mathfrak{g}} = {\kappa}({\,}{\cdot}{\,}, {\,}{\cdot}{\,}),
{\quad}
{\kappa}_{\mathfrak{h}} = {\kappa}({\,}{\cdot}{\,}, {\,}{\cdot}{\,})|_{{\mathfrak{h}} \times {\mathfrak{h}}},
\end{align*}
on ${\mathfrak{g}}$ and ${\mathfrak{h}}$, respectively, and we simply write ${V}^{\kappa}_{{\,\,} {\text{s}}}({\mathfrak{g}}) := {V}^{{\kappa}_{\mathfrak{g}}}_{{\,\,} {\text{s}}}({\mathfrak{g}})$ and ${V}^{-{\kappa}}_{{\,\,} {\text{s}}}({\mathfrak{h}}) := {V}^{-{\kappa}_{\mathfrak{h}}}_{{\,\,} {\text{s}}}({\mathfrak{h}})$. In addition, we denote the Killing forms of ${\mathfrak{g}}$ and ${\mathfrak{h}}$ by ${\overline{\kappa}}_{\mathfrak{g}}$ and ${\overline{\kappa}}_{\mathfrak{h}}$, respectively. Choose parity-homogeneous dual bases ${\{} {u}_{i} {\}}_{i \in G}$ and ${\{} {u}^{i} {\}}_{i \in G}$ of ${\mathfrak{g}}$ with respect to $({\,}{\cdot}{\,}, {\,}{\cdot}{\,})$ adapted to the orthogonal decomposition:
\begin{align*}
{\mathfrak{g}} = {\mathfrak{h}} \oplus {\mathfrak{h}}^{\perp},
\end{align*}
so that ${\{} {u}_i {\}}_{i \in H}$ and ${\{} {u}^{i} {\}}_{i \in H}$ are dual bases of ${\mathfrak{h}}$ with respect to $({\,}{\cdot}{\,}, {\,}{\cdot}{\,})|_{{\mathfrak{h}} \times {\mathfrak{h}}}$. Here, ${\mathfrak{h}}^{\perp}$ denotes the orthogonal complement of ${\mathfrak{h}}$ in ${\mathfrak{g}}$ with respect to the bilinear form $({\,}{\cdot}{\,}, {\,}{\cdot}{\,})$.

\vskip 1mm

We now recall the supersymmetric semi-infinite cohomology introduced in Section \ref{section: SUSY semi inf cohom}. For the coefficient module
\begin{align*}
{M} = {V}^{\kappa}_{{\,\,} {\text{s}}}({\mathfrak{g}}) \otimes {V}^{-{\kappa}}_{{\,\,} {\text{s}}}({\mathfrak{h}}),
\end{align*}
the differential ${d}_{\text{s}}$ on ${C}^{{\frac{\infty}{2}} + {\bullet}}_{{\,\,} {\text{s}}}({\mathfrak{h}}, {M}) := {M} \otimes {\textstyle{\bigwedge}}^{{\frac{\infty}{2}} + {\bullet}}_{{\,\,} {\text{s}}}({\mathfrak{h}})$ is given by:
\begin{align*}
{d}_{\text{s}}({\overline{g}} \otimes {\ket{0}}) &= {\sum}_{i \in H} (-1)^{(i + 1)(g + 1)}({\overline{[{u}_{i}, g]}} \otimes {\ket{0}}){\phi}^{i} + (-1)^{g + 1}{\kappa}{D}{\phi}^{g},
\\
{d}_{\text{s}}({\ket{0}} \otimes {\overline{h}}) &= {\sum}_{i \in H} (-1)^{(i + 1)(h + 1)}({\ket{0}} \otimes {\overline{[{u}_{i}, h]}}){\phi}^{i} + (-1)^{h}{\kappa}{D}{\phi}^{h},
\\
{d}_{\text{s}}({\phi}_{h}) &= - ({\overline{h}} \otimes {\ket{0}} + {\ket{0}} \otimes {\overline{h}}) + {\sum}_{i \in H} (-1)^{i}{\phi}^{i}{\phi}_{[{u}_{i}, h]},
\\
{d}_{\text{s}}({\phi}^{h}) &= {\frac{1}{2}}{\,}{\sum}_{i \in H} (-1)^{i}{\phi}^{i}{\phi}^{[{u}_{i}, h]},
\end{align*}
for $g \in {\mathfrak{g}}$ and $h \in {\mathfrak{h}}$. Here, we use the following notation:
\begin{align*}
{\phi}_{h} := {\sum}_{i \in H} {c}_{i}{\phi}_{i},
{\quad}
{\phi}^{h} := {\sum}_{i \in H} {c}^{i}{\phi}^{i},
\end{align*}
for $h = {\sum}_{i \in H} {c}_{i}{u}_{i} = {\sum}_{i \in H} {c}^{i}{u}^{i} \in {\mathfrak{h}}$. Moreover, for $g \in {\mathfrak{g}}$, we write ${\phi}^{g}$ for ${\phi}^{{\pi}_{\mathfrak{h}}(g)}$ where ${\pi}_{\mathfrak{h}} : {\mathfrak{g}} \rightarrow {\mathfrak{h}}$ is the projection map.

\vskip 1mm

To analyze the cohomology of ${C}^{{\frac{\infty}{2}} + {\bullet}}_{{\,\,} {\text{s}}}({\mathfrak{h}}, {M})$ using subcomplexes, let us consider the vertex subalgebras of ${C}^{{\frac{\infty}{2}} + {\bullet}}_{{\,\,} {\text{s}}}({\mathfrak{h}}, {M})$ denoted by ${C}^{\partial}_{1}$, ${C}^{\partial}_{2}$ and ${C}^{\partial}_{3}$. The first two vertex subalgebras are generated by ${\mathcal{B}}_{1}$ and ${\mathcal{B}}_{2}$, respectively, where
\begin{align*}
{\mathcal{B}}_{1} &=
{\Big\{} {\phi}_{a}, {\,} {\overline{a}} \otimes {\ket{0}} + {\ket{0}} \otimes {\overline{a}} + {\sum}_{i \in H} (-1)^{i + 1}{\phi}^{i}{\phi}_{[{u}_{i}, a]} {\,} {\Big|} {\,} a \in {\mathfrak{h}} {\Big\}},
\\
\begin{split}
{\mathcal{B}}_{2} &= \left\{
\begin{aligned}
&{\overline{a}} \otimes {\ket{0}} - {\ket{0}} \otimes {\overline{a}}, \\
&{D}{\phi}^{a} + {\frac{1}{2{\kappa}}}{\,}{\sum}_{i \in H} (-1)^{i(a + 1)}({\overline{[{u}_{i}, a]}} \otimes {\ket{0}} - {\ket{0}} \otimes {\overline{[{u}_{i}, a]}}){\phi}^{i}
\end{aligned}
{\,} \middle| {\,} a \in {\mathfrak{h}} \right\}.
\end{split}
\end{align*}
Since the ${\lambda}$-bracket of any two elements of ${\mathcal{B}}_{1}$, as well as that of any two elements of ${\mathcal{B}}_{2}$, is zero, it follows that ${C}^{\partial}_{1}$ and ${C}^{\partial}_{2}$ are differential supersymmetric algebras generated by ${\mathcal{B}}_{1}$ and ${\mathcal{B}}_{2}$, respectively. To define the subalgebra ${C}^{\partial}_{3}$, let us introduce the following notation:
\begin{align*}
{\widetilde{g}} := {g} - {\frac{1}{2{\kappa}}}{\,}{\sum}_{i \in G} (-1)^{i}{\overline{u}}^{i}{\overline{[{u}_{i}, g]}} \in {V}^{\kappa}_{{\,\,} {\text{s}}}({\mathfrak{g}}),
\end{align*}
for $g \in {\mathfrak{g}}$, and
\begin{align*}
X(h) := {\ket{0}} \otimes {h} + {X}_{1}(h) \otimes {\ket{0}} + {X}_{2}(h),
\end{align*}
for $h \in {\mathfrak{h}}$, where
\begin{align*}
{X}_{1}(h) &:= {\frac{1}{2\kappa}}{\,}{\sum}_{i\in H} (-1)^{i} {\overline{u}}^{i}{\overline{[{u}_{i}, h]}} \in {V}^{\kappa}_{{\,\,} {\text{s}}}({\mathfrak{g}}),
\\
{X}_{2}(h) &:= {\sum}_{i\in H} (-1)^{i + 1} ({D}{\phi}^{i}){\phi}_{[{u}_{i}, h]} \in {\textstyle{\bigwedge}}^{{\frac{\infty}{2}} + {\bullet}}_{{\,\,} {\text{s}}}({\mathfrak{h}}).
\end{align*}
Now, let ${C}^{\partial}_{3}$ be the differential subalgebra of ${C}^{{\frac{\infty}{2}} + {\bullet}}_{{\,\,} {\text{s}}}({\mathfrak{h}}, {M})$ generated by 
\begin{align*}
{\mathcal{B}}_{3} = {\Big\{} {\widetilde{g}} \otimes {\ket{0}}, {\,} X(h), {\,} {\overline{h}}^{\perp} \otimes {\ket{0}}, {\,} {\phi}^{i}, {D}{\phi}_{i} {\,} {\Big|} {\,} {g} \in {\mathfrak{g}}, {\,} {h} \in {\mathfrak{h}}, {\,} {h}^{\perp} \in {\mathfrak{h}}^{\perp}, {\,} {i} \in {H} {\big\}},
\end{align*}
under the even derivation ${\partial}$. Then we have $[{\overline{g}}_{1} {\,} {}_{\lambda} {\,} {\widetilde{g}}_{2}] = 0$ for ${g}_{1}, {g}_{2} \in {\mathfrak{g}}$. Moreover, ${\{} {\widetilde{g}} {\}}_{g \in {\mathfrak{g}}}$ generates the affine vertex algebra ${V}^{{\kappa} - {\frac{1}{2}}{\overline{\kappa}}_{\mathfrak{g}}}({\mathfrak{g}})$, that is:
\begin{align*}
[{\widetilde{g}}_{1} {\,} {}_{\lambda} {\,} {\widetilde{g}}_{2}] 
= {\widetilde{[{g}_{1}, {g}_{2}]}} + ({\kappa} - {\frac{1}{2}}{\overline{\kappa}}_{\mathfrak{g}})({g}_{1}, {g}_{2}){\lambda},
\end{align*}
for ${g}_{1}, {g}_{2} \in {\mathfrak{g}}$. Hence, we have an isomorphism of vertex algebras
\begin{align*}
{V}^{{\kappa}}_{{\,\,} {\text{s}}}({\mathfrak{g}}) \cong {V}^{{\kappa} - {\frac{1}{2}}{\overline{\kappa}}_{\mathfrak{g}}}({\mathfrak{g}}) \otimes {\mathcal{F}}({\overline{\mathfrak{g}}}),
\end{align*}
where ${\mathcal{F}}({\overline{\mathfrak{g}}})$ denotes the vertex subalgebra generated by ${\{} {\overline{g}} {\}}_{g \in {\mathfrak{g}}}$.

\vskip 1mm

Now, let ${\textstyle{\bigwedge}}^{{\frac{\infty}{2}} + {\bullet}}({\mathfrak{h}})$ be the vertex subalgebra of ${\textstyle{\bigwedge}}^{{\frac{\infty}{2}} + {\bullet}}_{{\,\,} {\text{s}}}({\mathfrak{h}})$ generated by $\{ {D}{\phi}_{i}, {\phi}^{i} \}_{i \in H}$. Then we obtain the following lemmas:

\vskip 1mm

\begin{lemma}\label{lem: sec cohom desc VA iso of C3}\rm
As vertex algebras, we have
\begin{align*}
{C}^{\partial}_{3} \cong 
{V}^{{\kappa} - {\frac{1}{2}}{\overline{\kappa}}_{\mathfrak{g}}}({\mathfrak{g}}) 
\otimes 
{\mathcal{F}}({\overline{\mathfrak{h}}}^{\perp}) 
\otimes 
{V}^{- {\kappa} - {\frac{1}{2}}{\overline{\kappa}}_{\mathfrak{h}}}({\mathfrak{h}}) 
\otimes 
{\textstyle{\bigwedge}}^{{\frac{\infty}{2}} + {\bullet}}({\mathfrak{h}}).
\end{align*}
\end{lemma}

\begin{proof}
By direct computations, one can check that
\begin{align*}
[X({h}_{1}) {\,} {}_{\lambda} {\,} X({h}_{2})] = X([{h}_{1}, {h}_{2}]) - ({\kappa} + {\frac{1}{2}}{\overline{\kappa}}_{\mathfrak{h}})({h}_{1}, {h}_{2}){\lambda},
\end{align*}
since we have
\begin{align*}
[{X}_{1}({h}_{1}) {\,} {}_{\lambda} {\,} {X}_{1}({h}_{2})] 
&= {X}_{1}([{h}_{1}, {h}_{2}]) + {\frac{1}{2}}{\overline{\kappa}}_{\mathfrak{h}}({h}_{1}, {h}_{2}){\lambda},
\\
[{X}_{2}({h}_{1}) {\,} {}_{\lambda} {\,} {X}_{2}({h}_{2})] 
&= {X}_{2}([{h}_{1}, {h}_{2}]) - {\overline{\kappa}}_{\mathfrak{h}}({h}_{1}, {h}_{2}){\lambda},
\end{align*}
for ${h}_{1}, {h}_{2} \in {\mathfrak{h}}$. Hence the map defined by
\begin{gather*}
{\widetilde{g}} \otimes {\ket{0}} \mapsto {g} \otimes {\ket{0}} \otimes {\ket{0}},
{\quad}
{X(h)} \mapsto {\ket{0}} \otimes {\ket{0}} \otimes {h},
\\
{\overline{h}}^{\perp} \otimes {\ket{0}} \mapsto {\ket{0}} \otimes {\overline{h}}^{\perp} \otimes {\ket{0}},
{\quad}
{\phi}^{i} \mapsto {\phi}^{i},
{\quad}
{D}{\phi}_{i} \mapsto {D}{\phi}_{i},
\end{gather*}
for ${g} \in {\mathfrak{g}}$, ${h} \in {\mathfrak{h}}$, ${h}^{\perp} \in {\mathfrak{h}}^{\perp}$, $i \in H$, is the desired isomorphism.
\end{proof}

\vskip 1mm

\begin{lemma}\label{lem: sec cohom desc SUSY decomposition of SUSY complex into subcomplexes}\rm
As vector superspaces, we have
\begin{align*}
{C}^{{\frac{\infty}{2}} + {\bullet}}_{{\,\,} {\text{s}}}({\mathfrak{h}}, {M}) 
\cong
{C}^{\partial}_{1} \otimes {C}^{\partial}_{2} \otimes {C}^{\partial}_{3}.
\end{align*}
\end{lemma}

\begin{proof}
As a differential algebra with respect to ${\partial}$, we know that ${C}^{{\frac{\infty}{2}} + {\bullet}}_{{\,\,} {\text{s}}}({\mathfrak{h}}, {M})$ is generated by
\begin{align*}
{\overline{g}} \otimes {\ket{0}},
{\quad}
{g} \otimes {\ket{0}},
{\quad}
{\ket{0}} \otimes {\overline{h}},
{\quad}
{\ket{0}} \otimes {h},
{\quad}
{\phi}_{i},
{\quad}
{D}{\phi}_{i},
{\quad}
{\phi}^{i},
{\quad}
{D}{\phi}^{i},
\end{align*}
for ${g} \in {\mathfrak{g}}$, ${h} \in {\mathfrak{h}}$, $i \in H$. Then, one can check that the above generators lie in the differential algebra generated by ${\mathcal{B}}_{1}$, ${\mathcal{B}}_{2}$, and ${\mathcal{B}}_{3}$. Hence, by the PBW basis argument, we obtain the desired result.
\end{proof}

\vskip 1mm

From now on, we will show that the differential ${d}_{\text{s}}$ restricts to ${C}^{\partial}_{3}$ and understand the supersymmetric semi-infinite cohomology
\begin{align*}
{H}^{{\frac{\infty}{2}} + {\bullet}}_{{\,\,} {\text{s}}}({\mathfrak{h}}, {M})
\end{align*}
through the cohomology ${H}({C}^{\partial}_{3}, {d}_{\text{s}})$. On ${C}^{\partial}_{3}$, the differential ${d}_{\text{s}}$ acts as follows:
\begin{align}
{d}_{\text{s}}({\widetilde{g}} \otimes {\ket{0}}) &= {\sum}_{i \in H} (-1)^{ig + i + g}({\widetilde{[{u}_{i}, {g}]}} \otimes {\ket{0}}){\phi}^{i} + {\sum}_{i \in H}(-1)^{g} ({\kappa} - {\frac{1}{2}}{\overline{\kappa}}_{\mathfrak{g}})({u}_{i}, {g}){\partial}{\phi}^{i},
\notag\\
{d}_{\text{s}}({\overline{h}}^{\perp} \otimes {\ket{0}}) &= {\sum}_{i \in H} (-1)^{(i + 1)({h}^{\perp} + 1)}({\overline{[{u}_{i}, {h}^{\perp}]}} \otimes {\ket{0}}){\phi}^{i},
\notag\\
{d}_{\text{s}}(X(h)) &= {\sum}_{i \in H} (-1)^{ih + i + h} X([{u}_{i}, {h}]){\phi}^{i} + {\sum}_{i \in H} (-1)^{h}(-{\kappa} - {\frac{1}{2}}{\overline{\kappa}}_{\mathfrak{h}})({u}_{i}, {h}){\partial}{\phi}^{i},
\label{eq: sec cohom desc d_s on C3}\\
{d}_{\text{s}}({\phi}^{h}) &= {\frac{1}{2}}{\,}{\sum}_{i \in H} (-1)^{i}{\phi}^{i}{\phi}^{[{u}_{i}, h]},
\notag\\
{d}_{\text{s}}({D}{\phi}_{h}) &= {\widetilde{h}} \otimes {\ket{0}} + X(h) + {\frac{1}{2{\kappa}}}{\,}{\sum}_{i \in {G} \setminus {H}} (-1)^{i} {\overline{u}}^{i}{\overline{[{u}_{i}, h]}} \otimes {\ket{0}} + {\sum}_{i \in H} {\phi}^{i}{D}{\phi}_{[{u}_{i}, h]},
\notag
\end{align}
for ${g} \in {\mathfrak{g}}$, ${h} \in {\mathfrak{h}}$, ${h}^{\perp} \in {\mathfrak{h}}^{\perp}$. Based on this computation, we can deduce the following proposition.

\vskip 1mm

\begin{proposition}\label{prop: sec cohom desc SUSY H(M) to H(C3)}\rm
We have ${H}({C}^{\partial}_{1}, {d}_{\text{s}}) = {H}({C}^{\partial}_{2}, {d}_{\text{s}}) = {\mathbb{C}}$ and hence
\begin{align*}
{H}^{{\frac{\infty}{2}} + {\bullet}}_{{\,\,} {\text{s}}}({\mathfrak{h}}, {M}) \cong {H}^{\bullet}({C}^{\partial}_{3}, {d}_{\text{s}}).
\end{align*}
\end{proposition}

\begin{proof}
Recall that ${C}^{\partial}_{1}$ is the differential algebra generated by ${\mathcal{B}}_{1}$, which consists of ${\phi}_{h}$ and ${d}_{\text{s}}({\phi}_{h})$ for $h \in {\mathfrak{h}}$. Hence ${d}_{\text{s}}({C}^{\partial}_{1}) \subset {C}^{\partial}_{1}$. Since ${C}^{\partial}_{1}$ is a differential supersymmetric algebra, by the K{\"u}nneth theorem, we have
\begin{align*}
{H}({C}^{\partial}_{1}, {d}_{\text{s}}) = S({\mathbb{C}}[\partial] \otimes {H}({\text{span}}_{\mathbb{C}}{\mathcal{B}}_{1}, {d}_{\text{s}})) = {\mathbb{C}}.
\end{align*}
Similarly, ${C}^{\partial}_{2}$ is the differential supersymmetric algebra generated by
\begin{align*}
{\overline{h}} \otimes {\ket{0}} - {\ket{0}} \otimes {\overline{h}},
{\quad}
{d}_{\text{s}}({\overline{h}} \otimes {\ket{0}} - {\ket{0}} \otimes {\overline{h}})
\end{align*}
for $h \in {\mathfrak{h}}$. Therefore, by the same argument as above, we have:
\begin{align*}
{H}({C}^{\partial}_{2}, {d}_{\text{s}}) = S({\mathbb{C}}[\partial] \otimes {H}({\text{span}}_{\mathbb{C}}{\mathcal{B}}_{2}, {d}_{\text{s}})) = {\mathbb{C}}.
\end{align*}
Finally, by \eqref{eq: sec cohom desc d_s on C3}, we also have ${d}_{\text{s}}({C}^{\partial}_{3}) \subset {C}^{\partial}_{3}$. Hence, the isomorphism of Lemma \ref{lem: sec cohom desc SUSY decomposition of SUSY complex into subcomplexes} is an isomorphism of complexes, and the desired result follows from the K{\"u}nneth theorem.
\end{proof}

\vskip 1mm

We now aim to establish an analogue of Proposition \ref{prop: sec cohom desc SUSY H(M) to H(C3)} for the relative supersymmetric semi-infinite cohomology. For this purpose, we adapt the contracting homotopy argument of \cite{LL07} and the chiral chain homotopy argument of \cite{LLS10}. To this end, we will construct a homotopy equivalence between
\begin{align*}
{C}^{{\frac{\infty}{2}} + {\bullet}}_{{\,\,} {\text{s}}}({\mathfrak{h}}, {M}) = {C}^{\partial}_{1} \otimes {C}^{\partial}_{2} \otimes {C}^{\partial}_{3}
\end{align*}
and ${C}^{\partial}_{3}$. Define the following linear maps on the PBW monomial bases:
\begin{align}
P : {C}^{\partial}_{1} \otimes {C}^{\partial}_{2} \otimes {C}^{\partial}_{3} \rightarrow {C}^{\partial}_{3},
{\quad}
P({\ket{0}} \otimes {\ket{0}} \otimes {v}_{3}) = {v}_{3},
{\quad}
P({v}_{1} \otimes {v}_{2} \otimes {v}_{3}) = 0,
\label{eq: sec cohom desc def of map P}
\end{align}
for ${v}_{1} \neq {\ket{0}}$ or ${v}_{2} \neq {\ket{0}}$, and
\begin{align}
I : {C}^{\partial}_{3} \rightarrow {C}^{\partial}_{1} \otimes {C}^{\partial}_{2} \otimes {C}^{\partial}_{3},
{\quad}
I({v}_{3}) = {\ket{0}} \otimes {\ket{0}} \otimes {v}_{3},
\label{eq: sec cohom desc def of map I}
\end{align}
where each ${v}_{i}$ is a PBW monomial in ${C}^{\partial}_{i}$ for $i = 1, 2, 3$. Then we have
\begin{align}
{\text{id}}_{{C}^{\partial}_{3}} - {P} \circ {I} = 0.
\label{eq: sec cohom desc homotopy eq1}
\end{align}

\vskip 1mm

For the other direction, we will construct a cochain homotopy between ${I} \circ {P}$ and ${\text{id}}_{{C}^{\partial}_{1} \otimes {C}^{\partial}_{2} \otimes {C}^{\partial}_{3}}$. To simplify the notation, we set
\begin{gather*}
{A}^{1}_{i} := {\phi}_{i},
{\quad}
{B}^{1}_{i} := {d}_{\text{s}}({\phi}_{i}),
\\
{A}^{2}_{i} :={\overline{u}_{i}} \otimes {\ket{0}} - {\ket{0}} \otimes {\overline{u}_{i}},
{\quad}
{B}^{2}_{i} := {d}_{\text{s}}({\overline{u}_{i}} \otimes {\ket{0}} - {\ket{0}} \otimes {\overline{u}_{i}}),
\end{gather*}
for $i \in H$. Here, we use the same notation for these elements and their canonical images in ${C}^{\partial}_{1} \otimes {C}^{\partial}_{2}$. Then we know that ${C}^{\partial}_{1} \otimes {C}^{\partial}_{2}$ is the supersymmetric algebra generated by ${\partial}^{m}{A}^{r}_{i}$ and ${\partial}^{m}{B}^{r}_{i}$ for $m \in {\mathbb{Z}}_{\geq 0}$, $r = 1, 2$, and $i \in H$. Let $({C}^{\partial}_{1} \otimes {C}^{\partial}_{2})[{\ell}]$ denote the subspace of ${C}^{\partial}_{1} \otimes {C}^{\partial}_{2}$ spanned by PBW monomials of length ${\ell}$. Then we have:
\begin{align*}
{C}^{\partial}_{1} \otimes {C}^{\partial}_{2} = {\bigoplus}_{{\ell} \in {\mathbb{Z}}_{\geq 0}} ({C}^{\partial}_{1} \otimes {C}^{\partial}_{2})[{\ell}].
\end{align*}
Now, define an odd derivation ${\widetilde{\Theta}}_{12} : {C}^{\partial}_{1} \otimes {C}^{\partial}_{2} \rightarrow {C}^{\partial}_{1} \otimes {C}^{\partial}_{2}$ by
\begin{align*}
{\partial}^{m}{A}^{r}_{i} \mapsto 0,
{\quad}
{\partial}^{m}{B}^{r}_{i} \mapsto {\partial}^{m}{A}^{r}_{i},
\end{align*}
for $m \in {\mathbb{Z}}_{\geq 0}$, $r = 1, 2$, $i \in H$. Using this derivation, define the linear map
\begin{align*}
{\Theta}_{12} := {0}{|}_{{\mathbb{C}}({\ket{0}} \otimes {\ket{0}})} \oplus {\big(}{\bigoplus}_{{\ell} \geq 1} {\frac{1}{\ell}} {\widetilde{\Theta}}_{12}[{\ell}]{\big)},
\end{align*}
where ${\widetilde{\Theta}}_{12}[{\ell}]$ denotes the restriction of ${\widetilde{\Theta}}_{12}$ to $({C}^{\partial}_{1} \otimes {C}^{\partial}_{2})[{\ell}]$. Then we can construct a homotopy operator ${\Theta}$ on ${C}^{\partial}_{1} \otimes {C}^{\partial}_{2} \otimes {C}^{\partial}_{3}$ by
\begin{align*}
{v}_{1} \otimes {v}_{2} \otimes {v}_{3} \mapsto {\Theta}_{12}({v}_{1} \otimes {v}_{2}) \otimes {v}_{3},
\end{align*}
that is, ${\Theta} = {\Theta}_{12} \otimes {\text{id}}_{{C}^{\partial}_{3}}$, where ${v}_{i} \in {C}^{\partial}_{i}$ for $i = 1, 2, 3$.

\vskip 1mm

\begin{proposition}\label{prop: sec cohom desc homotopy equiv for non-rel}\rm
The two complexes ${C}^{{\frac{\infty}{2}} + {\bullet}}_{{\,\,} {\text{s}}}({\mathfrak{h}}, {M}) = {C}^{\partial}_{1} \otimes {C}^{\partial}_{2} \otimes {C}^{\partial}_{3}$ and ${C}^{\partial}_{3}$ are homotopy equivalent.
\end{proposition}

\begin{proof}
It remains to show that
\begin{align}
{\text{id}}_{{C}^{\partial}_{1} \otimes {C}^{\partial}_{2} \otimes {C}^{\partial}_{3}} - {I} \circ {P} = {d}_{\text{s}} \circ {\Theta} + {\Theta} \circ {d}_{\text{s}},
\label{eq: sec cohom desc homotopy eq2}
\end{align}
since \eqref{eq: sec cohom desc homotopy eq1} has already been established. Suppose that
\begin{align*}
{A}_{j} \in {\{} {\partial}^{m}{A}^{1}_{i}, {\,} {\partial}^{m}{A}^{2}_{i} {\,} {|} {\,} m \in {\mathbb{Z}}_{\geq 0}, {\,} {i} \in H {\}},
\end{align*}
and set ${B}_{j} = {d}_{\text{s}}({A}_{j})$ for $j = 1, \cdots, r, {\widetilde{1}}, \cdots, {\widetilde{s}}$. Then, for the monomial
\begin{align*}
{v}_{12} = {A}_{1}{A}_{2} \cdots {A}_{r}{B}_{\widetilde{1}}{B}_{\widetilde{2}} \cdots {B}_{\widetilde{s}} \in {C}^{\partial}_{1} \otimes {C}^{\partial}_{2},
\end{align*}
a direct computation gives the following equalities:
\begin{align*}
({d}_{12} \circ {\Theta}_{12})({v}_{12}) 
&= {\frac{s}{r + s}}{v}_{12} + {\frac{1}{r + s}}(-1)^{{\overline{1}} + \cdots + {\overline{r}}}{v}_{0},
\\
({\Theta}_{12} \circ {d}_{12})({v}_{12}) 
&= {\frac{r}{r + s}}{v}_{12} - {\frac{1}{r + s}}(-1)^{{\overline{1}} + \cdots + {\overline{r}}}{v}_{0},
\end{align*}
where ${v}_{0}$ is given by
\begin{align*}
{v}_{0} = {\big(} {\sum}^{r}_{j = 1} &(-1)^{{\overline{1}} + \cdots + {\overline{j-1}}} {A}_{1} \cdots {A}_{j-1}{B}_{j}{A}_{j+1} \cdots {A}_{r} {\big)}
\\
& \times {\big(} {\sum}^{s}_{k = 1} (-1)^{{\widetilde{1}} + \cdots + {\widetilde{k-1}} + (k-1)}{B}_{\widetilde{1}} \cdots {B}_{\widetilde{k-1}}{A}_{\widetilde{k}}{B}_{\widetilde{k+1}} \cdots {B}_{\widetilde{s}} {\big)}.
\end{align*}
Here, ${\overline{i}}$ and ${\widetilde{i}}$ denote the parities of ${A}_{i}$ and ${A}_{\widetilde{i}}$, respectively. Consequently, we have
\begin{align*}
({d}_{12} \circ {\Theta}_{12} + {\Theta}_{12} \circ {d}_{12})(v) = v,
\end{align*}
for ${v} \in {\bigoplus}_{{\ell} \geq 1} ({C}^{\partial}_{1} \otimes {C}^{\partial}_{2})[{\ell}]$. On the other hand, we have:
\begin{align*}
{d}_{12} \circ {\Theta}_{12} + {\Theta}_{12} \circ {d}_{12} = 0
\end{align*}
on ${\mathbb{C}}({\ket{0}} \otimes {\ket{0}})$. Therefore, using the identities above, together with the fact that
\begin{align*}
{d}_{\text{s}}({v}_{1} \otimes {v}_{2} \otimes {v}_{3}) 
= {d}_{12}({v}_{1} \otimes {v}_{2}) \otimes {v}_{3} + (-1)^{{v}_{1} + {v}_{2}} ({v}_{1} \otimes {v}_{2}) \otimes {d}_{\text{s}}({v}_{3}),
\end{align*}
where ${d}_{12}$ on $({C}^{\partial}_{1} \otimes {C}^{\partial}_{2})$ is given by
\begin{align*}
{d}_{12}({v}_{1} \otimes {v}_{2}) = {d}_{\text{s}}({v}_{1}) \otimes {v}_{2} + (-1)^{{v}_{1}}{v}_{1} \otimes {d}_{\text{s}}({v}_{2}),
\end{align*}
we obtain the desired result, that is, \eqref{eq: sec cohom desc homotopy eq2} holds.
\end{proof}

\vskip 1mm

Now, we want to show that the homotopy equivalence in Proposition \ref{prop: sec cohom desc homotopy equiv for non-rel} induces a homotopy equivalence between the relative supersymmetric semi-infinite complex ${C}^{{\frac{\infty}{2}} + {\bullet}}_{{\,\,} {\text{s, rel}}}({\mathfrak{h}}, {M})$ and the subcomplex
\begin{align*}
{C}^{\partial}_{\text{3, rel}} 
:=& {\,} {C}^{{\frac{\infty}{2}} + {\bullet}}_{{\,\,} {\text{s, rel}}}({\mathfrak{h}}, {M}) \cap {C}^{\partial}_{3}
\\
=& {\,} {\{} v \in {C}^{\partial}_{3} {\,} {|} {\,} {J}_{{u}_{i} (0)}{v} = {D}{\phi}_{i (0)}{v} = 0 {\text{ for all }} {i} \in H {\}}.
\end{align*}
Note that the differential ${d}_{\text{s}}$ restricts to a differential on ${C}^{\partial}_{\text{3, rel}}$.

\vskip 1mm

\begin{lemma}\label{lem: sec cohom desc C_i closed to rel action and Theta commute with rel action}\rm
We have the following properties:
\begin{align}
&({J}_{u})_{(0)}({C}^{\partial}_{i}) \subset {C}^{\partial}_{i} {\text{ and }} ({D}{\phi}_{u})_{(0)}({C}^{\partial}_{i}) \subset {C}^{\partial}_{i} {\text{ for }} i = 1, 2, 3 {\text{ and }} u \in {\mathfrak{h}},
\label{eq: sec cohom desc C_i closed to rel action}
\\
&({J}_{u})_{(0)} \circ {\Theta} = (-1)^{u}{\Theta} \circ ({J}_{u})_{(0)} {\text{ and }} ({D}{\phi}_{u})_{(0)} \circ {\Theta} = (-1)^{u + 1}{\Theta} \circ ({D}{\phi}_{u})_{(0)} {\text{ for }} u \in {\mathfrak{h}}.
\label{eq: sec cohom desc Theta commute with rel action}
\end{align}
\end{lemma}

\begin{proof}
First, \eqref{eq: sec cohom desc C_i closed to rel action} follows directly from the following computations. Indeed, for ${u} \in {\mathfrak{h}}$, the operator $({J}_{u})_{(0)}$ acts on the generators as follows:
\begin{align}
&{\qquad} {\qquad} {\qquad} {\quad}
{\phi}_{h} \mapsto {\phi}_{[u, h]},
{\quad}
{d}_{\text{s}}({\phi}_{h}) \mapsto (-1)^{u}{d}_{\text{s}}({\phi}_{[u, h]}),
\notag
\\
&{\qquad} {\qquad}
{\overline{h}} \otimes {\ket{0}} - {\ket{0}} \otimes {\overline{h}} \mapsto (-1)^{u}({\overline{[{u}, {h}]}} \otimes {\ket{0}} - {\ket{0}} \otimes {\overline{[{u}, {h}]}}),
\notag
\\
&{\qquad} {\qquad} 
{d}_{\text{s}}({\overline{h}} \otimes {\ket{0}} - {\ket{0}} \otimes {\overline{h}}) \mapsto {d}_{\text{s}}({\overline{[{u}, {h}]}} \otimes {\ket{0}} - {\ket{0}} \otimes {\overline{[{u}, {h}]}}),
\label{eq: sec cohom desc J-action}
\\
&{\widetilde{g}} \otimes {\ket{0}} \mapsto {\widetilde{[u, g]}} \otimes {\ket{0}},
{\quad}
{X}(h) \mapsto {X}([u, h]),
{\quad}
{\overline{h}}^{\perp} \otimes {\ket{0}} \mapsto (-1)^{u}{\overline{[{u}, {h}^{\perp}]}} \otimes {\ket{0}},
{\quad}
\notag
\\
&{\qquad} {\qquad} {\qquad} {\qquad} {\quad}
{\phi}^{h} \mapsto {\phi}^{[u, h]},
{\quad}
{D}{\phi}_{h} \mapsto (-1)^{u}{D}{\phi}_{[u, h]},
\notag
\end{align}
for ${h} \in {\mathfrak{h}}$, ${h}^{\perp} \in {\mathfrak{h}}^{\perp}$, ${g} \in {\mathfrak{g}}$. The map $(D\phi_u)_{(0)}$ is given by 
\begin{align}
& {\qquad} {\qquad}
{\phi}_{h} \mapsto 0,
{\quad}
{d}_{\text{s}}({\phi}_{h}) \mapsto (-1)^{u}{\phi}_{[u, h]},
{\quad}
{\overline{h}} \otimes {\ket{0}} - {\ket{0}} \otimes {\overline{h}} \mapsto 0,
\notag
\\
&{d}_{\text{s}}({\overline{h}} \otimes {\ket{0}} - {\ket{0}} \otimes {\overline{h}}) \mapsto {\overline{[{u}, {h}]}} \otimes {\ket{0}} - {\ket{0}} \otimes {\overline{[{u}, {h}]}},
{\quad}
{\widetilde{g}} \otimes {\ket{0}} \mapsto 0,
{\quad}
{X}(h) \mapsto 0,
\label{eq: sec cohom desc Dphi-action}
\\
& {\qquad} {\qquad} {\qquad} {\qquad}
{\overline{h}}^{\perp} \otimes {\ket{0}} \mapsto 0,
{\quad}
{\phi}^{h} \mapsto ({u}, {h}),
{\quad}
{D}{\phi}_{h} \mapsto 0,
\notag
\end{align}
for ${h} \in {\mathfrak{h}}$, ${h}^{\perp} \in {\mathfrak{h}}^{\perp}$, ${g} \in {\mathfrak{g}}$.

\vskip 1mm

To prove the first identity in \eqref{eq: sec cohom desc Theta commute with rel action}, we can use the computations in \eqref{eq: sec cohom desc J-action} to verify that
\begin{align*}
({J}_{u})_{(0)} \circ {\Theta}_{12} = (-1)^{u} {\Theta}_{12} \circ ({J}_{u})_{(0)},
\end{align*}
on ${C}^{\partial}_{1} \otimes {C}^{\partial}_{2}$. Therefore, the same identity also holds on ${C}^{\partial}_{1} \otimes {C}^{\partial}_{2} \otimes {C}^{\partial}_{3}$. By the same argument, the operators $({D}{\phi}_{u})_{(0)}$ and ${\Theta}$ supercommute for every $u \in {\mathfrak{h}}$.
\end{proof}

\vskip 1mm

\begin{proposition}\label{prop: sec cohom desc rel SUSY H(M) to rel H(C3)}\rm
The relative supersymmetric semi-infinite complex ${C}^{{\frac{\infty}{2}} + {\bullet}}_{{\,\,} {\text{s, rel}}}({\mathfrak{h}}, {M})$ is homotopy equivalent to ${C}^{\partial}_{\text{3, rel}}$. Hence, we have:
\begin{align*}
{H}^{{\frac{\infty}{2}} + {\bullet}}_{{\,\,} {\text{s, rel}}}({\mathfrak{h}}, {M}) \cong {H}^{\bullet}({C}^{\partial}_{\text{3, rel}}, {d}_{\text{s}}).
\end{align*}
\end{proposition}

\begin{proof}
Note that the cochain map ${I}$ in \eqref{eq: sec cohom desc def of map I} restricts to a cochain map between the corresponding relative complexes. Moreover, using \eqref{eq: sec cohom desc J-action} and \eqref{eq: sec cohom desc Dphi-action}, we can show that the cochain map ${P}$ defined in \eqref{eq: sec cohom desc def of map P} also restricts to the relative complexes. Hence, the desired result follows from Lemma \ref{lem: sec cohom desc C_i closed to rel action and Theta commute with rel action} and Proposition \ref{prop: sec cohom desc homotopy equiv for non-rel}.
\end{proof}

\vskip 1mm

In the remainder of this section, we adopt the notation for semi-infinite cohomology used in Section 3.1 of \cite{CLNS25}. We now look for a (non-supersymmetric) semi-infinite complex that can be identified with the complex ${C}^{\partial}_{3}$. Consider the following ${\widehat{\mathfrak{h}}}$-module of level $-{\overline{\kappa}}_{\mathfrak{h}}$:
\begin{align*}
{M}_{0} 
:= {V}^{{\kappa} - {\frac{1}{2}}{\overline{\kappa}}_{\mathfrak{g}}}({\mathfrak{g}}) 
\otimes 
{\mathcal{F}}({\overline{\mathfrak{h}}}^{\perp}) 
\otimes 
{V}^{- {\kappa} - {\frac{1}{2}}{\overline{\kappa}}_{\mathfrak{h}}}({\mathfrak{h}}).
\end{align*}
Indeed, if we set ${\overline{\kappa}}_{{\mathfrak{h}}^{\perp}} := {\overline{\kappa}}_{\mathfrak{g}} - {\overline{\kappa}}_{\mathfrak{h}}$, then there exists a vertex algebra homomorphism ${\nu}_{\mathfrak{h}}$ from the affine vertex algebra ${V}^{{\frac{1}{2}}{\overline{\kappa}}_{{\mathfrak{h}}^{\perp}}}({\mathfrak{h}})$ to ${\mathcal{F}}({\overline{\mathfrak{h}}}^{\perp})$ given by
\begin{align*}
{\nu}_{\mathfrak{h}} : {h} \mapsto {\frac{1}{2{\kappa}}}{\,}{\sum}_{i \in {G} \setminus {H}} (-1)^{i}{\overline{u}}^{i}{\overline{[{u}_{i}, {h}]}},
\end{align*}
for $h \in {\mathfrak{h}}$, where ${\{} {u}_i {\}}_{i \in {G} \setminus {H}}$ and ${\{} {u}^{i} {\}}_{i \in {G} \setminus {H}}$ are parity-homogeneous dual bases of ${\mathfrak{h}}^{\perp}$ with respect to $({\,}{\cdot}{\,}, {\,}{\cdot}{\,})$. More precisely, we have
\begin{align*}
[{\nu}_{\mathfrak{h}}(h) {\,} {}_{\lambda} {\,} {\overline{h}}^{\perp}] 
= (-1)^{h}{\overline{[{h}, {h}^{\perp}]}},
{\quad}
[{\nu}_{\mathfrak{h}}({h}_{1}) {\,} {}_{\lambda} {\,} {\nu}_{\mathfrak{h}}({h}_{2})] 
= {\nu}_{\mathfrak{h}}([{h}_{1}, {h}_{2}]) + {\frac{1}{2}}{\overline{\kappa}}_{{\mathfrak{h}}^{\perp}}({h}_{1}, {h}_{2}){\lambda},
\end{align*}
for ${h}, {h}_{1}, {h}_{2} \in {\mathfrak{h}}$ and ${h}^{\perp} \in {\mathfrak{h}}^{\perp}$.

\vskip 1mm

Denote by ${C}^{{\frac{\infty}{2}} + {\bullet}}({\mathfrak{h}}, {M}_{0})$ the complex of the semi-infinite cohomology of ${\widehat{\mathfrak{h}}}$ with coefficients in ${M}_{0}$:
\begin{align*}
{C}^{{\frac{\infty}{2}} + {\bullet}}({\mathfrak{h}}, {M}_{0}) 
= {M}_{0} \otimes {\textstyle{\bigwedge}}^{{\frac{\infty}{2}} + {\bullet}}({\mathfrak{h}}).
\end{align*}
Then its differential is given by
\begin{align*}
{d} = ({Q}_{0})^{{C}^{{\frac{\infty}{2}} + {\bullet}}({\mathfrak{h}}, {M}_{0})}_{(0)},
\end{align*}
as a ${V}^{-{\overline{\kappa}}_{\mathfrak{h}}}({\mathfrak{h}}) \otimes {\textstyle{\bigwedge}}^{{\frac{\infty}{2}} + {\bullet}}({\mathfrak{h}})$-module action, where
\begin{align*}
{Q}_{0} = {\sum}_{i \in H} (-1)^{i}{u}_{i} \otimes {\phi}^{i} - {\frac{1}{2}}{\,}{\sum}_{j, k \in H} (-1)^{jk + j}{\ket{0}} \otimes {D}{\phi}_{[{u}_{j}, {u}_{k}]}({\phi}^{j}{\phi}^{k}).
\end{align*}
In particular, the differential ${d}$ acts on the generators as follows:
\begin{align}
{d}({g} \otimes {\ket{0}} \otimes {\ket{0}}) 
=& {\,} {\sum}_{i \in H} (-1)^{ig + i + g}([{u}_{i}, {g}] \otimes {\ket{0}} \otimes {\ket{0}}){\phi}^{i} 
\notag
\\
&+ {\sum}_{i \in H} (-1)^{g}({\kappa} - {\frac{1}{2}}{\overline{\kappa}}_{\mathfrak{g}})({u}_{i}, {g}){\partial}{\phi}^{i},
\notag
\\
{d}({\ket{0}} \otimes {\overline{h}}^{\perp} \otimes {\ket{0}}) 
=& {\,} {\sum}_{i \in H} (-1)^{(i + 1)({h}^{\perp} + 1)} ({\ket{0}} \otimes {\overline{[{u}_{i}, {h}^{\perp}]}} \otimes {\ket{0}}){\phi}^{i},
\notag
\\
{d}({\ket{0}} \otimes {\ket{0}} \otimes {h}) 
=& {\,} {\sum}_{i \in H} (-1)^{ih + i + h}({\ket{0}} \otimes {\ket{0}} \otimes [{u}_{i}, {h}]){\phi}^{i}
\notag
\\
&+ {\sum}_{i \in H} (-1)^{h}(- {\kappa} - {\frac{1}{2}}{\overline{\kappa}}_{\mathfrak{h}})({u}_{i}, {h}){\partial}{\phi}^{i},
\label{eq: sec cohom desc d on non-SUSY complex}
\\
{d}({\phi}^{r}) 
=& {\,} {\frac{1}{2}}{\,}{\sum}_{i \in H} (-1)^{i}{\phi}^{i}{\phi}^{[{u}_{i}, {u}^{r}]},
\notag
\\
{d}({D}{\phi}_{r}) 
=& {\,} {u}_{r} \otimes {\ket{0}} \otimes {\ket{0}} + {\ket{0}} \otimes {\big(}{\frac{1}{2{\kappa}}}{\,}{\sum}_{i \in {G} \setminus {H}} (-1)^{i}{\overline{u}}^{i}{\overline{[{u}_{i}, {u}_{r}]}}{\big)} \otimes {\ket{0}}
\notag
\\
&+ {\ket{0}} \otimes {\ket{0}} \otimes {u}_{r} + {\sum}_{i \in H} {\phi}^{i}{D}{\phi}_{[{u}_{i}, {u}_{r}]},
\notag
\end{align}
for ${g} \in {\mathfrak{g}}$, ${h}^{\perp} \in {\mathfrak{h}}^{\perp}$, ${h} \in {\mathfrak{h}}$, $r \in H$. Then, the following proposition is a direct consequence of \eqref{eq: sec cohom desc d_s on C3} and \eqref{eq: sec cohom desc d on non-SUSY complex}:

\vskip 1mm

\begin{proposition}\label{prop: sec cohom desc SUSY H(M) to non-SUSY H(M0)}\rm
The supersymmetric semi-infinite cohomology of ${\widehat{\mathfrak{h}}}$ with coefficients in ${M}$ is isomorphic to the semi-infinite cohomology of ${\widehat{\mathfrak{h}}}$ with coefficients in ${M}_{0}$, that is,
\begin{align*}
{H}^{{\frac{\infty}{2}} + {\bullet}}_{{\,\,} {\text{s}}}({\mathfrak{h}}, {M}) \cong {H}^{{\frac{\infty}{2}} + {\bullet}}({\mathfrak{h}}, {M}_{0}).
\end{align*}
\end{proposition}

\begin{proof}
Since Proposition \ref{prop: sec cohom desc SUSY H(M) to H(C3)} gives ${H}^{{\frac{\infty}{2}} + {\bullet}}_{{\,\,} {\text{s}}}({\mathfrak{h}}, {M}) \cong {H}^{\bullet}({C}^{\partial}_{3}, {d}_{\text{s}})$, it suffices to show that the complexes ${C}^{\partial}_{3}$ and ${C}^{{\frac{\infty}{2}} + {\bullet}}({\mathfrak{h}}, {M}_{0})$ are isomorphic. Consider the following vertex algebra isomorphism (see Lemma \ref{lem: sec cohom desc VA iso of C3}):
\begin{align*}
{\Psi} : {C}^{{\frac{\infty}{2}} + {\bullet}}({\mathfrak{h}}, {M}_{0}) \rightarrow {C}^{\partial}_{3},
\end{align*}
given by
\begin{gather*}
{g} \otimes {\ket{0}} \otimes {\ket{0}} \mapsto {\widetilde{g}} \otimes {\ket{0}},
{\quad}
{\ket{0}} \otimes {\overline{h}}^{\perp} \otimes {\ket{0}} \mapsto {\overline{h}}^{\perp} \otimes {\ket{0}},
{\quad}
{\ket{0}} \otimes {\ket{0}} \otimes {h} \mapsto {X}(h),
\\
{D}{\phi}_{i} \mapsto {D}{\phi}_{i},
{\quad}
{\phi}^{i} \mapsto {\phi}^{i},
\end{gather*}
for ${g} \in {\mathfrak{g}}$, ${h}^{\perp} \in {\mathfrak{h}}^{\perp}$, ${h} \in {\mathfrak{h}}$, $i \in H$. Then one can check that the following diagram commutes
\begin{equation}
\begin{tikzcd}
{C}^{{\frac{\infty}{2}} + {\bullet}}({\mathfrak{h}}, {M}_{0})
\arrow[r, "{d}"]
\arrow[d, swap, "{\Psi}"]
&
{C}^{{\frac{\infty}{2}} + {\bullet}}({\mathfrak{h}}, {M}_{0})
\arrow[d, "{\Psi}"]
\\
{C}^{\partial}_{3}
\arrow[r, swap, "{d}_{\text{s}}"]
&
{C}^{\partial}_{3}
\end{tikzcd}
\label{diagram: sec cohom desc Psi}
\end{equation}
by comparing \eqref{eq: sec cohom desc d_s on C3} and \eqref{eq: sec cohom desc d on non-SUSY complex}.
\end{proof}

\vskip 1mm

Now consider the vertex algebra homomorphism ${\widetilde{J}}_{\bigwedge} : {V}^{{\overline{\kappa}}_{\mathfrak{h}}}({\mathfrak{h}}) \rightarrow {\textstyle{\bigwedge}}^{{\frac{\infty}{2}} + {\bullet}}({\mathfrak{h}})$ given by
\begin{align*}
{\widetilde{J}}_{\bigwedge}({u}_{i}) = {\sum}_{j \in H} {\phi}^{j}{D}{\phi}_{[{u}_{j}, {u}_{i}]},
\end{align*}
for $i \in H$. This homomorphism is obtained by multiplying the image of ${u}_{i}$ under the vertex algebra homomorphism given in Section 3.1 of \cite{CLNS25} by $(-1)^{i}$. Moreover, let us introduce the following notation:
\begin{align*}
{\widetilde{J}}_{u} 
:=& {\,} {u} \otimes {\ket{0}} \otimes {\ket{0}} + {\ket{0}} \otimes {\nu}_{\mathfrak{h}}(u) \otimes {\ket{0}} 
\\
&+ {\ket{0}} \otimes {\ket{0}} \otimes {u} + {\widetilde{J}}_{\bigwedge}(u) \in {C}^{{\frac{\infty}{2}} + {\bullet}}({\mathfrak{h}}, {M}_{0}),
\end{align*}
for ${u} \in {\mathfrak{h}}$. Then the complex of the relative semi-infinite cohomology of ${\widehat{\mathfrak{h}}}$ with coefficients in ${M}_{0}$ is given by:
\begin{align*}
{C}^{{\frac{\infty}{2}} + {\bullet}}_{{\,\,} {\text{rel}}}({\mathfrak{h}}, {M}_{0}) = {\Big(}{\bigcap}_{i \in H}{\,}{\text{ker}}{\,}({\widetilde{J}}_{{u}_{i}})_{(0)}{\Big)} \cap {\Big(}{\bigcap}_{i \in H}{\,}{\text{ker}}{\,}({D}{\phi}_{i})_{(0)}{\Big)}.
\end{align*}
Moreover, a direct computation shows that, for ${u} \in {\mathfrak{h}}$, the operator $({\widetilde{J}}_{u})_{(0)}$ acts on the generators as follows:
\begin{align}
&{g} \otimes {\ket{0}} \otimes {\ket{0}} \mapsto {[u, g]} \otimes {\ket{0}} \otimes {\ket{0}},
{\quad}
{\ket{0}} \otimes {\overline{h}}^{\perp} \otimes {\ket{0}} \mapsto {\ket{0}} \otimes (-1)^{u}{\overline{[{u}, {h}^{\perp}]}} \otimes {\ket{0}},
\notag
\\
& {\quad}
{\ket{0}} \otimes {\ket{0}} \otimes {h} \mapsto {\ket{0}} \otimes {\ket{0}} \otimes {[u, h]},
{\quad}
{\phi}^{h} \mapsto {\phi}^{[u, h]},
{\quad}
{D}{\phi}_{h} \mapsto (-1)^{u}{D}{\phi}_{[u, h]},
\label{eq: sec cohom desc tildeJ action on non-SUSY complex}
\end{align}
and the operator $({D}{\phi}_{u})_{(0)}$ acts on the generators as follows:
\begin{align}
&{g} \otimes {\ket{0}} \otimes {\ket{0}} \mapsto 0,
{\quad}
{\ket{0}} \otimes {\overline{h}}^{\perp} \otimes {\ket{0}} \mapsto 0,
{\quad}
{\ket{0}} \otimes {\ket{0}} \otimes {h} \mapsto 0,
\notag
\\
& {\qquad} {\qquad} {\qquad} {\qquad}
{\phi}^{h} \mapsto ({u}, {h}),
{\quad}
{D}{\phi}_{h} \mapsto 0,
\label{eq: sec cohom desc Dphi action on non-SUSY complex}
\end{align}
for ${h} \in {\mathfrak{h}}$, ${h}^{\perp} \in {\mathfrak{h}}^{\perp}$,
${g} \in {\mathfrak{g}}$. Then we have the following proposition, which is an analogue of Proposition \ref{prop: sec cohom desc SUSY H(M) to non-SUSY H(M0)} for relative semi-infinite cohomology.

\vskip 1mm

\begin{proposition}\label{prop: sec cohom desc rel SUSY H(M) to rel non-SUSY H(M0)}\rm
The relative supersymmetric semi-infinite cohomology of ${\widehat{\mathfrak{h}}}$ with coefficients in ${M}$ is isomorphic to the relative semi-infinite cohomology of ${\widehat{\mathfrak{h}}}$ with coefficients in ${M}_{0}$, that is,
\begin{align*}
{H}^{{\frac{\infty}{2}} + {\bullet}}_{{\,\,} {\text{s, rel}}}({\mathfrak{h}}, {M}) \cong {H}^{{\frac{\infty}{2}} + {\bullet}}_{{\,\,} {\text{rel}}}({\mathfrak{h}}, {M}_{0}).
\end{align*}
\end{proposition}

\begin{proof}
From \eqref{eq: sec cohom desc J-action}, \eqref{eq: sec cohom desc Dphi-action}, \eqref{eq: sec cohom desc tildeJ action on non-SUSY complex}, and \eqref{eq: sec cohom desc Dphi action on non-SUSY complex}, we have:
\begin{align*}
{\Psi} \circ ({\widetilde{J}}_{{u}_{i}})_{(0)} = ({J}_{{u}_{i}})_{(0)} \circ {\Psi},
{\quad}
{\Psi} \circ ({D}{\phi}_{i})_{(0)} = ({D}{\phi}_{i})_{(0)} \circ {\Psi},
\end{align*}
for all $i \in H$. Hence, the result follows from \eqref{diagram: sec cohom desc Psi} and Proposition \ref{prop: sec cohom desc rel SUSY H(M) to rel H(C3)}.
\end{proof}

\vskip 1mm

Then we obtain the following theorem, which is the main result of this section:

\vskip 1mm

\begin{theorem}\label{theorem: sec cohom desc cohom and supercoset}\rm
Let ${\mathfrak{h}}$ be a direct sum of a reductive Lie algebra and Lie superalgebras of type ${\mathfrak{osp}}_{1|2m}$. Suppose that the restriction of ${\kappa}_{\mathfrak{h}} - {\frac{1}{2}}{\overline{\kappa}}_{\mathfrak{h}}$ to each nonabelian simple summand of ${\mathfrak{h}}$ is an irrational level. Then we have
\begin{align*}
{H}^{{\frac{\infty}{2}} + {n}}_{{\,\,} {\text{s, rel}}}({\mathfrak{h}}, {M}) \cong {\delta}_{n,0}{\,}{\text{Com}}({V}^{\kappa}_{{\,\,} {\text{s}}}({\mathfrak{h}}), {V}^{\kappa}_{{\,\,} {\text{s}}}({\mathfrak{g}})),
\end{align*}
for all $n \in {\mathbb{Z}}$.
\end{theorem}

\begin{proof}
We know that ${V}^{{\kappa} - {\frac{1}{2}}{\overline{\kappa}}_{\mathfrak{g}}}({\mathfrak{g}}) \otimes {\mathcal{F}}({\overline{\mathfrak{h}}}^{\perp})$ belongs to the Kazhdan--Lusztig category of ${\widehat{\mathfrak{h}}}$-modules with level ${\kappa} - {\frac{1}{2}}{\overline{\kappa}}_{\mathfrak{h}}$, and hence the coefficient module ${M}_{0}$ belongs to the Kazhdan--Lusztig category of ${\widehat{\mathfrak{h}}}$-modules with level $-{\overline{\kappa}}_{\mathfrak{h}}$. Therefore, we have:
\begin{align*}
{H}^{{\frac{\infty}{2}} + {n}}_{{\,\,} {\text{s, rel}}}({\mathfrak{h}}, {M}) 
&\cong {H}^{{\frac{\infty}{2}} + {n}}_{{\,\,} {\text{rel}}}({\mathfrak{h}}, {M}_{0})
\\
&\cong {\delta}_{n,0}{\,}{\text{Com}}({V}^{{\kappa} - {\frac{1}{2}}{\overline{\kappa}}_{\mathfrak{h}}}({\mathfrak{h}}), {V}^{{\kappa} - {\frac{1}{2}}{\overline{\kappa}}_{\mathfrak{g}}}({\mathfrak{g}}) \otimes {\mathcal{F}}({\overline{\mathfrak{h}}}^{\perp})),
\end{align*}
where the first isomorphism follows from Proposition \ref{prop: sec cohom desc rel SUSY H(M) to rel non-SUSY H(M0)}, and the second follows from Theorem 3.3 of \cite{CLNS25} (see also \cites{FGZ86, CFL22, CGNS22}). Since we now know that
\begin{align*}
{V}^{\kappa}_{{\,\,} {\text{s}}}({\mathfrak{h}}) \cong {V}^{{\kappa} - {\frac{1}{2}}{\overline{\kappa}}_{\mathfrak{h}}}({\mathfrak{h}}) \otimes {\mathcal{F}}({\overline{\mathfrak{h}}}),
{\quad}
{V}^{\kappa}_{{\,\,} {\text{s}}}({\mathfrak{g}}) \cong {V}^{{\kappa} - {\frac{1}{2}}{\overline{\kappa}}_{\mathfrak{g}}}({\mathfrak{g}}) \otimes {\mathcal{F}}({\overline{\mathfrak{h}}}^{\perp}) \otimes {\mathcal{F}}({\overline{\mathfrak{h}}}),
\end{align*}
we obtain the following isomorphism:
\begin{align*}
{\text{Com}}({V}^{{\kappa} - {\frac{1}{2}}{\overline{\kappa}}_{\mathfrak{h}}}({\mathfrak{h}}), {V}^{{\kappa} - {\frac{1}{2}}{\overline{\kappa}}_{\mathfrak{g}}}({\mathfrak{g}}) \otimes {\mathcal{F}}({\overline{\mathfrak{h}}}^{\perp})) 
\cong {\text{Com}}({V}^{\kappa}_{{\,\,} {\text{s}}}({\mathfrak{h}}), {V}^{\kappa}_{{\,\,} {\text{s}}}({\mathfrak{g}})),
\end{align*}
and hence the theorem follows.
\end{proof}

\vskip 1mm

Combining Theorem \ref{theorem: N=3/2 of saff}, Corollary \ref{corollary: sec N=2 of SUSY semi inf N=2 on rel SUSY cohom}, and Theorem \ref{theorem: sec cohom desc cohom and supercoset}, we obtain the following corollary.

\vskip 1mm

\begin{corollary}\label{corollary: sec cohom desc lift of N=2 in scoset to swzw}\rm
Under Assumption \ref{assum: second assumption}, suppose that ${\mathfrak{h}}$ is a direct sum of a reductive Lie algebra and Lie superalgebras of type ${\mathfrak{osp}}_{1|2m}$, and that the restriction of ${\kappa}_{\mathfrak{h}} - {\frac{1}{2}}{\overline{\kappa}}_{\mathfrak{h}}$ to each nonabelian simple summand of ${\mathfrak{h}}$ is an irrational level. Then the $N = {\frac{3}{2}}$ superconformal structure on the product of super-affine vertex algebras
\begin{align*}
{V}^{\kappa}_{{\,\,} {\text{s}}}({\mathfrak{g}}) \otimes {V}^{-{\kappa}}_{{\,\,} {\text{s}}}({\mathfrak{h}})
\end{align*}
induces an $N = 2$ superconformal structure on ${\text{Com}}({V}^{\kappa}_{{\,\,} {\text{s}}}({\mathfrak{h}}), {V}^{\kappa}_{{\,\,} {\text{s}}}({\mathfrak{g}}))$ with central charge ${\frac{3}{2}}{\,}{\text{sdim}}{\,}{\mathfrak{g}} - {\frac{3}{2}}{\,}{\text{sdim}}{\,}{\mathfrak{h}} + {3}{\,}{\text{str}}({\mathcal{D}}_{{\mathfrak{g}} \oplus {\mathfrak{h}}})$.
\end{corollary}

\newpage
\begin{appendices}
\section{Proof of Proposition \ref{prop: equiv def of N=3/2 sconf}}\label{section: Pf of Prop in Sec N=2 and N=3/2 sconf VA}
\setcounter{equation}{0}
\vskip 2mm

\begin{proof}
We begin by computing the OPEs of ${G}_{(0)}{U}$ and ${U}_{(0)}({G}_{(0)}{U})$ with the conformal vector. From the relation
\begin{align}
\label{eqa: G and U}
[{G} {\,} {}_{\lambda} {\,} {U}] = {G}_{(0)}{U},
\end{align}
and the fact that ${G}$ and ${U}$ are primary vectors of conformal weights ${\frac{3}{2}}$ and $1$, respectively, we obtain
\begin{align}
[{T} {\,} {}_{\lambda} {\,} {G}_{(0)}{U}] 
&= [{T} {\,} {}_{\lambda} {\,} [{G} {\,} {}_{\gamma} {\,} {U}]] 
= [[{T} {\,} {}_{\lambda} {\,} {G}] {\,} {}_{{\lambda} + {\gamma}} {\,} {U}] + [{G} {\,} {}_{\gamma} {\,} [{T} {\,} {}_{\lambda} {\,} {U}]]
\notag
\\
&= [({\partial} + {\frac{3}{2}}{\lambda}){G} {\,} {}_{{\lambda} + {\gamma}} {\,} {U}] + [{G} {\,} {}_{\gamma} {\,} ({\partial} + {\lambda}){U}]
\label{eqa: T and GU}
\\
&= ({\frac{1}{2}}{\lambda} - {\gamma})[{G} {\,} {}_{{\lambda} + {\gamma}} {\,} {U}] + ({\partial} + {\lambda} + {\gamma})[{G} {\,} {}_{\gamma} {\,} {U}] 
= ({\partial} + {\frac{3}{2}}{\lambda}){G}_{(0)}{U}.
\notag
\end{align}
Using the equality \eqref{eqa: T and GU} together with the following relation
\begin{align}
\label{eqa: U and GU}
[{U} {\,} {}_{\lambda} {\,} {G}_{(0)}{U}] = {U}_{(0)}({G}_{(0)}{U}),
\end{align}
we further obtain
\begin{align*}
[{T} {\,} {}_{\lambda} {\,} {U}_{(0)}({G}_{(0)}{U})] 
&= [{T} {\,} {}_{\lambda} {\,} [{U} {\,} {}_{\gamma} {\,} {G}_{(0)}{U}]] 
= [[{T} {\,} {}_{\lambda} {\,} {U}] {\,} {}_{{\lambda} + {\gamma}} {\,} {G}_{(0)}{U}] + [{U} {\,} {}_{\gamma} {\,} [{T} {\,} {}_{\lambda} {\,} {G}_{(0)}{U}]]
\\
&= [({\partial} + {\lambda}){U} {\,} {}_{{\lambda} + {\gamma}} {\,} {G}_{(0)}{U}] + [{U} {\,} {}_{\gamma} {\,} ({\partial} + {\frac{3}{2}}{\lambda}){G}_{(0)}{U}]
\\
&= -{\gamma}[{U} {\,} {}_{{\lambda} + {\gamma}} {\,} {G}_{(0)}{U}] + ({\partial} + {\frac{3}{2}}{\lambda} +{\gamma})[{U} {\,} {}_{\gamma} {\,} {G}_{(0)}{U}]
\\
&= ({\partial} + {\frac{3}{2}}{\lambda}){U}_{(0)}({G}_{(0)}{U}).
\end{align*}
Hence, ${G}^{\pm}$ and ${\Phi}$ are all odd primary vectors of conformal weight ${\frac{3}{2}}$, defined by
\begin{align*}
{\Phi} := - {\frac{1}{8}}({G} + {U}_{(0)}({G}_{(0)}{U})),
{\quad}
{G}^{+} := {\frac{1}{2}}({G} - {G}_{(0)}{U}) + {\Phi},
{\quad}
{G}^{-} := {\frac{1}{2}}({G} + {G}_{(0)}{U}) - {2}{\Phi}.
\end{align*}

\vskip 1mm

Next, we compute the OPEs of ${G}^{\pm}$ with ${U}$. Using equations \eqref{eqa: G and U}, \eqref{eqa: U and GU}, together with
\begin{align}
\label{eqa: U and Phi}
[{U} {\,} {}_{\lambda} {\,} {\Phi}] = -3{\Phi},
\end{align}
we have:
\begin{align*}
[{U} {\,} {}_{\lambda} {\,} {G}^{+}] 
&= {\frac{1}{2}}[{U} {\,} {}_{\lambda} {\,} {G}] - {\frac{1}{2}}[{U} {\,} {}_{\lambda} {\,} {G}_{(0)}{U}] + [{U} {\,} {}_{\lambda} {\,} {\Phi}]
\\
&= - {\frac{1}{2}}{G}_{(0)}{U} - {\frac{1}{2}}{U}_{(0)}({G}_{(0)}{U}) - {3}{\Phi}
\\
&= - {\frac{1}{2}}{G}_{(0)}{U} - {\frac{1}{2}}(- {8}{\Phi} - {G})- {3}{\Phi} 
= {G}^{+},
\end{align*}
since ${U}_{(0)}({G}_{(0)}{U}) = - {8}{\Phi} - {G}$ by definition of ${\Phi}$. Similarly, using \eqref{eqa: G and U}, \eqref{eqa: U and GU}, and \eqref{eqa: U and Phi}, we also obtain
\begin{align}
\label{eqa: U and G-}
[{U} {\,} {}_{\lambda} {\,} {G}^{-}] = - {G}^{-}.
\end{align}

\vskip 1mm

On the other hand, using the Jacobi identity of vertex algebras together with \eqref{eqa: G and U}, we have
\begin{align}
\label{eqa: G and GU}
[{G} {\,} {}_{\lambda} {\,} {G}_{(0)}{U}] = ({\partial} + {2}{\lambda}){U}.
\end{align}
Hence, using \eqref{eqa: G and U}, \eqref{eqa: U and GU}, \eqref{eqa: G and GU}, and the following OPEs:
\begin{align*}
[{U} {\,} {}_{\lambda} {\,} {U}] = {\frac{c}{3}}{\lambda},
\end{align*}
\begin{align}
\label{eqa: G and Phi}
[{G} {\,} {}_{\lambda} {\,} {\Phi}] = {G}_{(0)}{\Phi},
\end{align}
we obtain
\begin{align}
[{G}_{(0)}{U} {\,} {}_{\lambda} {\,} {G}_{(0)}{U}] 
&= [{G}_{(0)}{U} {\,} {}_{\lambda} {\,} [{G} {\,} {}_{\gamma} {\,} {U}]] 
= [[{G}_{(0)}{U} {\,} {}_{\lambda} {\,} {G}] {\,} {}_{{\lambda} + {\gamma}} {\,} {U}] - [{G} {\,} {}_{\gamma} {\,} [{G}_{(0)}{U} {\,} {}_{\lambda} {\,} {U}]]
\notag
\\
&= ({\gamma} - {\lambda})[{U} {\,} {}_{{\lambda} + {\gamma}} {\,} {U}] - [{G} {\,} {}_{\gamma} {\,} {G}] - {8}[{G} {\,} {}_{\gamma} {\,} {\Phi}]
\label{eqa: GU and GU}
\\
&= {\frac{c}{3}}({\gamma}^{2} - {\lambda}^{2}) - ({2}{T} + {\frac{c}{3}}{\gamma}^{2}) - {8}{G}_{(0)}{\Phi} 
= - ({2}{T} + {\frac{c}{3}}{\lambda}^{2}) - {8}{F},
\notag
\end{align}
where ${F} := {G}_{(0)}{\Phi}$.

\vskip 1mm

Now, using \eqref{eqa: G and GU}, \eqref{eqa: G and Phi}, \eqref{eqa: GU and GU}, together with the following OPEs:
\begin{align*}
[{G}_{(0)}{U} {\,} {}_{\lambda} {\,} {\Phi}] = - {F},
{\quad}
[{\Phi} {\,} {}_{\lambda} {\,} {\Phi}] = 0,
\end{align*}
we have:
\begin{align*}
[{G}^{+} {\,} {}_{\lambda} {\,} {G}^{+}] 
= [{\frac{1}{2}}{G} - {\frac{1}{2}}{G}_{(0)}{U} + {\Phi} {\,} {}_{\lambda} {\,} {\frac{1}{2}}{G} - {\frac{1}{2}}{G}_{(0)}{U} + {\Phi}] = 0,
\end{align*}
by direct expansion. Similarly, we also have the following OPEs
\begin{align}
[{G}^{+} {\,} {}_{\lambda} {\,} {G}^{-}] &= {T} + {\frac{1}{2}}{\partial}{U} + {\lambda}{U} + {\frac{c}{6}}{\lambda}^{2},
\label{eqa: G+ and G-}
\\
[{G}^{+} {\,} {}_{\lambda} {\,} {\Phi}] &= F,
\label{eqa: G+ and Phi}
\\
[{G}^{-} {\,} {}_{\lambda} {\,} {G}^{-}] &= - {2}{F},
\label{eqa: G- and G-}
\\
[{G}^{-} {\,} {}_{\lambda} {\,} {\Phi}] &= 0.
\label{eqa: G- and Phi}
\end{align}

\vskip 1mm

To complete the proof, we will derive the remaining OPEs of the $N = {\frac{3}{2}}$ superconformal algebra from the OPEs obtained so far. First, using \eqref{eqa: G- and G-}, the OPE of ${F}$ with the conformal vector ${T}$ is given by
\begin{align*}
[{T} {\,} {}_{\lambda} {\,} {F}] 
&= - {\frac{1}{2}}[{T} {\,} {}_{\lambda} {\,} [{G}^{-} {\,} {}_{\gamma} {\,} {G}^{-}]] 
= - {\frac{1}{2}}[[{T} {\,} {}_{\lambda} {\,} {G}^{-}] {\,} {}_{{\lambda} + {\gamma}} {\,} {G}^{-}] - {\frac{1}{2}}[{G}^{-} {\,} {}_{\gamma} {\,} [{T} {\,} {}_{\lambda} {\,} {G}^{-}]]
\\
&= - {\frac{1}{2}}[({\partial} + {\frac{3}{2}}{\lambda}){G}^{-} {\,} {}_{{\lambda} + {\gamma}} {\,} {G}^{-}] - {\frac{1}{2}}[{G}^{-} {\,} {}_{\gamma} {\,} ({\partial} + {\frac{3}{2}}{\lambda}){G}^{-}] 
= ({\partial} + {2}{\lambda}){F}.
\end{align*}
The OPE of ${F}$ with ${G}^{+}$ is given by
\begin{align}
[{G}^{+} {\,} {}_{\lambda} {\,} {F}] 
=& {\,} - {\frac{1}{2}}[{G}^{+} {\,} {}_{\lambda} {\,} [{G}^{-} {\,} {}_{\gamma} {\,} {G}^{-}]] 
= - {\frac{1}{2}}[[{G}^{+} {\,} {}_{\lambda} {\,} {G}^{-}] {\,} {}_{{\lambda} + {\gamma}} {\,} {G}^{-}] + {\frac{1}{2}}[{G}^{-} {\,} {}_{\gamma} {\,} [{G}^{+} {\,} {}_{\lambda} {\,} {G}^{-}]]
\notag
\\
=& {\,} - {\frac{1}{2}}[{T} {\,} {}_{{\lambda} + {\gamma}} {\,} {G}^{-}] - {\frac{1}{2}}{\big(}{\frac{1}{2}}{\lambda} - {\frac{1}{2}}{\gamma}{\big)}[{U} {\,} {}_{{\lambda} + {\gamma}} {\,} {G}^{-}]
\label{eqa: G+ and F}
\\
&+ {\frac{1}{2}}[{G}^{-} {\,} {}_{\gamma} {\,} {T}] + {\frac{1}{2}}{\big(}{\frac{1}{2}}{\partial} + {\lambda} + {\frac{1}{2}}{\gamma}{\big)}[{G}^{-} {\,} {}_{\gamma} {\,} {U}] = 0,
\notag
\end{align}
which follows from \eqref{eqa: U and G-}, \eqref{eqa: G+ and G-}, and \eqref{eqa: G- and G-}. To compute the OPE of ${F}$ with ${G}^{-}$, we can use \eqref{eqa: U and Phi}, \eqref{eqa: G+ and G-}, \eqref{eqa: G+ and Phi}, and \eqref{eqa: G- and Phi}, from which we obtain:
\begin{align*}
[{G}^{-} {\,} {}_{\lambda} {\,} {F}] 
&= [{G}^{-} {\,} {}_{\lambda} {\,} [{G}^{+} {\,} {}_{\gamma} {\,} {\Phi}]] 
= [[{G}^{-} {\,} {}_{\lambda} {\,} {G}^{+}] {\,} {}_{{\lambda} + {\gamma}} {\,} {\Phi}] - [{G}^{+} {\,} {}_{\gamma} {\,} [{G}^{-} {\,} {}_{\lambda} {\,} {\Phi}]]
\\
&= [{T} - {\frac{1}{2}}{\partial}{U} - {\lambda}{U} {\,} {}_{{\lambda} + {\gamma}} {\,} {\Phi}] 
= [{T} {\,} {}_{{\lambda} + {\gamma}} {\,} {\Phi}] - {\frac{1}{2}}({\lambda} - {\gamma})[{U} {\,} {}_{{\lambda} + {\gamma}} {\,} {\Phi}] 
= ({\partial} + {3}{\lambda}){\Phi}.
\end{align*}

\vskip 1mm

We now compute the OPEs of ${F}$ with ${U}$ and ${\Phi}$. Using \eqref{eqa: U and G-}, \eqref{eqa: G- and G-}, and \eqref{eqa: G- and G-}, \eqref{eqa: G- and Phi}, respectively, we obtain:
\begin{align*}
[{U} {\,} {}_{\lambda} {\,} {F}] 
&= -{\frac{1}{2}}[{U} {\,} {}_{\lambda} {\,} [{G}^{-} {\,} {}_{\gamma} {\,} {G}^{-}]] 
= -{\frac{1}{2}}[[{U} {\,} {}_{\lambda} {\,} {G}^{-}] {\,} {}_{{\lambda} + {\gamma}} {\,} {G}^{-}] - {\frac{1}{2}}[{G}^{-} {\,} {}_{\gamma} {\,} [{U} {\,} {}_{\lambda} {\,} {G}^{-}]]
\\
&= {\frac{1}{2}}[{G}^{-} {\,} {}_{{\lambda} + {\gamma}} {\,} {G}^{-}] + {\frac{1}{2}}[{G}^{-} {\,} {}_{\gamma} {\,} {G}^{-}] = - {2}{F},
\end{align*}
and
\begin{align}
\begin{split}
[{\Phi} {\,} {}_{\lambda} {\,} {F}] 
&= - {\frac{1}{2}}[{\Phi} {\,} {}_{\lambda} {\,} [{G}^{-} {\,} {}_{\gamma} {\,} {G}^{-}]] 
\\
&= -{\frac{1}{2}}[[{\Phi} {\,} {}_{\lambda} {\,} {G}^{-}] {\,} {}_{{\lambda} + {\gamma}} {\,} {G}^{-}] + {\frac{1}{2}}[{G}^{-} {\,} {}_{\gamma} {\,} [{\Phi} {\,} {}_{\lambda} {\,} {G}^{-}]] = 0.
\end{split}
\label{eqa: Phi and F}
\end{align}
Finally, the OPE of ${F}$ with itself is given by
\begin{align*}
[{F} {\,} {}_{\lambda} {\,} {F}] 
= [{F} {\,} {}_{\lambda} {\,} [{G}^{+} {\,} {}_{\gamma} {\,} {\Phi}]] 
= [[{F} {\,} {}_{\lambda} {\,} {G}^{+}] {\,} {}_{{\lambda} + {\gamma}} {\,} {\Phi}] + [{G}^{+} {\,} {}_{\gamma} {\,} [{F} {\,} {}_{\lambda} {\,} {\Phi}]] = 0,
\end{align*}
which follows from \eqref{eqa: G+ and Phi}, \eqref{eqa: G+ and F}, and \eqref{eqa: Phi and F}. Therefore, we have completed the proof of Proposition \ref{prop: equiv def of N=3/2 sconf}.
\end{proof}
\end{appendices}


\end{document}